\documentclass[a4paper,12pt]{article}
\usepackage{amsmath}
\usepackage{amsbsy}
\usepackage{amsthm}
\usepackage{amsfonts}
\usepackage{amssymb}
\usepackage{tikz}
\usepackage{pgfplots}
\usepackage{subfig}
\usepackage{epstopdf}
\usetikzlibrary{fit}
\usetikzlibrary{shapes.geometric}
\usetikzlibrary{decorations.pathmorphing}

\newtheorem{proposition}{Proposition}
\newtheorem{theorem}{Theorem}[section]

\numberwithin{equation}{section}

\let\a=\alpha \let\b=\beta    \let\g=\gamma     \let\d=\delta     \let\e=\varepsilon
\let\z=\zeta       \let\th=\theta \let\k=\kappa     \let\l=\lambda
    \let\n=\nu      \let\x=\xi        \let\p=\pi        \let\r=\rho
\let\s=\sigma \let\t=\tau     \let\f=\varphi       \let\c=\chi
   \let\o=\omega     
\let\G=\Gamma \let\D=\Delta       \let\L=\Lambda    
             
\let\O=\Omega \let\Y=\Upsilon

\def\ss{{\underline\s}}

\def\HHH{{\cal H}}\def\VV{{\cal V}}\def\RR{{\cal R}}\def\LL{{\cal L}}\def\NN{{\cal N}}\def\BB{{\cal B}}\def\DD{{\cal D}}
\def\CC{{\cal C}}\def\GG{{\cal G}}\def\PP{{\cal P}}\def\EE{{\cal E}}\def\TT{{\cal T}}\def\SS{{\cal S}}

\def\RRR{{\mathbb R}}\def\ZZZ{{\mathbb Z}}\def\CCC{{\mathbb C}}

\let\dpr=\partial
\let\io=\infty
\let\==\equiv

\def\media#1{{\left\langle#1\right\rangle}}

\def\wt#1{\widetilde{#1}}
\def\diam{{\rm diam}}

\def\lis{\overline}

\def\V#1{{\bf #1}}
\def\xx{{\bf x}}
\def\yy{{\bf y}}\def\kk{{\bf k}}\def\zz{{\bf z}}\def\nn{{\bf n}}\def\uu{{\bf u}}\def\dd{{\bf d}}\def\pp{{\bf p}}

\def\bsA{{\boldsymbol A}}

\def\be{\begin{equation}}
\def\ee{\end{equation}}
\def\bea{\begin{eqnarray}}\def\eea{\end{eqnarray}}

\def\pref#1{(\ref{#1})}

\DeclareMathOperator{\Pf}{Pf}

\newcommand{\piecewise}[1]{\left\{\begin{array}{ll} #1 \end{array}\right.}
\DeclareMathOperator{\black}{bl}
\DeclareMathOperator{\supp}{supp}

\def\bml{\begin{multline}}
\def\eml{\end{multline}}

\tikzset{vertex/.style={circle,fill=black,inner sep=2pt},
bigvertex/.style={circle,fill=black,inner sep=4pt},
specialEP/.style={rectangle,fill=white,draw,inner sep=3pt},
nuEP/.style={circle,fill=white,draw, inner sep=2pt},
linelabel/.style={sloped,above,very near start, inner sep=1pt,execute at begin node=$\scriptstyle,execute at end node=$},
baseline=(current  bounding  box.center),doubled/.style={double distance= 1pt,line width=1.5pt}
}
\pgfdeclarelayer{background}
\pgfsetlayers{background,main}

\begin{document}
\title{The scaling limit of the energy correlations in non integrable Ising models}

\author{Alessandro Giuliani\\
\small  Universit\`a di Roma Tre,\\
\small L.go S. L. Murialdo 1, 00146 Roma - Italy
\and Rafael L. Greenblatt\\
\small  Universit\`a di Roma Tre,\\
\small L.go S. L. Murialdo 1, 00146 Roma - Italy
\and
 Vieri Mastropietro\\
 \small Universit\`a di Roma Tor Vergata,\\
 \small V.le della Ricerca
Scientifica, 00133 Roma - Italy }

\maketitle
\begin{center}{\it Dedicated to Elliott Lieb on the occasion of his 80th birthday}\end{center}

\begin{abstract}
We obtain an explicit expression for the multipoint energy correlations of a
non solvable two-dimensional Ising models with nearest
neighbor ferromagnetic interactions plus a weak finite range
interaction of strength $\l$, in a scaling limit in which we send the lattice
spacing to zero and the temperature to the critical one. Our
analysis is based on an exact mapping of the model into an interacting lattice
fermionic theory, which generalizes the one originally used by Schultz, Mattis and Lieb for the
nearest neighbor Ising model. The interacting model is then analyzed by a multiscale method first 
proposed by Pinson and Spencer.
If the lattice spacing is finite,
then the correlations cannot be computed in closed form: rather, they 
are expressed in terms of infinite, convergent, power series in $\l$.
In the scaling limit, these infinite expansions radically simplify
and reduce to the limiting energy correlations of the integrable Ising model, up to a finite renormalization 
of the parameters. Explicit bounds on the speed of convergence to the scaling limit 
are derived.
\end{abstract}

\tableofcontents

\section{Introduction}\label{sec1}

\subsection{The model and the main results}
The two-dimensional (2D) Ising model is an oversimplified
description of a planar magnet in which the dipoles can point only
in two directions and only the interactions among neighboring spins are
considered. Remarkably, it can be solved explicitly, as first
shown by Onsager \cite{On44}: the free energy, the magnetization and several correlation 
functions in the absence of an external field
can be computed in closed form \cite{Ka49, KO49, Ya52,MW}. 

From a physical point of view, there is no special reason to consider only nearest neighbor (n.n.) interactions:
it is natural (and more realistic) to consider generalizations of the n.n. Ising model where the spins
interact via a finite range  interaction of the form:
\be H_M=-J\sum_{\xx\in
a\mathbb{Z}_{M}^{2}}\,\sum_{j=1,2}\sigma_{\xx}\s_{\xx+a\hat{\bf
e}_j}-\l\sum_{\{\xx,\yy\}}\s_\xx v(\xx-\yy)\s_\yy\equiv H^{(0)}_M +\l W_M\;,\label{2.1}\ee
where $J$ is a positive constant, $\mathbb{Z}_{M}^{2}\subset \ZZZ^2$
is a finite square box of side $M$ with periodic boundary
conditions, $a$ is the lattice spacing, $\s_\xx\in\{\pm1\}$, and
$\hat{\bf e}_j$ are the two unit coordinate vectors on $\ZZZ^2$.
The sum in the second term of Eq.\pref{2.1}  is over all unordered
pairs of sites in $\L^a_L:=a{\mathbb Z}^2_M$; the interaction
potential $v(\xx-\yy)$ is rotation invariant and has finite range,
namely: $v(\xx-\yy)=0$, $\forall |\xx|>R_0:=a M_0$, for a suitable
positive integer $M_0$. The standard n.n.\ case corresponds to $\l=0$.

Despite the fact that the models with Hamiltonian $H^{(0)}_M$ or
$H_M$ look ``physically equivalent", Onsager's exact solution is crucially based 
on the assumption that $\l=0$: therefore, it is natural to ask how much of it survives the presence of the
perturbation $\l W_M$.

The thermodynamical properties of the system are obtained by 
averaging with respect to the Gibbs measure at inverse temperature $\b$: given an 
observable $F(\ss)$ (here $\ss=\{\s_\xx\}_{\xx\in\L^a_L}$), its statistical average 
is defined as
\be \media{F}_{\b,L} = \frac{\sum_{\ss\in\O_M}e^{-\b H_M(\ss)} F(\ss)}{\sum_{\ss\in\O_M} e^{-\b H_M(\ss)}} \ee
where $\O_M=\{\pm1\}^{\L^a_L}$ is the spin configuration space. In particular, the averages $\media{\s_{\xx_1}\cdots\s_{\xx_n}}_{\b,L}$ are the 
multipoint spin correlation functions. Taking $L\to\io$ at $\{\xx_1,\ldots,
\xx_n\}$ fixed produces the infinite volume correlation functions:
\be \media{\s_{\xx_1}\cdots\s_{\xx_n}}_{\b}=\lim_{L\to\infty}\media{\s_{\xx_1}\cdots\s_{\xx_n}}_{\b,L}\ee
the collection of which characterizes the infinite volume Gibbs state. Two key thermodynamic quantities are the
free energy and the specific heat, defined as follows:
the free energy per site in the thermodynamic limit is 
\be f_\b=-\lim_{M\to\io}\frac1{\b M^2}\log Z\;,\ee
where $Z=\sum_{\ss\in\O_M}\exp\{-\b H_M(\ss)\}$
is the partition function. The specific heat is
 $C_v=-\b^2\frac{\partial^2(\b f_{\b})}{ \partial \b^2}$,
and can be obtained by summing over the energy density correlation functions: if we define the 
energy density operator as $\e_{\xx,j}:=a^{-1}\sigma_{\xx}\s_{\xx+a\hat{\bf e}_j}$,
then 
\be C_v=(\b J)^2\lim_{L\to\infty}\frac1{L^2}\sum_{j,j'}\int\limits_{\L^a_L\times\L^a_L}\!\!d\xx\, d\yy 
\media{\e_{\xx,j};\e_{\yy,j'}}^T_{\b,L}\;,\ee
where $\int_{\L^a_L} d\xx$ is a shorthand for $a^2\sum_{\xx\in\L^a_L}$ and $\media{\cdot;\cdot}^T_{\b,L}$
indicates the truncated expectation: $\media{\e_{\xx,j};\e_{\yy,j'}}^T_{\b,L}=
\media{\e_{\xx,j}\e_{\yy,j'}}_{\b,L}-
\media{\e_{\xx,j_1}}_{\b,L}\media{\e_{\yy,j_2}}_{\b,L}$. Note that the thermodynamic functions
like the free energy and the specific heat are {\it independent} of 
the lattice spacing $a$ while the correlation functions explicitly (if straightforwardly) depend on it.

If $\l=0$, Onsager's exact solution shows that the free energy is analytic for $\b\neq\b_c$,
where $\b_c=J^{-1}\tanh^{-1}(\sqrt2-1)$ is the inverse critical temperature.
For $\b<\b_c$ there is a unique infinite volume Gibbs state (the ``high temperature state") with zero 
spontaneous magnetization (i.e. $\media{\s_\xx}_\b=0$) and characterized by exponential decay
of the multipoint spin correlations as the separation between spins is sent to infinity. For $\b>\b_c$ 
there are exactly two pure Gibbs states $\media{\cdot}_\b^\pm$, which generate all possible infinite volume 
Gibbs states by convex combinations \cite{Ai80,Hi} (the ``low temperature states"): 
they have a non zero spontaneous magnetization (i.e. 
$\media{\s_\xx}_\b^\pm=\pm m^*(\b)\neq 0$)  and are characterized by exponential decay
of the multipoint {\it truncated} spin correlations as the separation between spins is sent to infinity. 
At the critical point, the rate of the exponential decay goes to zero proportionally to $|\b-\b_c|$ and the 
correlation functions decay polynomially to zero, with specific {\it critical exponents}; e.g.,
 the spin-spin as $\media{\s_\xx\s_\yy}\sim
({\rm const.})|\xx-\yy|^{-1/4}$, while
the energy-energy correlation decays as $\media{\e_{\xx,j}\e_{\yy,j'}}\sim
({\rm const.})|\xx-\yy|^{-2}$, asymptotically as $|\xx-\yy|\to\infty$. Correspondingly, the specific heat diverges 
logarithmically $C_v\sim({\rm const.})\log|\b-\b_c|$ at the critical point. We remark that, even after 
having understood the integrability structure of the model, underlying the Onsager's computation of the free energy,
the determination of the spin-spin critical exponent is very involved: $\media{\s_\xx\s_\V0}_\b$ 
is expressed as a determinant of a large matrix, of  size increasing with the distance between
spins, whose asymptotic behavior along special directions can be obtained via the use of Szego's lemma
\cite{MW} or via the use of the analyticity structure of a set of exact
quadratic difference equations that $\media{\s_\xx\s_\V0}_\b$
satisfies exactly at the lattice level \cite{MWP}.

On the other hand, the computation of the multipoint energy
correlations at $\l=0$ is a relatively simple matter. Defining the truncated energy correlations as
\be \media{\e_{\xx_1,j_1};\cdots;\e_{\xx_m,j_m}}^T_{\b,L}=
a^{-2n}\frac{\dpr^n}{\dpr A_{\xx_1,j_1}
\cdots\dpr A_{\xx_n,j_n}}\log Z({\boldsymbol A})\Big|_{{\bsA}={\bf 0}}\;,\label{2.2uk}\ee
with
\be Z({\boldsymbol A})=\sum_{\ss\in\O_M}\exp\{-\b H_M(\ss)+\sum_{j=1}^2\int_{\L^a_L}\e_{\xx,j} A_{\xx,j}\},
\label{1.7gh}\ee
it turns out that at $\l=0$ the correlations $ \media{\e_{\xx_1,j_1};\cdots;\e_{\xx_m,j_m}}^T_{\b,L}$ 
can be written in closed form for all finite $a,L$. Their expression further simplifies in the thermodynamic limit 
$L\to\infty$ 
and in a scaling limit in
which $a$ and $\b-\b_c$ are simultaneously sent to zero: defining $t_c=\tanh\b_c J=\sqrt2-1$, $t=\tanh\b(a) J$ and 
fixing $\b(a)$ in such a way that
$$\frac{t-t_c}{t_c}=\frac{a\s(a)}2\;,$$
with $\s(a)\neq 0$
for $a\neq 0$ and $\lim_{a\to 0}\s(a)=m^*\in\mathbb R$, then, for
all the $m$-tuples of non-coinciding points
$\xx_1,\ldots,\xx_m\in\RRR^2$, $m\ge 2$, $\xx_i\neq\xx_j$ $\forall
i\neq j$,
\bea &&
 \lim_{a\to 0}\lim_{L\to\io} \media{\e_{\xx_1,j_1};\cdots;\e_{\xx_m,j_m}}^T_{\b(a),L}\big|_{\l=0}
 =\label{1.maina}\\
 &&=-\frac{1}{2m}\Big(\frac{i}\p\Big)^m\sum_{\p\in\Pi_{\{1,\ldots,m\}}}\sum_{\o_1,\ldots,\o_m}\Big[\prod_{k=1}^{m}\o_k
{\mathfrak
g}^0_{\o_k,-\o_{k+1}}(\xx_{\p(k)}-\xx_{\p(k+1)})\Big]\nonumber\eea
where: (1) $j_k\in\{1,2\}$ and the limit is independent of the
choice of these labels; (2) $\Pi_{\{1,\ldots,n\}}$ is the set of
permutations over $\{1,\ldots,n\}$, with $\p(n+1)$ interpreted as
equal to $\p(1)$, and $\o_k\in\{\pm\}$, with $\o_{n+1}$
interpreted as equal to $\o_1$; (3) $\mathfrak g^0 (\xx)$ is given
by
\begin{equation}\label{mai}
\mathfrak g^0 (\xx) =\int \frac{d\kk}{2\p}\,\frac{e^{-i\kk\,\xx}}{
\kk^2+(m^*)^2}
\begin{pmatrix} ik_1+k_2 & im^*\\-i m^*&ik_1-k_2\end{pmatrix}\;,\ee
which is a well known object in Quantum Field Theory (QFT): it verifies 
the {\it Dirac equation} in $1+1$ dimensions, which describes a quantum
relativistic fermion
\be \left(\begin{array}{cc} \partial_1+i\partial_2 & im^*\\ -im^*
&
\partial_1-i\partial_2
\end{array}\right)\mathfrak g^0(\xx)=0\ee
The derivation of Eq.(\ref{1.maina}) will be reviewed below. It is based on a correspondence 
between the energy operator of the Ising model at the site $\xx$ 
and the mass operator $(i/\p)\psi_{\xx,+}\psi_{\xx,-}$ of a QFT of {\it non-interacting Majorana fermions},
with propagator $\mathbb E(\psi_{\xx,\o}\psi_{\yy,\o'})=\mathfrak g^0_{\o,\o'}(\xx-\yy)$.
The r.h.s.\ of Eq.(\ref{1.maina}) can be recognized as the multipoint truncated 
fermionic correlation $(i/\p)^m
\mathbb E^T(\psi_{\xx_1,+}\psi_{\xx_1,-};\cdots; \psi_{\xx_m,+}\psi_{\xx_m,-})$, \
where $\mathbb E^T$ is computed via the fermionic Wick rule and is given by the sum over 
the pairings (``contractions") of the $\psi$ fields; each pairing is equal to the product of the propagators 
associated with all the contracted pairs, times the sign of the permutation needed to bring the contracted 
fermionic fields next to each other. The outcome can be naturally represented in terms of Feynman diagrams
and truncation means that only connected diagrams are considered; in graphical terms, the r.h.s.\ of 
Eq.(\ref{1.maina}) can be thought of as a sum over simple {\it loop graphs}, see Fig.\ref{fig_circle0}.

\tikzset{ZspecialEP/.style={rectangle,fill=black,draw,inner
sep=3pt}}
\begin{figure}
\centering
$\media{\e_{\xx_1,j_1};\cdots;\e_{\xx_m,j_m}}^T_{\b(a),L}\big|_{\l=0} \to $
\begin{tikzpicture}[decoration=snake,x=0.5cm,y=0.5cm]
\node[specialEP] (x1) at (0:2) {}; \node[specialEP] (x2) at (60:2)
{}; \node[specialEP] (x3) at (120:2) {}; \node[specialEP] (x4) at
(180:2) {}; \node[specialEP] (x5) at (240:2) {}; \node[specialEP]
(x6) at (300:2) {}; \draw (x1) .. controls +(90:0.75)
and+(330:0.75) .. (x2) ; \draw (x1) .. controls +(270:0.75)
and+(30:0.75).. (x6) ; \draw (x5) .. controls +(330:0.75)  and
+(210:0.75) .. (x6); \draw (x2) .. controls +(150:0.75) and
+(30:0.75) .. (x3) ; \draw (x3) .. controls +(210:0.75) and
+(90:0.75) .. (x4) ; \draw (x4) .. controls +(270:0.75) and
+(150:0.75) .. (x5); \draw[decorate] (x1) -- + (0:1.5);
\draw[decorate] (x2) -- +(60:1.5); \draw[decorate] (x3) --
+(120:1.5); \draw[decorate] (x4) -- + (180:1.5); \draw[decorate]
(x5) -- + (240:1.5); \draw[decorate] (x6) -- + (300:1.5);
\end{tikzpicture}
$+$ permutations
\caption{Graphical representation of Eq.(\ref{1.mpoint}) in the
case $m=6$.
The wavy lines represent the external fields $A_{\xx,j}$, which can be thought of as being associated to the 
vertices. These are labeled by the space coordinate $\xx$ of the pair of fields $\psi_{\xx,+}\psi_{\xx,-}$ ``exiting
from the vertex". The fields are contracted two by two: the result of the contraction is represented by a solid line
(labeled by the $\o,\o'$ indices of the two contracted $\psi$ fields) and is associated with a propagator 
$\mathfrak g^0$.} \label{fig_circle0}
\end{figure}
Note that the thermodynamic limit $L\to\io$ is performed at
$\b=\b(a)\not=\b_c$; after the thermodynamic limit $\b(a)$ is sent to $\b_c$ and, depending on the 
speed at which $\b(a)\to\b_c$, the resulting expression has a different {\it effective mass} $m^*$, which 
is a finite real number. In particular, in the special case $m^*=0$, the propagator reduces to 
\be \mathfrak
g^0_{\o\o'}(\xx)\big|_{m^*=0}=\frac{\d_{\o,\o'}}{x_1+i\o
x_2}\;.\label{1.11l}\ee
In the literature, the explicit expression Eq.(\ref{1.maina}) with $m^*=0$ is usually
referred to as the {\it critical} multipoint energy correlation {\it in the plane}, see e.g.\ \cite{DSZ}; 
the result is independent of the specific boundary conditions imposed on 
$\L^a_L$. On the contrary, if the lattice spacing $a$ is sent to $0$ directly at $\b=\b_c$,
with $L$ kept fixed, then one gets a different
expression, which is sensitive to the specific boundary conditions imposed on $\L^a_L$: 
e.g. if periodic 
boundary conditions are considered, the multipoint energy correlations tend to the so-called  
{\it critical correlations on the torus} \cite{DSZ}. \\

All these results rely on the exact solution of the n.n.\ Ising
model. As mentioned above, an exact expression
for the energy correlations is known also at finite $a$ and has a
form similar to Eq.(\ref{1.maina}), still expressed in terms of loop graphs
but with a different (and more involved) lattice propagator; in the scaling limit,
it simplifies and reduces to the nice expression Eq.(\ref{1.maina}), formally
coinciding with the correlations of non interacting Majorana
fermions. When $\l\not=0$, the situation is radically different;
there is {\it no exact solution} in that case, and in particular
no explicit expression for the $m$-point energy correlations is
known. If the temperature is well inside the high or low
temperature phase, one can pursue a perturbative approach 
based on cluster expansion methods, which allows one to
prove quantitatively that the behavior of the perturbed system is
close to the exactly solvable one, see e.g.\
\cite{MS67,Ru69,GMM73}. The difficult and subtle case is when the
system is close to or at the critical point. In this case, the
cluster expansion argument breaks down: the perturbation theory is
affected by infrared divergences related to the slow decay of
correlations, which may a priori change the critical exponents and
the nature of the critical point. Spencer \cite{S00} and Pinson
and Spencer \cite{PS} introduced a new point of view using a
generalization of the mapping between 2D Ising models and
fermionic systems, first discovered by Schultz, Mattis and Lieb
\cite{SML}. While the n.n.\ Ising model is equivalent to a system
of non-interacting fermions, when $\l\not=0$ the equivalence is
with a system of {\it interacting fermions}. By such mapping
(which is reproduced and extended to more general interactions than those in \cite{PS} in Section 
\ref{sec:Grassman} below) one immediately gets a perturbative expansion 
for the energy correlations in terms of Feynman graphs of arbitrary order, see the first line of Fig.\ref{fig_circle}; 
contrary to the
$\l=0$ case, now graphs with any number of loops appear.\\
\tikzset{ZspecialEP/.style={rectangle,fill=black,draw,inner
sep=3pt}}
\begin{figure}
\centering
\begin{tikzpicture}[decoration=snake,x=0.5cm,y=0.5cm]
\node[specialEP] (x1) at (0:2) {}; \node[specialEP] (x2) at (60:2)
{}; \node[specialEP] (x3) at (120:2) {}; \node[specialEP] (x4) at
(180:2) {}; \node[specialEP] (x5) at (240:2) {}; \node[specialEP]
(x6) at (300:2) {}; \draw (x1) .. controls +(90:0.75)
and+(330:0.75) .. (x2) ; \draw (x1) .. controls +(270:0.75)
and+(30:0.75).. (x6) ; \draw (x5) .. controls +(330:0.75)  and
+(210:0.75) .. (x6); \draw (x2) .. controls +(150:0.75) and
+(30:0.75) .. (x3) ; \draw (x3) .. controls +(210:0.75) and
+(90:0.75) .. (x4) ; \draw (x4) .. controls +(270:0.75) and
+(150:0.75) .. (x5); \draw[decorate] (x1) -- + (0:1.5);
\draw[decorate] (x2) -- +(60:1.5); \draw[decorate] (x3) --
+(120:1.5); \draw[decorate] (x4) -- + (180:1.5); \draw[decorate]
(x5) -- + (240:1.5); \draw[decorate] (x6) -- + (300:1.5);
\end{tikzpicture}
$+$
\begin{tikzpicture}[decoration=snake,x=0.5cm,y=0.5cm]
\node[specialEP] (x1) at (0:2) {}; \node[specialEP] (x2) at (60:2)
{}; \node[specialEP] (x3) at (120:2) {}; \node[specialEP] (x4) at
(180:2) {}; \node[specialEP] (x5) at (240:2) {}; \node[specialEP]
(x6) at (300:2) {}; \draw (x1) .. controls +(90:0.75)
and+(330:0.75) .. (x2) ; \draw (x5) .. controls +(330:0.75)  and
+(210:0.75) .. (x6); \draw (x2) .. controls +(150:0.75) and
+(30:0.75) .. (x3) ; \draw (x4) .. controls +(270:0.75) and
+(150:0.75) .. (x5); \node[vertex] (v1) {}; \draw (x4) .. controls
+(90:1) and +(195:1) .. (v1); \draw (x3) .. controls +(210:1) and
+(105:1) .. (v1); \draw (x1) .. controls +(270:1) and + (15:1) ..
(v1); \draw (x6) .. controls +(30:1) and +(285:1) .. (v1);
\draw[decorate] (x1) -- + (0:1.5); \draw[decorate] (x2) --
+(60:1.5); \draw[decorate] (x3) -- +(120:1.5); \draw[decorate]
(x4) -- + (180:1.5); \draw[decorate] (x5) -- + (240:1.5);
\draw[decorate] (x6) -- + (300:1.5);
\end{tikzpicture}
$+$
\begin{tikzpicture}[decoration=snake,x=0.5cm,y=0.5cm]
\node[specialEP] (x1) at (0:2) {}; \node[specialEP] (x2) at (60:2)
{}; \node[specialEP] (x3) at (120:2) {}; \node[specialEP] (x4) at
(180:2) {}; \node[specialEP] (x5) at (240:2) {}; \node[specialEP]
(x6) at (300:2) {}; \draw[decorate] (x1) -- + (0:1.5);
\draw[decorate] (x2) -- +(60:1.5); \draw[decorate] (x3) --
+(120:1.5); \draw[decorate] (x4) -- + (180:1.5); \draw[decorate]
(x5) -- + (240:1.5); \draw[decorate] (x6) -- + (300:1.5); \draw
(x1) .. controls +(90:0.75) and+(330:0.75) .. (x2) ; \draw (x5) ..
controls +(330:0.75)  and +(210:0.75) .. (x6); \draw (x2) ..
controls +(150:0.75) and +(30:0.75) .. (x3) ; \draw (x4) ..
controls +(270:0.75) and +(150:0.75) .. (x5); \node[vertex] (v1)
at (0,0.5) {}; \node[vertex] (v2) at (0,-0.5) {}; \draw (x4) ..
controls +(90:1) and +(195:1) .. (v2); \draw (x3) .. controls
+(210:1) and +(105:1) .. (v1); \draw (x1) .. controls +(270:1) and
+ (15:1) .. (v1); \draw (x6) .. controls +(30:1) and +(285:1) ..
(v2); \draw (v1) .. controls +(-0.5,-0.5) and +(-0.5,0.5) .. (v2);
\draw (v1) .. controls +(0.5,-0.5) and +(0.5,0.5) .. (v2);
\end{tikzpicture}
$+ \cdots$
\\
$\to$
\begin{tikzpicture}[decoration=snake,x=0.5cm,y=0.5cm]
\node[ZspecialEP] (x1) at (0:2) {}; \draw (x1) .. controls
+(90:0.75) and+(330:0.75) .. (x2) ; \draw (x1) .. controls
+(270:0.75) and+(30:0.75).. (x6) ; \draw (x5) .. controls
+(330:0.75)  and +(210:0.75) .. (x6); \node[ZspecialEP] (x2) at
(60:2) {}; \draw (x2) .. controls +(150:0.75) and +(30:0.75) ..
(x3) ; \node[ZspecialEP] (x3) at (120:2) {}; \draw (x3) ..
controls +(210:0.75) and +(90:0.75) .. (x4) ; \node[ZspecialEP]
(x4) at (180:2) {}; \draw (x4) .. controls +(270:0.75) and
+(150:0.75) .. (x5); \node[ZspecialEP] (x5) at (240:2) {};
\node[ZspecialEP] (x6) at (300:2) {}; \draw[decorate] (x1) -- +
(0:1.5); \draw[decorate] (x2) -- +(60:1.5); \draw[decorate] (x3)
-- +(120:1.5); \draw[decorate] (x4) -- + (180:1.5);
\draw[decorate] (x5) -- + (240:1.5); \draw[decorate] (x6) -- +
(300:1.5);
\end{tikzpicture}
\caption{Schematic summary of the renormalization of the
correlation functions in the scaling limit.  The symbol  \protect
\tikz[baseline=-0.4em]{ \protect \node [ZspecialEP] {};} denotes a dressed 
vertex, proportional to a renormalization
factor $\bar Z (\l)$.\label{fig_circle}}
\end{figure}
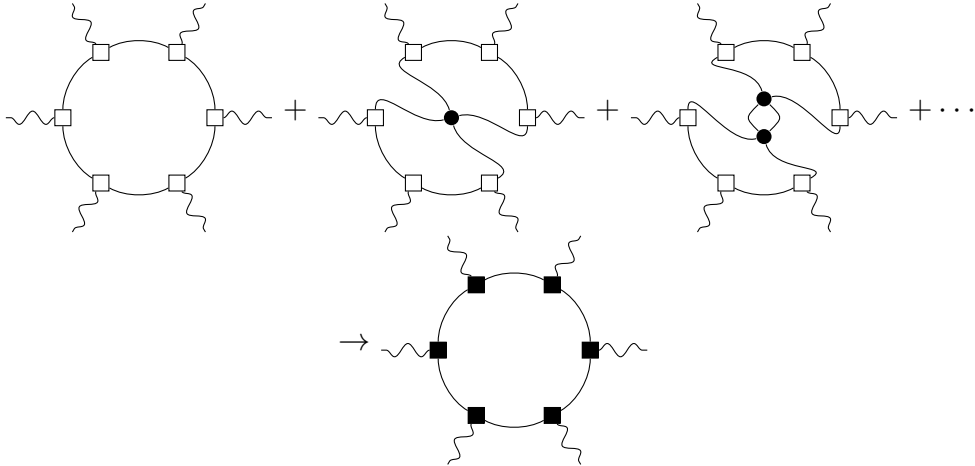

The aim of this paper is to control and compute the scaling limit 
of these energy correlations: this means to control and compute the ground state correlations 
of the underlying interacting fermionic theory, uniformly in the effective mass $m^*$. We show that the 
apparent infrared divergences at $m^*=0$ can be resummed 
by rigorous Renormalization Group methods: the outcome is a {\it
renormalized expansion} for the $m$-point energy correlations that is convergent for $\l$ small enough, 
as shown in \cite{PS} for the 2-point function. On top of it, and in addition to the estimates derived in 
\cite{PS}, we compute the {\it explicit exact expression} of the interacting multipoint energy correlations in the 
scaling limit and show that they coincide with an expression completely analogous to the one obtained  
for the integrable Ising model, up to a finite 
renormalization of the parameters. Graphically, the sum over infinite Feynman diagrams with an arbitrary number of
loops in the first line of Fig.\ref{fig_circle} can be shown to be {\it equal} to a sum over simple {\it dressed} 
loop diagrams, like the one in the second line of Fig.\ref{fig_circle}, 
up to subdominant corrections that vanish in the scaling limit. 
We stress that the use of fermionic methods is crucial to establish these results, which remained open problems 
for several years and appeared out of reach in terms of the original spin variables.
Our main result is summarized in the following theorem.

\begin{theorem}\label{mainprop}
There exists $\l_0>0$ such that if $|\l|\le \l_0$ the following is
true. There exists an analytic function of $\l$, called
$\b_c(\l)$, which is the critical temperature of the Ising model
Eq.(\ref{2.1}). Defining $t_c(\l)=\tanh\b_c(\l)J$, $t=\tanh\b(a)J$  and fixing $\b(a)$ in such a way that 
$$\frac{t-t_c(\l)}{t_c(\l)}=\frac{t_c(\l)}{t_c(0)}\frac{a\s(a)}{2}\;,$$
with $\s(a)\neq 0$ for $a\neq 0$ and
$\lim_{a\to 0}\s(a)=m^*\in\mathbb R$, then, for all the $m$-tuples
of non-coinciding points $\xx_1,\ldots,\xx_m\in\RRR^2$, $m\ge 2$,
$\xx_i\neq\xx_j$ $\forall i\neq j$, 
\bea &&
 \lim_{a\to 0}\lim_{L\to\io} \media{\e_{\xx_1,j_1};\cdots;\e_{\xx_m,j_m}}^T_{\b(a),L}
 =\label{1.main}\\
 &&=-\frac{1}{2m}\Big(\frac{i}\p\Big)^m\sum_{\p\in\Pi_{\{1,\ldots,m\}}}\sum_{\o_1,\ldots,\o_m}\Big[\prod_{k=1}^{m}\o_k
{\mathfrak
g}_{\o_k,-\o_{k+1}}(\xx_{\p(k)}-\xx_{\p(k+1)})\Big]\nonumber\eea
where: (1) $j_i\in\{1,2\}$ and the limit is independent of the
choice of these labels; (2) $\Pi_{\{1,\ldots,n\}}$ is the set of
permutations over $\{1,\ldots,n\}$, with $\p(n+1)$ interpreted as
equal to $\p(1)$, and $\o_k\in\{\pm\}$, with $\o_{n+1}$
interpreted as equal to $\o_1$; (3) the dressed propagator in the
scaling limit is
\begin{equation}
\mathfrak g (\xx) =\bar Z(\l) \int
\frac{d\kk}{2\p}\,\frac{e^{-i\kk\,\xx}}{ \kk^2+[Z^*(\l)m^*]^2}
\begin{pmatrix} ik_1+k_2 & iZ^*(\l)m^*\\-iZ^*(\l)m^*&ik_1-k_2\end{pmatrix}\;,\ee
where the renormalizations $\bar Z(\l), Z^*(\l)\in\mathbb R$ are
analytic functions of $\l$, analytically close to 1.
\end{theorem}

Therefore, even in presence of an interaction, the continuum limit
of the energy operator can still be identified with the mass
operator of a {\it free} fermion, up to a finite renormalization
of the mass and of the wave function. Note that the case $m^*=0$ is included in the theorem.
The critical temperature is renormalized by the interaction, too; on the contrary, the energy critical
exponents are protected against renormalization. The theorem is proved by
combining the methods of Pinson and
Spencer \cite{PS,S00} with the analysis of the massive Thirring model
developed by Benfatto, Falco and Mastropietro \cite{BFM02}. The proof of the existence of the scaling limit 
behind Theorem \ref{mainprop} also allows us to fully control how fast the
scaling limit of the energy correlations is reached, as stated explicitly in
the next theorem.

\begin{theorem}\label{thm2}
Given $\l, \b(a)$ and $\xx_1,\ldots,\xx_m$ as in the statement of
Theorem \ref{mainprop}, defining
\be \mathfrak g^{(a)} (\xx) =\bar Z(\l)\hskip-.3truecm
\int\limits_{[-\frac{\p}{a},\frac{\p}{a}]^2}\hskip-.2truecm\frac{d\kk}{2\p}
\frac{e^{-i\kk\,\xx}}{ |D_a(\kk)|^2+[Z^*(\l)\s(a)]^2}
\begin{pmatrix} D_a^+(\kk)& iZ^*(\l)\s(a)\\-iZ^*(\l)\s(a)& D_a^-(\kk)\end{pmatrix}\label{1.42q}\ee
we have:
\bea&&   \lim_{L\to\infty}\media{\e_{\xx_1,j_1};\cdots;\e_{\xx_m,j_m}}_{\b(a),L}^{T}=\label{1.43w}\\
&&\qquad
=-\frac{1}{2m}\Big(\frac{i}\p\Big)^m\sum_{\p\in\Pi_{\{1,\ldots,m\}}}\sum_{\o_1,\ldots,\o_m}\Big[\prod_{k=1}^{m}\o_k
{\mathfrak g}^{(a)}_{\o_k,-\o_{k+1}}(\xx_{\p(k)}-\xx_{\p(k+1)})\Big]+\nonumber\\
&&\qquad \quad+R^{(a)}(\xx_1,j_1;\cdots;\xx_m,j_m)\;,\nonumber\eea
where, denoting by $D_{\underline{\xx}}$ the diameter of
$\underline{\xx}=\{\xx_1,\ldots,\xx_m\}$ and by
$\d_{\underline{\xx}}$ the minimal distance among the points in
$\underline{\xx}$, for all $\th\in(0,1)$ and $\e\in(0,1/2)$ and a
suitable $C_{\th,\e}>0$, the correction term $R^{(a)}$ can be
bounded as
\be |R^{(a)}(\xx_1,j_1;\cdots;\xx_m,j_m)|\le  C_{\th,\e}^m m!
\frac1{\d^m_{\underline{\xx}}}\left(\frac{a}{\d_{\underline{\xx}}}\right)^\th
\left(\frac{\d_{\underline{\xx}}}{D_{\underline{\xx}}}\right)^{2-2\e}\;.\label{10b}\ee
\end{theorem}
Note that the contributions in the second line of Eq.(\ref{1.43w})
tend to Eq.(\ref{1.main}) in the limit as $a\to 0$, in a way
controlled by the difference $||\mathfrak g^{(a)}(\xx)-\mathfrak
g(\xx)||$, which is bounded by $|\xx|^{-1}e^{-({\rm
const.})m^*|\xx|}\big(a|\xx|^{-1}+|\s(a)-m^*|\cdot|\xx|\big)$ as
$a\to 0$ and $\s(a)\to m^*$. Note also that the combinatorial factor $m!$ in the r.h.s. of Eq.(\ref{10b})
is the optimal one and corresponds to the same combinatorial growth as $m\to\infty$ as the dominant term in 
the second line of Eq.(\ref{1.43w}).

Moreover, let us remark that the proof of
Theorem \ref{thm2} tells more than what is stated above. In
particular, it allows us to distinguish, among the contributions
to $R^{(a)}$, the terms of order 1 in $\l$ from those of order at
least $\l$; the former are very explicit and can be put in a form
similar to (even thought a bit more cumbersome than) the second
line of Eq.(\ref{1.43w}); if we isolated them from $R^{(a)}$, then
the remaining contributions to $R^{(a)}$ would be bounded by the
r.h.s.\ of Eq.(\ref{10b}) times $|\l|$. Finally, in the massive
case, $m^*\neq 0$, Eq.(\ref{10b}) can obviously be improved, by
taking into account the exponential decay of the propagators on
distances larger than $(m^*)^{-1}$. In this respect, the bound
Eq.(\ref{10b}) is stated with the massless case in mind, this
being the most difficult to treat.

\subsection{Ising models, fermionic systems and a comparison with existing literature}

In order to put our results in the right context and clarify
the relation between spin and fermionic systems, we
recall here some basic facts about the Ising model and its
generalizations. Of course the literature on the Ising model is
immense and we will give here only a partial view on the aspects more
related to our work, pointing out the connections with the
more recent literature and the main open problems.\\

Let us start by discussing the relation between the {\it integrable} Ising model and fermionic systems.
Schultz, Mattis and Lieb \cite{SML} derived the exact solution of
the n.n.\ Ising model by a method alternative to the Onsager's
one: they showed that the transfer matrix can be written as the
exponential of a quantum Hamiltonian describing {\it non
interacting} fermions on a one-dimensional lattice, and they
computed its trace by using second quantization methods. The
Schultz-Mattis-Lieb solution points out a connection between two
apparently unrelated systems, namely 2D classical spin systems and
quantum many body 1D non relativistic fermions. 
The relation between the {\it critical} Ising theory and a {\it relativistic} fermionic QFT
can then be realized in several ways. One is to start from the
Schultz-Mattis-Lieb representation, as in \cite{ID,BT}, and then
use the fact, discovered by Tomonaga \cite{To}, that a many body
system of non relativistic fermions in $d=1$ can be effectively
described in terms of relativistic fermions in $1+1$ dimensions.
Another possibility is to start from the Grassmann integral
representation of the Ising model partition function, as proposed
by Hurst and Green \cite{HG} and Samuel \cite{S80a}, see also
Itzykson-Drouffe \cite{ID}. This approach starts from the
observation, see \cite{MW}, that the generating function of the
n.n.\ Ising model with periodic boundary conditions can be
written in terms of four Pfaffians, each of which can be
represented as a suitable Grassmann integral. Recalling the definition of $Z(\bsA)$, Eq.(\ref{1.7gh}), 
the outcome is (see e.g.\ \cite{GM05,S80a})
\be Z^0({\bsA}):=Z({\bsA})\Big|_{\l=0}
=\frac12\sum_{\boldsymbol\a\in\{\pm\}^2}\t_{\boldsymbol\a}
Z^0_{\boldsymbol\a}(\bsA)\;,\label{2.4tt}\ee
with $\t_{+,-}=\t_{-,+}=\t_{-,-}=-\t_{+,+}=1$ and 
\bea  Z^0_{\boldsymbol\a}(\bsA)&=&(- 2a^{-2})^{M^2}  \Big[ \prod_{\xx \in \L^a_L}
\prod_{j=1}^2 \cosh ( \beta J+aA_{\xx,j})\Big]\int \DD\Phi\,
e^{S(\Phi,\bsA)}\;,\nonumber\\
S(\Phi,\bsA)&=&a\sum_{\xx\in \L^a_L}\Big[\sum_{j=1}^2
 \tanh( \beta J+aA_{\xx,j}) E_{\xx,j} +\label{2.5tt}\\
& &+\lis H_\xx H_\xx
+\lis V_\xx V_\xx+\lis V_\xx \lis H_\xx+ V_\xx
\lis H_\xx+ H_\xx \lis V_\xx+ V_\xx H_\xx\Big]\;.\nonumber
\eea
Here $\lis H_\xx, H_\xx,\lis V_\xx,  V_\xx$ are independent {\it Grassmann
variables} (see e.g.\ \cite{GM01} for the basic properties of Grassmann variables and integration),
four for each lattice site, and $E_{\xx,1}=\lis H_\xx H_{\xx+a\hat{\bf e}_1}$,
while $E_{\xx,2}=\lis V_\xx V_{\xx+a\hat{\bf e}_2}$. Moreover, $\Phi=\{\lis H_\xx,H_\xx,\lis V_\xx,$
$V_\xx\}_{\xx\in \L^a_M}$ denotes the collection of all of these
Grassmann symbols and $ \DD \Phi$ is a shorthand for
$\prod_{\xx}d\lis H_\xx dH_\xx  d\lis V_\xx dV_\xx$. The label $\boldsymbol\a=(\a_1,\a_2)$, with
$\a_1,\a_2\in\{\pm\}$, refers to the boundary conditions, which are periodic or antiperiodic in the
horizontal (resp. vertical) direction, depending on whether $\a_1$ (resp. $\a_2$) is equal to $+$ or $-$:
\bea &&
\lis H_{\xx + L \hat e_i} = \a_i \lis H_{\xx}, \qquad H_{\xx + L \hat e_i} = \a_i H_{\xx}\;,\\
&&
\lis V_{\xx + L \hat e_i} = \a_i \lis V_{\xx}, \qquad V_{\xx + L \hat e_i} = \a_i V_{\xx}
\;.\nonumber\eea
These boundary conditions may be important or not in the thermodynamic limit, depending on whether
we are at or outside the critical point. If $\b\neq \b_c$, then $|Z^0_{\boldsymbol\a}(\V0)/Z^0_{-,-}(\V0)|$ converges
exponentially to 1 as $L\to\infty$, see \cite[Appendix G]{M04};
moreover, if $\b>\b_c$, then $Z_{\boldsymbol\a}^0(\V0)$ is positive as $L\to\infty$ for all ${\boldsymbol\a}$,
while if $\b<\b_c$, then $Z^0_{+,+}(\V0)$ is definitely positive and the other partition functions
with different boundary conditions are
negative. At  $\b=\b_c$, $Z^0_{+,+}(\V0)$ 
is exactly zero while the other partition functions can be expressed in terms of  special functions \cite{DSZ,MW}.

The Grassmann representation is particularly convenient for computing energy-energy correlation functions: for
instance
\be\media{\e_{\xx,j};\e_{\yy,j'}}^T_{\b,L}\big|_{\l=0}=(1-t^2)^2
\sum_{\boldsymbol\a}\frac{\t_{\boldsymbol\a}Z^0_{\boldsymbol\a}(\V0)}{2Z^0(\V0)}\media{E_{\xx,j};E_{\yy,j'}}^T_{{\boldsymbol\a},L}\;,\ee
where $\media{\cdot}_{{\boldsymbol\a},L}$ is the average with respect to the Grassmann ``measure"
$\DD\Phi e^{S(\Phi,\V0)}$ with ${\boldsymbol\a}$ boundary conditions; we shall also indicate
by $\media{\cdot}_{{\boldsymbol\a}}$ the $L\to\infty$ limit of $\media{\cdot}_{{\boldsymbol\a},L}$.
As mentioned above, if $\b\not=\b_c$ then $\langle E_{\xx,j};E_{\yy_j'}\rangle^T_{{\boldsymbol\a},L}$
is {\it exponentially insensitive} to
boundary conditions as $L\to\infty$, so that, for all $\b\not=\b_c$,
\be\lim_{L\to\infty} \media{\e_{\xx,j};\e_{\yy,j'}}^T_{\b,L}\big|_{\l=0}=(1-t^2)^2
\langle E_{\xx,j};E_{\yy,j'}\rangle^T_{-,-}\;.\ee
A similar formula is true for the $m$-point functions, $m\ge 2$:
\be\lim_{L\to\infty} \media{\e_{\xx_1,j_1};\cdots;\e_{\xx_m,j_m}}^T_{\b,L}\big|_{\l=0}=(1-t^2)^{m}
\media{E_{\xx_1,j_1};\cdots;E_{\xx_m;j_m}}^T_{-,-}\;,\label{1.mpoint}\ee
valid for all $\b\neq \b_c$.
As mentioned above, if $\l=0$ the Grassmann ``measure" used to compute the r.h.s.\ of this equation is 
gaussian (i.e. $S(\Phi,\V0)$ is quadratic): therefore, the r.h.s.\ of Eq.(\ref{1.mpoint}) can be computed via
the fermionic Wick rule, leading to a sum over all the Feynman
graphs obtained by pairing (``contracting") the Grassmann fields
involved; truncation means that only connected Feynman diagrams
should be considered. From a graphical point of view,
Eq.(\ref{1.mpoint}) can be graphically interpreted as a sum over
simple {\it loop graphs}, as in Fig.\ref{fig_circle0}.
The fermionic representation outlined above also allows one to compute the spin-spin correlations,
even though these are represented by much more involved expressions. In particular, if $x>0$,
\be \media{\s_{(x,0)}\s_{0,0}}_{\b,L}\big|_{\l=0}=t^{x}\langle\!\langle e^{
a(t^{-1}-t)\sum_{z=0}^{x-a} E_{(z,0),1}}\rangle\!\rangle_{L}\;,\label{1.16}\ee
where $\langle\!\langle\cdot\rangle\!\rangle_L=\sum_{\boldsymbol\a}\frac{\t_{\boldsymbol\a}
Z^0_{\boldsymbol\a}}{2Z^0}\media{\cdot}_{{\boldsymbol\a},L}$. Eq.(\ref{1.16}) can be equivalently rewritten as
\bea && \log \media{\s_{(x,0)}\s_{0,0}}_{\b,L}\big|_{\l=0}=x\log t+\label{hh1}\\
&&+\sum_{m\ge 1}\frac{(t^{-1}-t)^m}{m!}a^m\sum_{z_1=0}^{x-a}\cdots \sum_{z_m=0}^{x-a}
\langle\!\langle E_{(z_1,0),1};\cdots; E_{(z_m,0),1}\rangle\!\rangle^T_L\;.\nonumber\eea
In other words, the spin-spin correlation can be expressed as a
series of truncated energy correlations. Such representation is
not very manageable and the asymptotic behavior of
$\media{\s_{(x,0)}\s_{0,0}}_{\b}\big|_{\l=0}$ at or close to the critical
point is usually derived by the (equivalent) representation in
term of a determinant of an $x \times x$ matrix \cite{MW}, even
though several attempts have been pursued along the years to resum
such series (see below).\\

In any case, the above Grassmann representation looks quite different from
those describing relativistic QFT fermionic models. However, we can 
perform a suitable change of variables \cite{ID} from $\{\lis
H_\xx, H_\xx,\lis V_\xx, V_\xx\}$ to the {\it critical modes}
$\{\psi_{\xx,\pm},\c_{\xx,\pm}\}$, see next section for a precise definition; 
the propagator of the $\psi$ field is massless at the critical point, while the $\c$ field is
uniformly massive. After the integration of the massive modes,
which can be performed exactly, we find that if $\l=0$
\be \media{\psi_{\xx_1,\o_1}\cdots\psi_{\xx_n,\o_n}}_{{\boldsymbol\a},L}=\frac1{\NN_{\boldsymbol\a}}\int \Big[\prod_{\substack{\xx\in\L^a_L\\ \o=\pm}}
d\psi_{\xx,\o}\Big] e^{\lis S (\psi)} \psi_{\xx_1,\o_1}\cdots\psi_{\xx_n,\o_n}
\label{1.18}\ee
where $\NN_{\boldsymbol\a}$ is a normalization and, if $\psi_\xx=\begin{pmatrix}\psi_{\xx,+}\\ \psi_{\xx,-}\end{pmatrix}$,
\be \lis S(\psi)=- \frac{1}{4\p} \int d\xx \,\psi^T_{\xx} \begin{pmatrix}
\hat \dpr_1+i\hat \dpr_2 & i \s+\hat \D_a \\
-i\s - \hat \D_a^\dagger & \hat \dpr_1-i\hat \dpr_2
\end{pmatrix}\psi_{\xx}\label{1.20}
\ee
where $\hat\dpr_j$ is a lattice approximation of the derivative in the $j$-th coordinate direction,
$\hat \D_a$ is a second order lattice differential operator, formally vanishing as $a\to 0$,
and $\s=\frac2{a}\frac{t-t_c(0)}{t_c(0)}$ is the mass, where $t_c(0)=\sqrt2-1$.
Remarkably, the above functional
integral coincides with the lattice regularization of a QFT
describing Majorana fermions in $d=1+1$ dimensions, with a Wilson
term (the second order operator $\hat \D_a$) which allows one to avoid the fermion doubling problem \cite{ID}.

Once again, the  r.h.s.\ of Eq.(\ref{1.18}) can be computed exactly in terms of the fermionic Wick rule, 
with propagator given
by the covariance of the quadratic form in Eq.(\ref{1.20}). In the scaling limit,
this propagator tends to $\mathfrak g^0$ of Eq.(\ref{mai}); correspondingly, the r.h.s.\ of Eq.(\ref{1.mpoint}),
after the change of variables from $\{\lis
H_\xx, H_\xx,\lis V_\xx, V_\xx\}$ to 
$\{\psi_{\xx,\pm},\c_{\xx,\pm}\}$ and the integration of the $\c$ fields, tends to the r.h.s.\ of Eq.(\ref{1.maina}).
If $m^*{=}0$, the limiting theory is a scale-free relativistic Conformal Field Theory (CFT). Conformal 
invariance appears to be a robust property, stable under changes in the shape of the box and of the underlying 
lattice, as proved by Chelkak and Smirnov \cite{CS} by making use of suitable fermionic operators
on the lattice \cite{RC06} and the ideas of Schramm-Lowner Evolution (SLE) \cite{LSW}.

While the discussion of the continuum limit of the energy
correlations starting from the exact solution (that is, from the
Grassmann integral representation) is straightforward and fully
rigorous, the analogous discussion for the spin-spin correlation
is much harder. 
One tempting route to its computation is to start from Eq.\pref{hh1} and then replace 
the sums by integrals and the energy
truncated expectations by their scaling limit Eq.(\ref{1.maina}); this is the strategy followed by \cite{BI, DD,ZI}. 
However, these replacements introduce spurious ultraviolet divergences that 
require a suitable interpretation; some justification of this procedure has been
provided by Dotsenko and Dotsenko \cite{DD}, but it is fair to say that so far there is 
no mathematically sound way to resum the series in
Eq.\pref{hh1} to get the spin-spin correlation. 
A better alternative approach is based on bosonization, which
connects this (among many other quantites) to observables of a
discrete height model \cite{Dubedat}. This height model has long
been known to converge in a weak sense to the Gaussian free
field~\cite{K}, and a recent preprint by
Dub{\'e}dat~\cite{Dubedat2} extends this convergence to the
relevant class of observables.\\

Let us return now to the {\it non integrable} model Eq.\pref{2.1} with $\l\neq 0$. A
basic principle in statistical physics, which has a wide
experimental confirmation, is {\it universality} \cite{Ba,F,G},
which tells us in particular that the physical properties of a
system close to a second order phase transitions are largely
independent of the microscopic details of the interaction among
its elementary components: only a few general properties related
to dimensionality and the symmetries of the Hamiltonian matter as
far as the computation of its critical exponents is concerned.
It is widely believed that the Ising model with finite range interactions Eq.\pref{2.1} 
belongs to the same class of universality as the n.n.\
Ising model, that is, the critical exponents of the two models are the
same. On the other hand, other physical quantities, like the
critical temperature, are expected to be non-universal. However, a
mathematical proof of this conjecture is very difficult:
universality in two dimensions is not a trivial issue at all. A striking illustration of this fact is provided 
by the exact solution of the eight vertex model by Baxter \cite{Ba1},
which came after the solution of the six vertex model by Lieb
\cite{Li}, and showed that the critical exponents of these vertex
models are continuous non trivial functions of the coupling. Eight
vertex models can be represented equivalently as a pair of Ising
models coupled by a four spin interaction \cite{Ba}; therefore,
the Ising universality class should be stable under in-layer
perturbations, but not against perturbations coupling two
different layers. The classification of the possible critical theories close to the Ising model is 
by itself a very rich and interesting research field. It is based on the remark, due to Belavin, Polyakov and 
Zamolodchikov \cite{BPZ}, that if a 2D critical theory admits a scaling limit, this should be 
invariant under the  infinite dimensional group of conformal transformations of the complex plane.
Such a large symmetry group imposes infinitely many constraints on the correlation functions and in some cases 
these are sufficient to compute them in closed form. Even more strikingly, in many cases the critical exponents are 
all functions of a single parameter, the central charge $c$, which in turn can be computed from the
self-correlation of the energy-momentum tensor $T(z)$, $\media{T(z)T(0)}=\frac{c}{2z^4}$
\cite{BPZ,ID2}. In our case, 
$T(z)=-\frac12\psi\dpr\psi$ and
$\media{T(z)T(0)}=\frac1{4z^4}$, from which $c=\frac12$ and the
simplest conformal field theory with central charge $1/2$ has two
primary fields with critical exponents $2$ and $1/4$,
respectively, which are naturally identified with the energy and
the spin operators of the continuum limit of the critical Ising
model \cite{BPZ,ID2}. 
Universality and conformal invariance of the nearest neighbor Ising model on regular lattices and
domains of different shapes has been recently proved by Smirnov \cite{S}, Chelkak and Smirnov \cite{CS}
and by Chelkak, Hongler and Izyurov \cite{CHI}, by SLE methods.
However, it should be stressed that the results of Smirnov, Chelkak and collaborators heavily rely on the underlying
integrability properties of the nearest neighbor Ising model and are very fragile under perturbations of the form 
Eq.(\ref{2.1}). The problem of proving the existence of the scaling limit of the critical theory 
Eq.(\ref{2.1}) at $\l\neq 0$ and of its conformal invariance  still remains a big challenge.
\\

A different approach to universality which is suitable to include perturbations of the form Eq.(\ref{2.1})
 is provided by the Renormalization Group \cite{WK}. A rigorous application 
of this idea to the present context has been proposed by Pinson and Spencer \cite{PS,S00}. 
As mentioned above, their starting point is an exact expression of the 
partition function of certain
non-integrable Ising models in terms of a non-gaussian Grassmann integral,
which can be studied by constructive QFT methods, similar to those used by
Lesniewski \cite{Le} to analyze the Yukawa$_2$ QFT: in both cases the interaction is {\it irrelevant} in
the Renormalization Group sense. Using such methods, Pinson and Spencer 
proved that the critical exponent of the energy correlation is
independent of the next-to-nearest-neighbor interaction, provided
this is chosen of a suitable form. Their proof provides the first
example of universality of a critical exponent in a perturbed 2D
Ising model. This approach based on fermionic mapping and
constructive Renormalization Group analysis can be applied to
coupled Ising models, like the Eight Vertex or the Ashkin-Teller
model; in such cases, as shown by Mastropietro \cite{M04}, the
interaction becomes {\it marginally relevant} in the
Renormalization Group sense and the energy exponents are non
universal continuous function of the coupling. Extensions of such
methods allowed Benfatto, Falco and Mastropietro to prove the
Kadanoff relations \cite{Ka} between the specific heat, the energy
and the crossover critical exponents in eight vertex and
Ashkin-Teller models \cite{BFM02}; they also allowed Giuliani and
Mastropietro to compute a new critical exponent controlling the
crossover from universal to non-universal behavior in the
asymmetric Askhin-Teller model \cite{GM05}. It should be remarked
that these exponents cannot be computed by other means, since the
Ashkin-Teller model is not solvable (and we also recall that the eight
vertex model is solvable but only certain exponents can be deduced
from the solution).

In this paper we extend the analysis of Pinson and Spencer \cite{PS,S00} 
in several directions. First, we generalize the class of spin perturbations that can be considered:
while they required special next-to-nearest neighbor interactions, we just need the function $v$ in Eq.(\ref{2.1})
to be of finite range and symmetric under the natural lattice symmetries; the resulting exact fermionic action
is defined in terms of  an exponentially
decaying self-interaction of the involved Grassmann field, see Proposition
\ref{prop3} in Section \ref{sec2.1} below. It
should be noted that even though the interaction in terms of spins
is finite range, the corresponding fermionic interaction is of
infinite range and of arbitrary high degree in the fermionic
field, but exponentially decaying on the scale of the lattice. The
result is based on a first cluster expansion, which is
exponentially convergent provided that the strength $\l$ of the
perturbation is small enough. Second, we combine the analysis of Pinson and Spencer with one used by
Benfatto, Falco and Mastropietro \cite{BFM07} to prove and control the convergence of its energy 
correlations to the continuum limit. In this way, besides the critical exponent of the 2-point energy correlation, 
we can explicitly compute the scaling limit of all the multipoint correlation functions: in particular we have a 
constructive procedure to compute the amplitude of the energy correlations and the subdominant corrections 
that vanish in the scaling limit. As a technical point, let us also mention that, as compared to \cite{BFM07},
our bound on the subdominant corrections is optimal from a combinatorial point of view: it grows as $m!$ as 
$m\to \infty$, exactly as the dominant term, contrary to the bound in \cite{BFM07}, which grows as $(m!)^\a$,
$\a>1$.
Our results can be seen as a strong statement of universality with respect to perturbations of the form Eq.(\ref{2.1}):
the critical correlations in the scaling limit have the same analytical expression as those of the n.n. Ising model,
up to a renormalization of the wave function and of the
mass. \\

Of course, many important problems remain to be faced. One is to
repeat the analysis for the energy correlations on the torus,
taking the scaling limit directly at the critical point. However, the most
urgent problem is  to prove universality of the spin critical exponent. As we have seen before, the explicit
expressions for the energy correlations of the Ising model have
been used by Dotsenko and Dotsenko \cite{DD} to compute the spin-spin critical exponent
(modulo a number of audacious exchange of limits). 
The spin-spin correlation for
the model Eq.\pref{2.1} with $\l\not=0$ can be still expressed by a formula similar to 
Eq.\pref{hh1}, with $(t^{-1}-t)$ replaced by a renormalized coefficient $t_\l$ and 
the free truncated energy correlations replaced by the interacting ones at $\l\neq0$.
If we proceed as in \cite{DD} and, in the scaling limit, we 
replace the
sums in Eq.(\ref{hh1}) by integrals and the truncated energy correlations by their limiting value,
we get an expression essentially identical to the non
interacting one, with the important difference that the $m$-th order term in the series has a prefactor
proportional to $(t_\l\bar Z(\l))^m$. After a resummation (if formal) of the series, the factor $t_\l\bar Z(\l)$ appears to be related to the critical exponent of the spin-spin correlation. In this perspective, the proof
of universality for this exponent could be reduced to the proof that the combination $t_\l\bar Z(\l)$
is identically $1$ as a function of $\l$. However, one might worry that the exchange of the
scaling limit with the integrals could produce an extra finite 
renormalization that could be difficult to control; this question
requires further investigation. Another possible strategy for
determining the exponent of the spin-spin correlation is by
computing the self-correlation of the energy-momentum tensor: if we could
identify the right lattice counterpart of this operator (which is
usually defined directly in the scaling limit) then we could have
direct access to the central charge of the interacting theory
(which should be $1/2$ if the universality principle is correct).\\

The rest of the paper is devoted to the proofs of Theorems \ref{mainprop} and \ref{thm2}. In
Section~\ref{sec:Grassman} we prove the Grassmann representation of the generating function for the multi-point
energy correlations. In Section \ref{sec3} we describe the Renormalization Group procedure used to control
the Grassmann generating function uniformly in the mass $m^*$. Finally, in Section \ref{sec:correlations}
we explain how to adapt the general expansion and the bounds discussed in Section \ref{sec3} to the
multi-point energy correlation functions and conclude the proof of the two theorems stated above.

\section{The representation of the interacting Ising model in Grassmann variables}\label{sec:Grassman}

\subsection{The generating function for the energy correlations}\label{sec2.1}

Consider the model Eq.(\ref{2.1}) where $v(\xx)$ has the properties spelled after Eq.(\ref{2.1}). Without loss of generality we can assume that the additional nearest neighbor interaction is 
zero: $v(a\hat{\bf e}_j)=0$, $j=1,2$. The interaction is written in terms of the macroscopic coordinates for a unified notation, but in 
fact depends on the lattice distance; that is $v(\xx-\yy)= v_0(\frac{\xx-\yy}{a})$ for some $v_0$ 
independent of $a$.
We shall assume that $v$ is normalized in such a way that 
$\frac12\sum_{\xx} |v(\xx)|=1$. For notational simplicity, we shall also assume that 
$a=\ell_02^{-N}$, where $\ell_0$ is the macroscopic unit length and $N\in\mathbb N$. Moreover, 
we define $\BB^a_M$ to be the set of nearest neighbor bonds in 
$\L^a_L=a{\mathbb Z}^2_M$.

As reviewed in the previous section, in the simple Ising model case ($\l=0$), the generating function for the energy correlations can be 
represented in terms of a gaussian Grassmann integral, see Eqs.(\ref{2.4tt})-(\ref{2.5tt}). A similar 
representation in terms of a non-gaussian Grassmann integral is valid also in the $\l\neq 0$ case, as stated in 
the following proposition and proved below. 
\begin{proposition}\label{prop3}
There exists $\l_0>0$ such that, if $|\l|\le \l_0$, then for any $m$-tuple of distinct pairs $(\xx_1,j_1),\ldots,
(\xx_m,j_m)$, with $m\ge2$, $\xx_i\in \L^a_M$ and $j_i\in\{1,2\}$,
\be  \media{\e_{\xx_1,j_1};\cdots;\e_{\xx_m,j_m}}^T_{\b,L}=a^{-2m}\frac{\dpr^m}{\dpr A_{\xx_1,j_1}
\cdots\dpr A_{\xx_m,j_m}}\log\Xi(\bsA)\Big|_{\bsA={\bf 0}}\;,\ee
where $\Xi(\bsA)$ is the Grassmann generating functional
\be \Xi(\bsA)
 = \frac{C_M}2 \sum_{\boldsymbol\a\in\{\pm\}^2}\t_{\boldsymbol\a}\,
 \int  \DD \Phi \,
e^{S_t(\Phi)+ (1-t^2)({\boldsymbol{E}},\bsA)+\VV(\Phi,\bsA)}\label{2.51}
\ee
where $t:=\tanh(\b J)$, $({\boldsymbol{E}},\bsA):=\sum_{j=1}^2\int d\xx\, E_{\xx,j} A_{\xx,j}$ and:
\begin{itemize}
\item $C_M$ is a normalization constant, defined as
\be 
C_M=(- 2a^{-2}\cosh^2(\b J))^{M^2}e^{V_M(\l)}\prod_{\{\xx,\yy\}} \cosh^2\!\big(\frac{\b\l}2 v(\xx-\yy)\big) \label{2.52}
\ee
with $V_M(\l)$ an analytic function of $\l$, independent of $a$ and satisfying the bound 
$|V_M(\l)|\le C|\l|M^2$, for a suitable $C>0$;
\item $S_t(\Phi)$ is the unperturbed quadratic part of the action, defined as
\bea S_t(\Phi) &=&a\sum_{\xx\in \L^a_L}\!\!\big(
tE_{\xx,1}+tE_{\xx,2} +\\
&&+\lis H_\xx H_\xx
+\lis V_\xx V_\xx+\lis V_\xx \lis H_\xx+ V_\xx
\lis H_\xx+ H_\xx \lis V_\xx+ V_\xx H_\xx\big)\nonumber\eea
\item $\VV(\Phi,\bsA)$ is a polynomial in $\{E_{\xx,j}\}_{\xx\in \L^a_M}^{j=1,2}$ and 
$\{A_{\xx,j}\}_{\xx\in \L^a_M}^{j=1,2}$,
which can be expressed as
\bea
\VV(\Phi,\bsA) &=& \sum_{\substack{n,m\ge 0\\ n+m\ge 1}}\sum_{\substack{j_1,\ldots,j_n\\ j'_1,\ldots,j'_m}}
\int d\xx_1\cdots d\xx_n\,d\yy_1\cdots d\yy_m\cdot\label{V_as_W}\\
&&\cdot W_{\underline{j},{\underline{j}}'}(\xx_1,\ldots,\xx_n;\yy_1,\ldots,\yy_m)
 \prod_{i=1}^n E_{\xx_i,j_i}\prod_{i'=1}^m A_{\yy_{i'},j_{i'}'} \nonumber
\eea
where $\underline{j}=(j_1,\ldots,j_n)$, ${\underline{j}}'=(j_1',\ldots,j_m')$ and, if $\underline{\xx}=
(\xx_1,\ldots,\xx_n)$ and $\underline{\yy}=(\yy_1,\ldots,\yy_m)$,
\be
|W_{\underline{j},{\underline{j}}'}(\underline{\xx};\underline{\yy})|\le 
C^{n+m}(\b |\l|)^{\max\{1,c(n+m)\}} a^{-(2-n-m)} \frac{e^{-\k\, \d(\underline{\xx},\underline{\yy})/a}}
{a^{2(n+m-1)}}\label{W_decay}\ee
for suitable constants $C,c,\k>0$ depending only on $M_0$ (the range of the interaction); 
here $\d(\xx_1,\ldots,\yy_m)$ is the tree 
distance of the set $X=\{\xx_1,\ldots,\yy_m\}$, that is the length of the shortest tree graph composed of
bonds in $\BB_M^a$ which connects all the elements of $X$.
\end{itemize}
\end{proposition}
{\bf Remark.} The factor $\frac{e^{-\k \d(\underline{\xx},\underline{\yy})/a}}{a^{2(n+m-1)}}$ in Eq.(\ref{W_decay}) is normalized in 
such a way that its $L_1$ norm is essentially independent of $a$, 
that is $$\frac1{L^2}\int d\xx_1\cdots d\yy_m \frac{e^{-\k \d(\underline{\xx},\underline{\yy})/a}}
{a^{2(n+m-1)}}\le C_\k$$ where $C_\k$ is a constant independent of $a$. On the other hand,
the factor $(\b|\l|)^{\max\{1,c(n+m)\}}$ measures the smallness in $\l$ of the kernel $W$, while 
the factor $a^{-(2-n-m)}$ (or, better, its a-dimensional version, $(a/\ell_0)^{-(2-n-m)}$) plays the role of its {\it scaling dimension}. Recall that $a/\ell_0= 2^{-N}$,
see the beginning of Section \ref{sec2.1}: therefore, the scaling dimension of the kernel can be 
rewritten as $2^{N(2-n-m)}$. We shall see in the following that 
the generating function can be integrated by an iterative multiscale procedure; at each step, it will 
be expressed in terms of an effective potential on scale $h$, with $h\le N$, 
analogous to $\VV(\Phi,\bsA)$,
whose kernels $W^{(h)}$ have decay properties analogous to $W$, with the 
important difference that the scale $2^N$ is replaced by $2^h$ and, therefore, 
the scaling dimension $2^{N(2-n-m)}$ is replaced by $2^{h(2-n-m)}$.
The result of the iteration, as $h\to-\infty$, gives the generating function of interest. The 
relevant scaling properties of the multi-point energy correlation function as $a\to 0$ and $\b\to\b_c$ will 
be controlled in terms of the kernels $W^{(h)}$ and of their limits as $N\to\infty$ and $h\to-\infty$.

\begin{proof}[Proof of Proposition \ref{prop3}]
 For notational convenience, given a bond $b\in\BB^a_M$, 
we denote by $\e_b$, $E_b$ and $A_b$ the corresponding bond operators: that is, 
if $b=(\xx,\xx+a\hat{\bf e}_1)$ (resp. $b=(\xx,\xx+a\hat{\bf e}_2)$), then $\e_b=\e_{\xx,1}$, 
$E_b=E_{\xx,1}$ and $A_b=
A_{\xx,1}$ (resp. $\e_b=\e_{\xx,2}$, $E_b=E_{\xx,2}$ and $A_b=
A_{\xx,2}$). The key to the first step in the proof is the remark is that if $b_1,\ldots,b_n$ are $n$ 
{\it distinct} 
bonds, then the Grassmann representation for the nearest neighbor Ising model 
induces the following:
\bea && \sum_{{\underline \s}\in\O_M}e^{\sum_{b}(\b J+a A_{b})a\e_b}
\e_{b_1}\cdots\e_{b_n}=\frac12\sum_{\boldsymbol\a}\t_{\boldsymbol\a}\,a^{-2n}
\frac{\dpr^n}{\dpr A_{b_1}\cdots\dpr A_{b_n}}Z^0_{\boldsymbol\a}(\bsA)=\nonumber\\
&&=\frac12\sum_{\boldsymbol\a}\t_{\boldsymbol\a}\,
(- 2a^{-2})^{M^2}  \Big[ \prod_{b \in \BB^a_M} \cosh ( \beta J+aA_b) \Big]\cdot\label{2.6a}\\
&&\cdot\int \DD\Phi\,\big(a^{-1}t_{b_1}+(1-t_{b_1}^2)E_{b_1}\big)\cdots
\big(a^{-1}t_{b_n}+(1-t_{b_n}^2)E_{b_n}\big)e^{S(\Phi,\bsA)}\;,\nonumber\eea
where $t_b=\tanh(\b J+aA_b)$. This correspondence is invalid for repeated bond variables: note that 
$[a^{-1}t_b+(1-t_b^2)E_b]^2=a^{-2}t_b^2+2a^{-1}(1-t_b^2)^2E_b$, while  $\e_b^2=a^{-2}$. 
This last observation can be used to remove repeated bond operators from any expression; therefore, in order
to derive a Grassmann representation for
$Z(\bsA)$, it is enough to express the interaction term (i.e. the $\l$-dependent term) in Eq.(\ref{1.7gh}) as sum of products of 
\emph{distinct} bond operators; then we can replace every bond operator $\e_b$ by $a^{-1}t_{b}+
(1-t_{b}^2)E_{b}$, in the sense explained above, and finally we can re-exponentiate the big sum
of products of Grassmann variables, so obtaining the desired Grassmann functional integral 
representation of the Ising model at hand. This can be implemented as follows. By definition,
\be Z(\bsA)=\sum_{{\underline \s}\in\O_M}e^{
\sum_{b}(\b J+a A_{b})a\e_b}
\prod_{\{\xx,\yy\}}e^{\b \l\s_\xx v(\xx-\yy)\s_\yy}\;.\label{2.58}\ee
Consider a pair of sites $\{\xx,\yy\}$ contributing to the product in the r.h.s.\ of this equation, i.e., a pair of 
sites such that $|\xx-\yy|\le R_0$. Note that $\s_\xx\s_\yy$ can 
be rewritten in terms of a product of energy density operators localized along a path connecting $\xx$ 
and $\yy$ on the lattice: $\s_\xx\s_\yy=\frac12U_{\xx,\yy}+\frac12D_{\xx,\yy}$, where 
$U_{\xx,\yy}=\prod_{b\in\CC_U(\xx,\yy)}a\e_b$ and $D_{\xx,\yy}=\prod_{b\in\CC_D(\xx,\yy)}a\e_b$.
Here $\CC_U(\xx,\yy)$ and $\CC_D(\xx,\yy)$ (where $U$ and $D$ stand for ``up" and ``down") are 
the two paths connecting $\xx$ and $\yy$ on the lattice described in Fig. \ref{fig1}.  
Note that we choose paths in this way in order to be sure that we have an expression which manifestly 
retains the rotation and reflection symmetries of the original interaction.
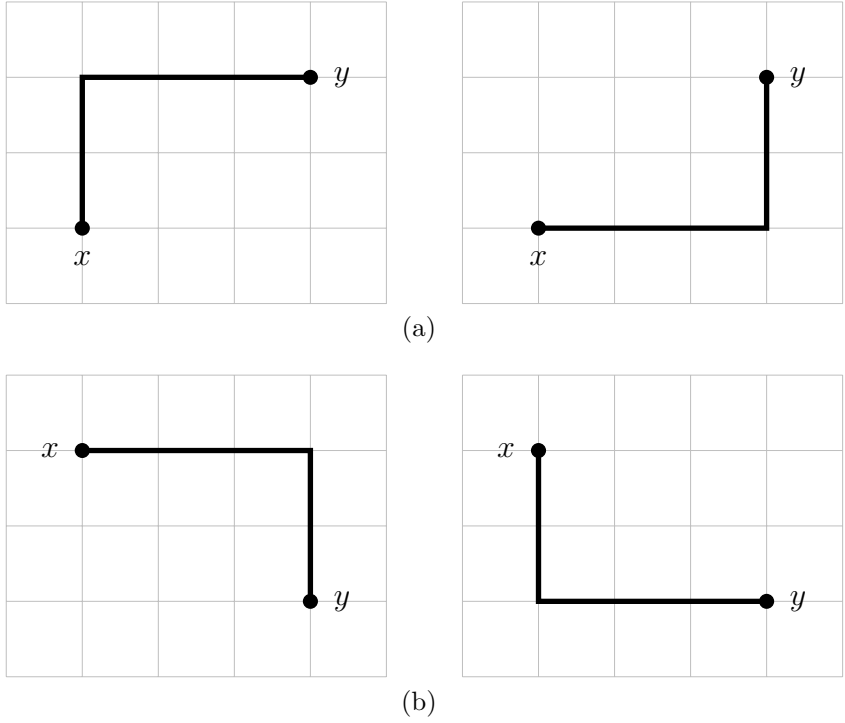
\begin{figure}[h]
\centering
\subfloat[]{
\begin{tikzpicture}[line width=2pt]
\draw[lightgray,line width=0] (-1,-1) grid ++(5,4);
\draw +(0,0) node [vertex,label=below:$x$] {} |- ++(3,2) node [vertex,label=right:$y$] {};
\draw[lightgray,line width=0] (5,-1) grid ++(5,4);
\draw (6,0) node [vertex,label=below:$x$] {} -| ++(3,2) node [vertex,label=right:$y$] {};
\end{tikzpicture}}\\
\subfloat[]{
\begin{tikzpicture}[line width=2pt]
\draw[lightgray,line width=0] (-1,-1) grid ++(5,4);
\draw +(0,2) node [vertex,label=left:$x$] {} -| ++(3,0) node [vertex,label=right:$y$] {};
\draw[lightgray,line width=0] (5,-1) grid ++(5,4);
\draw (6,2) node [vertex,label=left:$x$] {} |- ++(3,-2) node [vertex,label=right:$y$] {};
\end{tikzpicture}}
\caption{The two paths $\CC_U(\xx,\yy)$ (on the left) and $\CC_D(\xx,\yy)$ (on the right) in the two cases where: (a) $x_1<y_1$ and $x_2<y_2$; (b) $x_1<y_1$ and $x_2>y_2$. If $x_1=y_1$ or $x_2=y_2$, we choose $\CC_U(\xx,\yy)=\CC_D(\xx,\yy)$ to be the horizontal or vertical segment 
connecting $\xx$ with $\yy$; this means that horizontal or vertical strings are counted twice.}
\label{fig1}
\end{figure}
In terms of these ``string operators", we can use the identity 
$$
 e^{\pm x}=\cosh x(1 \pm \tanh x)
$$ and the fact that $U, D = \pm 1$ to rewrite the product in the r.h.s.\ of Eq.(\ref{2.58}) as:
\bea && \prod_{\{\xx,\yy\}} e^{\b\l\s_\xx v(\xx-\yy)\s_\yy}=
\prod_{\{\xx,\yy\}} e^{\frac{\b\l}2 v(\xx-\yy)\big(U_{\xx,\yy}+D_{\xx,\yy}\big)}=
\Big[\prod_{\{\xx,\yy\}} \cosh^2\!\big(\tfrac{\b\l}2 v(\xx-\yy)\big)\Big]\cdot\nonumber\\
&&\cdot\Big[\prod_{\{\xx,\yy\}}\Big(1+\tanh\!\big(\tfrac12 \b \l v(\xx-\yy)\big) U_{\xx,\yy} \Big)
\Big(1+\tanh\!\big(\tfrac12 \b \l v(\xx-\yy)\big)  D_{\xx,\yy} \Big)
\Big]
\;.\nonumber\eea
The second line is a product of binomials, each consisting of 1 plus a non trivial term; these non trivial 
terms can each be graphically associated with a ``string" $S$ (i.e., the union of the bonds in
$\CC_U(\xx,\yy)$ or $\CC_D(\xx,\yy)$), either of type U or D, depending on 
whether it is associated to the path $\CC_U(\xx,\yy)$ 
or $\CC_D(\xx,\yy)$; in both cases we shall write $v_S:=v(\xx-\yy)$.  Note that the assumption that the interaction has range $a M_0$ means that all of the strings we are considering consist of no more than $M_0$ bonds. 
If we now expand the product $$\Big[\prod_{\{\xx,\yy\}}\Big(1+\tanh\!\big(\tfrac12 \b \l v(\xx-\yy)\big) U_{\xx,\yy} \Big)
\Big(1+\tanh\!\big(\tfrac12 \b \l v(\xx-\yy)\big)  D_{\xx,\yy} \Big)
\Big]$$ we get a 
sum over all subsets of the collection of strings which appear in the interaction under consideration, in which each term is a product of string operators.  Moreover, every such subset can be thought of as a collection of its maximal connected components 
(here we say that a set of strings $\SS=\{S_1,\ldots,S_m\}$ is connected if, given $1\le i_0<j_0\le m$,
we can find a sequence $(S_{i_0},S_{i_1},\ldots,S_{i_p}\=S_{j_0})$ such that $S_{i_l}\cap 
S_{i_{l+1}}\neq \emptyset$). From a graphical point of view, every connected component $\SS$
corresponds in a non-unique way to a polymer $\g(\SS)$, i.e., a connected set of bonds.   It is helpful to color the bonds in $\g(\SS)$ black or gray, depending 
on whether the given bond belongs to an odd or even number of strings in $\SS$, and denote the set of bonds thus colored black by $\black(\SS)$. 
We call the collection of terms associated with a polymer $\g$ the activity of $\g$ and denote it by
$z(\g)$ and noting that $a^2\e_b^2 \equiv 1$
\be z(\g)=\sum_{\substack{\SS \text{ connected}: \\ \g(\SS)=\g}}\Big[\prod_{S\in\SS}\tanh(\tfrac12 \b\l v_S)\Big]\Big[\prod_{b \in \black(\SS)} a\e_b\Big]\;,\label{2.7a}\ee
in terms of which we can finally rewrite
\be
 \prod_{\{\xx,\yy\}} e^{\b\l\s_\xx v(\xx-\yy)\s_\yy}=
\Big[\prod_{\{\xx,\yy\}} \cosh^2 \!\big(\tfrac12 \b\l v(\xx-\yy)\big)\Big]\sum_{\G\subseteq\L} \f(\G) \prod_{\g\in\G} z(\g) \label{2.7b}
\ee
where the sum $\sum_{\G\subseteq\L}$ in the r.h.s.\ runs over sets of polymers, $\G=\{\g_1,\ldots,\g_n\}$,
to be called {\it polymer collections}, such that each polymer is contained in $\L=\L^a_M$, i.e. it is formed by bonds in $\BB^a_M$.
Moreover, the function $\f(\{\g_1,\ldots,\g_n\})$ implements the hard core condition, that is $\f$ is equal to 1 if none of the 
polymers overlap, and 0 otherwise; the term with $\G=\emptyset$ should be interpreted as 1. 
Note that the r.h.s.\ 
of Eq.~(\ref{2.7b}), like that of~\eqref{2.7a}, is multilinear in the bond variables (i.e., each bond variable 
$\e_b$ appears at most
once in every term in the sum). Hence, plugging Eq.(\ref{2.7b}) into Eq.(\ref{2.58}) and using 
Eq.(\ref{2.6a}), we find:
\begin{multline}
Z(\bsA)=\sum_{\boldsymbol\a}\frac{\t_{\boldsymbol\a}}2\,
(- 2a^{-2})^{M^2}  \Big[ \prod_{b \in \BB^a_M} \cosh ( \beta J+aA_b) \Big]
\Big[\prod_{\{\xx,\yy\}} \cosh^2 \!\big(\tfrac12 \b\l v(\xx-\yy)\big)\Big] \cdot\\
\cdot\int\DD\Phi \sum_{\G\subseteq\L} \f(\G) 
\Big[\prod_{\g\in\G} \tilde\z_G(\g)\Big]e^{S(\Phi,\bsA)}\;,\label{2.7c}\end{multline}
where 
\begin{figure}
\centering
\begin{tikzpicture}[line width=2pt]
\draw[lightgray,line width=0] (-0.5,-0.5) grid (4.5,4.5);
\draw (0,0) node [vertex] {} -| (2,3) node [vertex] {};
\node[below] at (1.5,0) {$S_1$};
\draw (0,0.2) node [vertex] {}  -| (1.8,1) |- (1,4) node [vertex] {};
\node[left] at (1.8,2.5) {$S_2$};
\draw (2.2,2) node [vertex] {} |- (4,4) node [vertex] {};
\node[below] at (3.5,4) {$S_3$};
\end{tikzpicture} 
$\to$
\begin{tikzpicture}[line width = 4pt]
\draw[lightgray,line width=0] (-0.5,-0.5) grid (4.5,4.5);
\draw[lightgray](0,0) -| (2,2);
\draw (2,2) -- (2,3);
\draw[lightgray] (2,3) -- (2,4);
\draw (1,4) -- (4,4);
\node[right] at (2,2.5) {$b_1$};
\foreach \x in {2,3,4} \node[above] at (\x-0.5,4) {$b_\x$};
\node[below] at (1.5,0) {$\g$};
\end{tikzpicture} 
\\
$\to \foreach \x in {1,2,3} { \tanh (\tfrac12 \b \l v_{S_\x} )} \times \foreach \x in {1,2,3,4} { \left( t_{b_\x} + (1-t_{b_\x}^2) aE_{b_\x} \right) }$ 
\caption{Example of a set of strings $\SS=\{S_1,S_2,S_3\}$ and the corresponding coloring of the associated polymer $\g$.  The contribution to $\tilde\z_G$ associated with $\SS$ is given below.}
\end{figure}
\be \tilde\z_G(\g)=\sum_{\substack{ \SS \text{ connected}:\\ 
\g(\SS)=\g}} \prod_{S\in\SS} \tanh( \tfrac12 \b\l v_S) \prod_{b \in \black(\SS)} \left( t_b + a(1-t_b^2) E_b \right) \;,\label{2.8}\ee
This is exactly the sort of expression which is the subject of the standard cluster expansion (for a 
presentation in a convenient form, see~\cite[Chap. 7]{GBG}),
which gives the identity
\be \sum_{\G\subseteq\L}\f(\G) \prod_{\g\in\G} \tilde\z_G (\g) 
= \exp\Big\{\sum_{\G\subseteq\L}\f^T(\G) \prod_{\g\in\G} \tilde\z_G (\g) \Big\} \label{re_exp}\ee
which is valid provided that the sum in the r.h.s.\ is absolutely convergent. In the r.h.s.\
the sum over $\G$ involves polymer collections that include overlapping and even repeated 
polymers (we let $\G(\g)$ be the multiplicity of $\g$ in $\G$) and 
 $\f^T(\{\g_1,\ldots,\g_n\})$ are {\it Mayer's coefficients},  
 which admit the following explicit representation.
Given $\g_1,\ldots,\g_n$, consider
the graph $\GG$ with $n$ nodes, labelled by $1,\ldots,n$, with edges
connecting all pairs $i,j$ such that $\g_i\cap \g_j\neq\emptyset$ ($\GG$ is sometimes called the 
connectivity graph of the collection of polymers $\g_1,\ldots,\g_n$). 
Then one has $\f^T(\emptyset)=0$, $\f^T(\{\g\})=1$ and, for $n>1$:
\be \f^T(\{\g_1,\ldots,\g_n\})= \frac1{\G!}\sum_{C\subseteq \GG}^* (-1)^{\hbox{\text number of
edges in}\ C} , \label{04.4}\ee
where $\G!=\prod_\g \G(\g)!$ and the sum runs over all the connected subgraphs $C$ of $G$ {\it that
visit all the $n$ points} $1,\ldots,n$. In particular, if $n>1$, then $\f^T(\{\g_1,\ldots,\g_n\})=0$ unless $\{\g_1,\ldots,\g_n\}$ is connected. 

Plugging Eq.\eqref{re_exp} into Eq.(\ref{2.7c}) gives
\bea 
Z(\bsA)&=&\sum_{\boldsymbol\a}\frac{\t_{\boldsymbol\a}}2\,
(- 2a^{-2})^{M^2}  \Big[ \prod_{b \in \BB^a_M} \cosh ( \beta J+aA_b) \Big]\label{2.65a}
 \cdot\\
&&\cdot\Big[\prod_{\{\xx,\yy\}} \cosh^2 \!\big(\tfrac12 \b\l v(\xx-\yy)\big)\Big]
\int\DD\Phi e^{S(\Phi,\bsA)+\tilde\VV(\Phi,\bsA)}\;,\nonumber\eea
with 
\be
\tilde \VV(\Phi,\bsA):= 
\sum_{\G\subseteq\L} \f^T(\G)\prod_{\g\in\G}\tilde\z_G(\g), 
\label{Vt_as_zeta}
\ee
provided that the r.h.s.\ of this equation is absolutely convergent, in the following sense: 
$\tilde\VV(\Phi,\bsA)$ is a polynomial in the Grassmann variables $\{\lis H_\xx,H_\xx,$ 
$\lis V_\xx,V_\xx\}$,
whose coefficients can be written as infinite sums over polymer configurations 
$\{\g_1,\ldots,\g_n\}$; by absolute convergence of $\tilde\VV$,
we mean that each of these infinite sums are absolutely convergent.
As far as the 
dependence of $\tilde\VV$ on $\bsA$ is concerned, we note that since $\tilde\z_G(\g)$ is defined through finite sums and 
products of hyperbolic tangents $t_b=\tanh(\b J+aA_b)$, the activities $\tilde\z_G(\g)$ are analytic as functions of each $A_b$. 
However, only the terms at most linear in each $A_b$ can 
contribute to the computation of $\media{\e_{b_1};\cdots;\e_{b_m}}^T_{\b,L}$,
with $b_1,\ldots,b_m$ all distinct. Therefore, to the purpose of computing these truncated expectations,
we can safely replace $\tilde\z_G(\g)$ by its multilinear part in $\{A_b\}_{b\in\BB^a_M}$, namely:
\be \tilde\z_G(\g)\to\z_G(\g):= \sum_{R \subseteq \g} \sum_{Y \subseteq \g}  \z(R,Y;\g) \prod_{b \in R} aE_b \prod_{b \in Y} aA_b\;,\label{2.67}\ee
with
\be
\z(R,Y;\g) := \sum_{\substack{ \SS \text{ connected}:\\ \g(\SS) = \g \\ \black(\SS) \supseteq R \cup Y}}   \,\prod_{S\in\SS} \tanh( \tfrac12 \b\l v_S) 
 \prod_{b \in \black(\SS)} \piecewise{
t, & b \notin R \cup Y \\
1- t^2, & b \in R \D Y \\
-2 t (1 - t^2), & b \in R \cap Y}
\label{eq:zDef}
\ee
where $(R,Y;\g)$ can be associated with a decorated contour
\begin{figure}
\centering
\begin{tikzpicture}[line width = 4pt]
\draw[lightgray,line width=0] (-0.5,-0.5) grid (4.5,4.5);
\draw[lightgray](0,0) -| (2,2);
\draw[orange] (2,2) -- (2,3);
\draw[lightgray] (2,3) -- (2,4);
\draw[red] (1,4) -- (2,4);
\draw[yellow] (2,4) -- (3,4);
\draw (3,4)--(4,4);
\node[right] at (2,2.5) {$b_1$};
\foreach \x in {2,3,4} \node[above] at (\x-0.5,4) {$b_\x$};
\end{tikzpicture}
\caption{A fully colored contour as used in evaluating $W$: $b_1$ is orange, $b_2$ is
red, $b_3$ is yellow and $b_4$ is black; the other unlabeled 
bonds are grey (color online). By definition, the activity of this colored contour is proportional 
to $\left( -2t(1-t^2) \right) aE_{b_1} aA_{b_1} \times (1-t^2)aE_{b_2}\times (1-t^2) aA_{b_3} \times t$. \label{fig:decor_gamma}}
\end{figure}
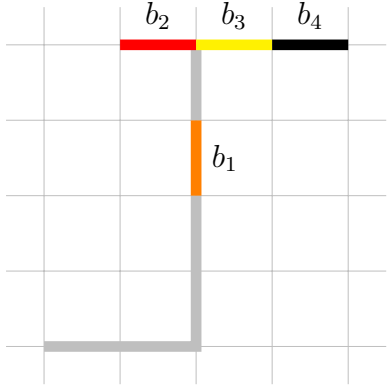
where the bonds in $R$ are drawn red, all bonds in $Y$ are drawn yellow, and all bonds in both are 
drawn orange (see Figure~\ref{fig:decor_gamma}). Similarly, we can replace the factor 
$ \Big[ \prod_{b \in \BB^a_M} \cosh ( \beta J+aA_b) \Big]e^{S(\Phi,\bsA)}$ in Eq.(\ref{2.65a}) by 
its multilinear part in $\{A_b\}_{b\in\BB^a_M}$. Performing these replacements corresponds to 
replacing the generating function in Eq.(\ref{2.65a}) by:
\be Z(\bsA) \to e^{at\sum_bA_b}\,\Xi(\bsA)\;,\label{gen_fun}\ee
where $\Xi(\bsA)$ is the same as in Eq.(\ref{2.51}), with
\bea
V_M(\l)+\VV(\Phi,\bsA)&=& 
\sum_{\G\subseteq\L}\f^T(\G)\prod_{\g\in\G}\z_G(\g)=\nonumber\\
&=&\sum_{R\subseteq\BB^a_M}\sum_{Y\subseteq\BB^a_M}W(R,Y)\Big[\prod_{b\in R}aE_b\Big]\Big[\prod_{b\in Y}aA_b\Big]\;, 
\label{V_as_zeta}
\eea
and $V_M(\l)=W(\emptyset,\emptyset)$. Here the coefficients $W(R,Y)$ are defined as
\be
W(R,Y) =\sum_{\substack{n\ge 0\\ {\boldsymbol\g}=(\g_1,\ldots,\g_n)}} \frac{\G({\boldsymbol\g})!}{n!}\,
 \f^T(\G({\boldsymbol\g}) )
\sum_{\substack{(R_1,\ldots,R_n)\in\PP_{R,{\boldsymbol\g}}\\(Y_1,\ldots,Y_n)\in
\PP_{Y,{\boldsymbol{\g}}}}}\,
\prod_{j=1}^n \z(R_j,Y_j;\g_j)\;, \label{w_expansion}
\ee
where ${\boldsymbol\g}=(\g_1,\ldots,\g_n)$ is an ordered $n$-tuple of polymers, $\G({\boldsymbol\g})$
is the corresponding unordered $n$-tuple and the combinatorial factor $\frac{\G({\boldsymbol\g})!}{n!}$
is used to pass from the summation over $\G$ to the one over ${\boldsymbol\g}$.
Moreover, $\PP_{R,{\boldsymbol\g}}$ is the 
set of ordered $n$-ples of disjoint sets $(R_1,\ldots,R_n)$, such that $R_i\subseteq \g_i$ and 
$\cup_{i=1}^nR_i=R$. Once again, these rewritings are valid provided that the sums entering the 
definition of $W(R,Y)$ are absolutely convergent. Note also that the prefactor $e^{at\sum_bA_b}$ 
in the r.h.s.\ of Eq.(\ref{gen_fun}) is irrelevant to the purpose of computing 
the truncated expectations $\media{\e_{b_1};\cdots;\e_{b_m}}^T_{\b,L}$ with $m\ge 2$, which explains
why we did not insert this prefactor in the statement of Proposition \ref{prop3}.

The absolute convergence of Eq.(\ref{w_expansion}) follows from an important result in the cluster 
expansion \cite[Proposition 7.1.1]{GBG}, which states (in part) that for any positive 
function on polymers which decays as
\be
f(\g) \le \n_0^{2|\g|} e^{-2\k_0 \d(\g)/a}\label{2.70}
\ee
for some $\k_0>0$ and $\n_0$ sufficiently small (here $|\g|$ is the number 
of bonds in $\g$), then
\be
\sup_{b} 
\sum_{\substack{\G=\{\g_1, \ldots,\g_n\} \\ \diam \left( \cup_j \g_j \right) \ge R \\ \bigcup_j \g_j \ni b}}
\left| \f^T(\G) \right| \:  \prod_{\g\in\G} f(\g) \le 2 \n_0 e^{-\frac{\k_0}{2} R/a}\;.\label{2.71}
\ee
Actually, the proof of \cite[Proposition 7.1.1]{GBG} implies a slightly more general bound; namely, given 
$B\subseteq\BB^a_M$,
\be
\sum_{\substack{\G=\{\g_1, \ldots,\g_n\} \\ \cup_j \g_j \supseteq B}}
\left| \f^T(\G) \right| \:  \prod_{\g\in\G} f(\g) \le C_1 \n_0^{c_1|B|} e^{-c_1\k_0 \d(B)/a}\;,
\label{2.72}\ee
for two suitable constants $C_1,c_1>0$.

To apply this, we first note that the generic contribution to the r.h.s.\ of Eq.\eqref{w_expansion}
vanishes unless the collection ${\boldsymbol\g}=(\g_1,\ldots,\g_n)$ is connected and touches all the bonds in
$B:=R\cup Y$. Therefore, 
\be
\begin{split} 
|W(R,Y)|\le  \sum_{\substack{n\ge 0\\ {\boldsymbol\g}=(\g_1,\ldots,\g_n)}} \frac{\G({\boldsymbol\g})!}{n!}\,
| \f^T(\{\g_1,\ldots,\g_n\}) | \,
\sum_{\substack{(R_1,\ldots,R_n)\in\PP_{R,{\boldsymbol\g}}\\(Y_1,\ldots,Y_n)\in
\PP_{Y,{\boldsymbol{\g}}}}}\,\prod_{j=1}^n| \z(R_j,Y_j;\g_j)|\;.\end{split}
\ee
where
\be \sum_{\substack{(R_1,\ldots,R_n)\in\PP_{R,{\boldsymbol\g}}\\(Y_1,\ldots,Y_n)\in
\PP_{Y,{\boldsymbol{\g}}}}}\,\prod_{j=1}^n| \z(R_j,Y_j;\g_j)|\le
\prod_{j=1}^n \Big(\sum_{\substack{R_j \subseteq \g_j\\Y_j \subseteq \g_j}} |\z(R_j,Y_j,\g_j)| \Big)
\le \prod_{j=1}^n f(\g_j)\;,\nonumber\ee
and
\be f(\g):= 4^{|\g|}
\max_{\substack{R \subseteq \g \\ Y \subseteq \g}} |\z(R,Y;\g)|\ee
plays the role of the function $f(\g)$ in the l.h.s.\ of Eq.(\ref{2.72}). In order to use the bound 
Eq.(\ref{2.72}), we need to prove Eq.(\ref{2.70}).
Using the definition Eq.~\eqref{eq:zDef} of $\z(R,Y,\g)$, it is clear that
\bea
f(\g) &\le&4^{|\g|}\sum_{\substack{ \SS \text{ connected}:\\ \g(\SS) = \g}}
   \prod_{S\in\SS}  \tfrac12 \b|\l v_S| \nonumber\\
&\le& 4^{|\g|}\sum_{m\ge \frac{\d(\g)}{aM_0}}\frac1{m!}\Big(\frac{\b|\l|}{2}\Big)^m\Big(\frac12\sum_{b\in\g}
\sum_{S\ni b}^*|v_S|\Big)^m 
   \;,\label{zbound}\eea
where the $*$ on the sum $\sum_{S\ni b}^*|v_S|$ indicates the constraint that $b$ is either the first or the last bond in $S$. By using the normalization 
$\frac12\sum_\xx|v(\xx)|=1$ (which means that $\frac12\sum_{S\ni b}^*|v_S| = 1$) and defining $m_0=m_0(\g)= \frac{\d(\g)}{aM_0}\equiv \frac{|\g|}{M_0}$, we get:
\bea f(\g)&\le&  4^{|\g|}\sum_{m\ge m_0}\frac1{m!}\Big(\frac{\b|\l|\,|\g|}2\Big)^m
\nonumber\\
&\le&4^{|\g|}\frac{|\g|^{m_0}}{m_0!}\Big(\frac{\b|\l|}2\Big)^{m_0}\sum_{m\ge 0}\frac1{m!}
\Big(\frac{\b|\l|\,|\g|}2\Big)^m\label{2.76}\\
&\le&(4 e)^{|\g|}\Big(\frac{\b|\l|}2\Big)^{m_0}e^{\frac12\b|\l|\,|\g|}\equiv
\n_0^{2|\g|}e^{-2\k_0\d(\g)/a}\;,\nonumber
\eea
which proves Eq.(\ref{2.70}) with $\n^2_0=4e^{1+\frac12\b|\l|}\big(\frac{\b|\l|}2)^{\frac1{2M_0}}$ and
$e^{-2\k_0}=
\big(\frac{\b|\l|}2\big)^{\frac1{2M_0}}$. This concludes the proof of the representation 
Eq.(\ref{2.51}). The bound Eq.(\ref{W_decay}) on the decay of the kernels of $\VV(\Phi,\bsA)$ is 
just a restatement of the decay bound on $W(R,Y)$ that we just derived (note that the factor 
$a^{-(2-n-m)}$ comes from the factors $a$ in front of the fields $E_b$ and $A_b$ in Eq.(\ref{V_as_zeta}) and from the definition of $\int d\xx$, which brings along a factor $a^2$; in fact, $\int d\xx=a^2\sum_{\xx\in \L^a_M}$).

All that remains is to prove the stated properties of $V_M(\l) \equiv W(\emptyset,\emptyset)$, where
\bea  W(\emptyset,\emptyset)&=&
\sum_{\G\subseteq\L}\f^T(\G) 
\prod_{\g\in\G} \z(\emptyset,\emptyset;\g) \label{w_empty}\\
&=&\sum_{b \in \BB_M^a} \Big[\sum_{\substack{\G\subseteq\L:\\
\supp\G\ni b}} \frac{\f^T(\G)}
{|\supp\G|} \prod_{\g\in\G} \z(\emptyset,\emptyset;\g)\Big] \;,\nonumber\eea
where $\supp\G=\cup_{\g\in\G}\g$ and $|\supp\G|$ is the number of bonds in $\supp\G$.
The sum in square brackets is independent of $b$, by translation invariance. Moreover, 
it is independent of $a$, as it follows by the definition of $\z(\emptyset,\emptyset;\g)$ and by a 
relabeling of the lattice spacing $a$.
Therefore, $W(\emptyset,\emptyset)$ can be bounded by 
$2M^2 s(\l)$, where $s(\l)$ is a bound on the expression in square brackets. A repetition 
of the argument used for $W(R,Y)$ above implies that $s(\l)$ is of order $\l$, which concludes the proof 
of Proposition \ref{prop3}.
\end{proof}

\subsection{Majorana form of the action}

As we have shown, the non-integrable Ising model under consideration can be expressed in the form of 
an interacting fermionic system, described by the action Eq.(\ref{2.51}), which consists of a leading 
term, $S_t(\Phi)+\sum_b\big[t+(1-t^2)E_b\big]A_b$, plus an {\it interaction}, which vanishes as $\l\to 0$ 
and, in this respect, is ``subdominant". 
Since in the following we want to treat this subdominant term as a perturbation in the vicinity of the 
critical point, it is convenient to use coordinates adapted to the critical modes of the leading term, as 
described in the following.

The Fourier transforms of the Grassmann variables are defined as:
\bea &&
\hat H_\kk := a^2 \sum_{\xx \in a \ZZZ^2_M} e^{i\kk \cdot \xx} H_\xx\;,\qquad \hat{\lis H}_\kk := a^2 \sum_{\xx \in a \ZZZ^2_M} e^{i\kk \cdot \xx} \lis H_\xx\;,\\
&&\hat V_\kk := a^2 \sum_{\xx \in a \ZZZ^2_M} e^{i\kk \cdot \xx} V_\xx\;,\qquad \hat{\lis V}_\kk := a^2 \sum_{\xx \in a \ZZZ^2_M} e^{i\kk \cdot \xx} \lis V_\xx\;,\eea
where, if $\boldsymbol\a=(\a_1,\a_2)$,  then $\kk\in \DD_M^{\boldsymbol\a}:=\{\frac{2\p}{L}(\nn+\frac12{\boldsymbol\a}): \nn\in\mathbb Z^2/ M\mathbb Z^2\}$. The inverse transformation reads:
\bea &&
H_\xx := \frac1{L^2} \sum_{\kk\in\DD_M^{\boldsymbol\a}} e^{-i\kk \cdot \xx} \hat H_\kk\;,\qquad 
\lis H_\xx := \frac1{L^2} \sum_{\kk\in\DD_M^{\boldsymbol\a}} e^{-i\kk \cdot \xx} {\hat{\lis H}}_\kk\;,\\
&&
V_\xx := \frac1{L^2} \sum_{\kk\in\DD_M^{\boldsymbol\a}} e^{-i\kk \cdot \xx} \hat V_\kk\;,\qquad 
\lis V_\xx := \frac1{L^2} \sum_{\kk\in\DD_M^{\boldsymbol\a}} e^{-i\kk \cdot \xx} {\hat{\lis V}}_\kk\;.\eea
If $\hat\Phi_\kk$ is the column vector with components 
${\hat{\lis H}}_\kk, \hat H_\kk, {\hat{\lis V}}_\kk, \hat V_\kk$, respectively, the leading quadratic part of the 
action, $S_t(\Phi)$, can be rewritten in Fourier space as:
\be S_t(\Phi)= \frac{1}{L^2}\sum_{\kk \in \DD_M^{\boldsymbol\a}}\hat\Phi^T_{-\kk}
C_\kk\hat\Phi_{\kk}\;,\ee
where
\be
C_\kk := \frac{a^{-1}}2 \begin{pmatrix}
0 & 1 + te^{-iak_1} & -1 & -1 \\
-1 - t e^{ia k_1}& 0 & 1 & -1 \\
1 & -1 & 0 & 1 + te^{-iak_2} \\
1 & 1 & -1 -  te^{iak_2}  & 0 
\end{pmatrix}\;,
\ee
the inverse of which has the meaning of free propagator of the Grassmann field $\Phi$.
Note that $C^{-1}_\kk$ is singular only at $\kk=\V0$ and $t=\sqrt{2}-1$ (criticality condition). In this 
case, $C_{\V0}$ has two vanishing eigenvalues and two purely imaginary eigenvalues $\pm i a^{-1}\sqrt2$, 
and the corresponding eigenmodes read:
\be
\begin{pmatrix}
\hat{\psi}_{\V0,+}' \\ \hat\psi'_{\V0,-} \\ \hat{\chi}_{\V0,+}' \\ \hat\chi_{\V0,-}'
\end{pmatrix}
= U \hat\Phi_\V0\;,\qquad U = \frac12
\begin{pmatrix}
e^{i\frac{\pi}4} & e^{-i\frac\pi4} & 1 & -i \\
e^{-i \frac\pi4} & e^{i\frac{\pi}4} & 1 & i \\
-e^{i\frac{\pi}4} & -e^{-i\frac\pi4} & 1 & -i \\
-e^{-i\frac\pi4} & -e^{i\frac{\pi}4} & 1 & i
\end{pmatrix}\;.\label{u1}
\ee
The natural variables at the critical point are the ``critical eigenmodes" defined by the unitary 
transformation $U$  in Eq.(\ref{u1}), 
namely 
\be
\begin{pmatrix}
\hat{\psi}_{\kk,+}' \\ \hat\psi_{\kk,-}' \\ \hat{\chi}_{\kk,+}' \\ \hat\chi_{\kk,-}'
\end{pmatrix}
= U \hat\Phi_\kk\;.\ee
For later convenience, we rescale these variables as (where $\o=\pm$)
\be \hat{\psi}_{\kk,\o}'=\frac{i\o}{\sqrt{\p t}}
\hat{\psi}_{\kk,\o}\;,\qquad 
\hat{\c}_{\kk,\o}'=\frac{i\o}{\sqrt{\p t}}
\hat{\c}_{\kk,\o}\;,\label{2.resc}\ee
so that, defining $\hat\psi_{\kk}$ (resp. $\hat\c_\kk$) 
as the column vector with components $\hat{\psi}_{\kk,+},\hat\psi_{\kk,-}$ 
(resp. $\hat{\c}_{\kk,+},\hat\c_{\kk,-}$), we can rewrite:
\be 
S_t (\Phi) = -\frac{1}{4\p L^2} \sum_{\kk \in \DD_M^{\boldsymbol\a}}\Big(
\hat\psi^T_{-\kk} C_\psi(\kk)\hat\psi_{\kk}+\hat\c^T_{-\kk}
C_\c(\kk)\hat\c_{\kk}\Big)+Q(\psi,\c)\;,\ee
where, if $\#=\psi,\c$:
\be
C_\#(\kk) = \begin{pmatrix}
a^{-1}(-i \sin a k_1 +\sin a k_2) & i \s_\#(\kk) \\
-i\s_\#(\kk)& a^{-1}(-i \sin a k_1-\sin a k_2)
\end{pmatrix}\label{2.90a}\ee
and
\bea&& \s_\psi(\kk)=\frac1{a}\Big(\cos ak_1+\cos a k_2-2\frac{\sqrt 2-1}{t}\Big)\;,\\
&&\s_\c(\kk)=\frac1{a}\Big(\cos ak_1+\cos a k_2+2\frac{\sqrt 2+1}{t}\Big)\;.\eea
Moreover,
\bea&&
Q(\psi,\c) = \frac{1}{4\p L^2}\sum_{\kk\in\DD_M^{\boldsymbol\a}}\big( \hat\psi^T_{-\kk} Q(\kk)\hat\c_{\kk}+
\hat\c^T_{-\kk} Q(\kk)\hat\psi_{\kk}\big)\;,\label{2.93a}\\
&& Q(\kk)=
\begin{pmatrix}
a^{-1}(-i\sin a k_1- \sin a k_2) & a^{-1}(i \cos ak_1 - i\cos ak_2 ) \\
a^{-1}(- i \cos ak_1 +i\cos ak_2 ) & a^{-1}(-i\sin ak_1+\sin a k_2 )\end{pmatrix}\;.\nonumber
\eea
In terms of this notation, we can rewrite Eq.(\ref{2.51}) as
\be \Xi(\bsA)
 = \sum_{{\boldsymbol\a}\in\{\pm\}^2}C_{M,{\boldsymbol\a}} \int  P_{\boldsymbol\a}(d\psi)
P_{\boldsymbol\a}(d\c)e^{Q(\psi,\c)+\BB(\psi,\c,\bsA)+\VV(\psi,\c,\bsA)}\;,\label{2.88}\ee
where:
\begin{itemize}
\item The normalization constant $C_{M,{\boldsymbol\a}}$ is defined as 
$$C_{M,{\boldsymbol\a}}=\frac{C_M}{2}(a^2M)^{4M^2} \t_{\boldsymbol\a}\,
 {\NN}_{\psi,{\boldsymbol\a}}{\NN}_{\c,{\boldsymbol\a}}\;,$$
 where: $C_M$ was defined in Eq.(\ref{2.52}); the factor $(a^2M)^{4M^2}$ takes into account the change of variable from the set of Grassmann
variables
$\{\lis H_{\xx},H_\xx,\lis V_\xx,V_\xx\}$ to $\{\hat\psi_{\kk,\o},\hat\c_{\kk,\o}\}$: in fact a computation 
shows that 
$$\int \DD \Phi \Big[\prod_{\kk \in \DD_M^{\boldsymbol\a}}  \hat{\lis H}_\kk \hat H_\kk \hat{\lis V}_\kk \hat
V_\kk\Big] = a^{8M^2} \big[{\rm det}_{\kk,\xx} (e^{i\kk\cdot\xx})\big]^4= (a^2M)^{4M^2}\;;$$
$\t_{\boldsymbol\a}$ was defined after Eq.(\ref{2.4tt});  ${\NN}_{\psi,{\boldsymbol\a}}$ and ${\NN}_{\c,{\boldsymbol\a}}$ are the normalization constants of the two Grassmann gaussian integrations 
$P_{\boldsymbol\a}(d\psi)$ and $P_{\boldsymbol\a}(d\c)$, see next item.
\item The Grassmann gaussian integrations 
$P_{\boldsymbol\a}(d\psi)$, $P_{\boldsymbol\a}(d\c)$ are defined as
\bea && 
P_{\boldsymbol\a}(d\psi):=\frac1{\NN_{\psi,{\boldsymbol\a}}}\Big[\prod_{\kk\in\DD_M^{\boldsymbol\a}}\prod_{\o=\pm}d\hat\psi_{\kk,\o}\Big] \exp\Big\{- \frac{1}{4\p L^2} \sum_{\kk \in 
\DD_M^{\boldsymbol\a}} \hat\psi^T_{-\kk} C_\psi(\kk)\hat\psi_{\kk}\Big\}\;,\nonumber\\
&&P_{\boldsymbol\a}(d\c):=\frac1{\NN_{\c,{\boldsymbol\a}}}\Big[\prod_{\kk\in\DD_M^{\boldsymbol\a}}\prod_{\o=\pm}d\hat\c_{\kk,\o}\Big] \exp\Big\{- \frac{1}{4\p L^2} \sum_{\kk \in 
\DD_M^{\boldsymbol\a}} \hat\c^T_{-\kk} C_\c(\kk)\hat\c_{\kk}\Big\}\;,\nonumber\eea
and $\NN_{\psi,{\boldsymbol\a}}, \NN_{\c,{\boldsymbol\a}}$ are two normalization constants,
fixed in such a way that $\int P_{\boldsymbol\a}(d\psi)=\int P_{\boldsymbol\a}(d\c)=1$.
\item The source term $\BB(\psi,\c,\bsA)$ is the rewriting of $(1-t^2)({\boldsymbol{E}}, \bsA)$ 
in terms of the new variables, namely 
\bea  && \BB(\psi,\c,\bsA)=-i\,\frac{1-t^2}{2 t}\int \frac{d\xx}{2\p}\cdot\label{2.en}\\
&&\Big[A_{\xx,1}
(\psi_{\xx,+}-i\psi_{\xx,-}-\c_{\xx,+}+i\c_{\xx,-})\t_1
(-i\psi_{\xx,+}+\psi_{\xx,-}+i\c_{\xx,+}-\c_{\xx,-})\nonumber\\
&&+A_{\xx,2}(\psi_{\xx,+}-\psi_{\xx,-}+\c_{\xx,+}-\c_{\xx,-})
\t_2(\psi_{\xx,+}+\psi_{\xx,-}+\c_{\xx,+}+\c_{\xx,-})\Big]\;,\nonumber
\eea
where $\t_j$ is the translation operator that shifts by one lattice step 
the argument of the field which it acts on: $\t_j\psi_{\xx,\o}=\psi_{\xx+a\hat{\bf e}_j,\o}$ and similarly 
for $\c$.
\item $\VV(\psi,\c,\bsA)$ is the rewriting of $\VV(\Phi,\bsA)$ in terms of the new variables. It is easy to check that 
its kernels satisfy the same decay estimates as those of $\VV(\Phi,\bsA)$, see Eq.(\ref{W_decay}) 
in Proposition \ref{prop3}. 
\end{itemize}
As mentioned in the introduction, we are concerned with a scaling limit such that 
we take the thermodynamic limit $M\to\infty$ first, keeping $\b=\b(a)\neq \b_c(\l)$, and then (simultaneously) 
$a\to 0$ and $\b\to\b_c(\l)$. In doing so, the resulting multi-point energy correlations are insensitive 
to the Grassmann boundary conditions, labeled by the 
four possible values of $\boldsymbol\a$; this is true both in the $\l=0$ and in the $\l\neq 0$ case, see
\cite[Appendix G]{M04}. Therefore, the multi-point energy correlation functions in the specific
scaling limit that we consider are the same as those computed from the following generating function:
\be \Xi_{-,-}(\bsA)= \int  P(d\psi)P(d\c)
e^{Q(\psi,\c)+\BB(\psi,\c,\bsA)+\VV(\psi,\c,\bsA)}\;, \label{2.90}\ee
where $P(d\psi)$ is a shorthand for $P_{-,-}(d\psi)$ and similarly for $P(d\c)$. For future reference, 
let us note that the propagator of the $\psi$ and $\c$ fields associated with the gaussian 
integrations $P(d\psi)$ and $P(d\c)$ are given by, if $D_a^\pm(\kk)=a^{-1}(i\sin ak_1\pm
\sin ak_2)$ and $|D_a(\kk)|^2=-D_a^+(\kk)D_a^-(\kk)$,
\bea  g^\psi_{\o,\o'}(\xx-\yy)&=&\int P(d\psi) \psi_{\xx,\o}\psi_{\yy,\o'}\\
&=&
\frac{2\p}{L^2}\sum_{\kk\in\DD_M}\frac{e^{-i\kk(\xx-\yy)}}{|D_a(\kk)|^2+[\s_\psi(\kk)]^2}
\begin{pmatrix} D_a^+(\kk)& {i\s_\psi(\kk)}\label{2.54}\\
-{i\s_\psi(\kk)}& D_a^-(\kk)\end{pmatrix}_{\!\!\o,\o'}\nonumber
\eea
and
\bea g^\c_{\o,\o'}(\xx-\yy)&=&\int P(d\c) \c_{\xx,\o}\c_{\yy,\o'}\label{2.gchi}\\
&=&
\frac{2\p}{L^2}\sum_{\kk\in\DD_M}\frac{e^{-i\kk(\xx-\yy)}}{|D_a(\kk)|^2+[\s_\c(\kk)]^2}
\begin{pmatrix} D_a^+(\kk)& {i\s_\c(\kk)}\\
-{i\s_\c(\kk)}& D_a^-(\kk)\end{pmatrix}_{\!\!\o,\o'}\nonumber
\eea
where $\DD_M$ is a shorthand for $\DD_M^{-,-}$. \\

{\bf Remark.} We define the {\it unperturbed scaling limit} as follows. 
Let $\b(a)$ be fixed in such a way that, if $t=t(a)=\tanh\big(\b(a)J\big)$ and $t_c(0)=\sqrt2-1$, 
\be\frac{t(a)-t_c(0)}{t(a)}=\frac{a\s^0(a)}2\;,\label{2.unsc}\ee
where $\s^0(a)\neq 0$ for all 
$a\neq 0$ and $\lim_{a\to 0} \s^0(a)=m^*$. The limit as $a\to 0$ with $t(a)$ fixed in this way 
will be referred to as the unperturbed scaling limit. 
The explicit form of the 
propagator of the $\psi$ field shows that in the unperturbed scaling limit it behaves as:
\be g^\psi(\xx)\to \mathfrak g^0(\xx):=\int \frac{d\kk}{2\p}\,\frac{e^{-i\kk\,\xx}}{
\kk^2+(m^*)^2}
\begin{pmatrix} ik_1+k_2 & im^*\\-im^*&ik_1-k_2\end{pmatrix}\;.\ee
Note that the limiting propagator $\mathfrak g^0(\xx)$ is normalized in such a way that in the massless
case, $m^*=0$, it reduces to Eq.(\ref{1.11l}).
In the same limit, the propagator of the $\c$ field $g^\c(\xx)$ tends to zero for every fixed 
$\xx\in\mathbb R^2$; moreover, the combination $a^{-1}\int d\xx\, g^\c(\xx)$ tends 
to $-c_\c\s_2$, where $c_\c=\lim_{a\to 0} a\s_\c(\V0)=\frac{2}{t_c(0)}(t_c(0)+\sqrt2+1)$ and 
$\s_2=\begin{pmatrix} 0&-i\\i&0\end{pmatrix}$ 
is the second Pauli matrix. In this sense, as $a$ tends to zero:
\be g^\c(\xx)\simeq -a c_\c \d(\xx) \s_2\;,\label{2.chisca}\ee
where $\d(\xx)$ is the Dirac delta function.
\vskip.2truecm
The quadratic coupling $Q(\psi,\c)$ in 
Eq.(\ref{2.90}) can be eliminated by the linear change of variables 
\be \hat \c_\kk\to
\hat\c_\kk +C_\c^{-1}(\kk)Q(\kk)\hat\psi_\kk\;,\label{2.lin}\ee
which leaves $\Xi_{-,-}(\bsA)$ invariant. Note that for small $\kk$, the kernel $C_\c^{-1}(\kk)Q(\kk)$ 
associated to this transformation is small, namely $C_\c^{-1}(\kk)Q(\kk)=O(a|\kk|)$; that is, from a 
dimensional point of view, the action of $C_\c^{-1}(\kk)Q(\kk)$ on $\hat\psi_\kk$ is the same as 
the action of the differential operator $a\partial_\xx$.
After the transformation Eq.(\ref{2.lin}) we can rewrite:
\be \Xi_{-,-}(\bsA)= \frac{\lis \NN_\psi}{\NN_\psi}\int  \lis P(d\psi)P(d\c)
e^{\lis \BB(\psi,\c,\bsA)+\lis \VV(\psi,\c,\bsA)}\;,\label{2.91}\ee
where $\lis \BB(\psi,\c,\bsA)$ and $\lis\VV(\psi,\c,\bsA)$ are the rewritings of 
$\BB(\psi,\c,\bsA)$ and $\VV(\psi,\c,\bsA)$, respectively, in terms of the new variables;
moreover, the gaussian integration $\lis P(d\psi)$ is defined as
\be \lis P(d\psi)=\frac1{\lis \NN_{\!\psi}}\Big[\prod_{\kk\in\DD_M}\prod_{\o=\pm}d\hat\psi_{\kk,\o}\Big]
\exp\Big\{ -\frac{1}{4\p L^2} \sum_{\kk \in 
\DD_M} \hat\psi^T_{-\kk} \lis C_\psi(\kk)\hat\psi_{\kk}\Big\}\;,\label{pdpsiN}\ee
where the normalization constant $\lis \NN_{\psi}$
is chosen in such a way that $\int \lis P(d\psi)=1$ and
\be\lis C_\psi(\kk)=C_\psi(\kk)-Q(\kk)C_\c^{-1}(\kk)Q(\kk).\ee
{\bf Remarks.}\begin{enumerate}
\item It is easy to check that the change of variables Eq.(\ref{2.lin}) does not affect the bounds on the 
kernels of the effective potential; i.e., the kernels of $\lis\BB$ and $\lis\VV$ satisfy the same bounds 
as the kernels in Proposition \ref{prop3}, see Eq.(\ref{W_decay}).
\item The correction $\lis C_\psi(\kk)-C_\psi(\kk)$ is small at small $\kk$, i.e., 
$$Q(\kk)C_\c^{-1}(\kk)Q(\kk)=O(a|\kk|^2)\;.$$
In particular, this means that the unperturbed scaling limit of the propagator $\lis g^\psi_{\o,\o'}(\xx)=\int
\lis P(d\psi)\psi_{\xx,\o}\psi_{0,\o'}$, in the sense of the remark after Eq.(\ref{2.gchi}), is the same as the one of $g^\psi(\xx)$, that is, it is equal to 
$\mathfrak g^0(\xx)$. 
\item The representation Eq.(\ref{2.91}) is valid also when $\l=0$, in 
which case $\lis\VV=0$ and the scaling limit of the 
energy-energy correlations can be computed explicitly. In fact, using the explicit forms of
$\BB(\psi,\c,\bsA)$ in Eq.(\ref{2.en}), of the linear transformation 
Eq.(\ref{2.lin}) and of the gaussian integrations $\lis P(d\psi)P(d\c)$, we find
that for all $m$-tuples of distinct points $\xx_1,\ldots,\xx_m$, $m\ge 2$ the unperturbed scaling 
limit (in the sense of the remark after Eq.(\ref{2.gchi})) of the energy correlations is:
\bea && \lim_{a\to 0}\lim_{L\to\infty}\media{\e_{\xx_1,j_1};\cdots;\e_{\xx_m,j_m}}^T_{\b(a),L}=\nonumber\\
&&= \lim_{a\to 0}\lim_{L\to\infty} a^{-2m}\frac{\dpr^m}{\dpr A_{\xx_1,j_1}\cdots\dpr A_{\xx_m,j_m}}
 \Xi_{-,-}(\bsA)\big|_{\bsA=\V0}=\label{2.101}\\
 &&=\Big(\frac{-i}\p\Big)^m\EE^T_{\mathfrak g^0}(\psi_{\xx_1,+}\psi_{\xx_1,-};\cdots;
\psi_{\xx_m,+}\psi_{\xx_m,-})\;,\nonumber\eea
where in the last line $\EE^T_{\mathfrak g^0}$ indicates the truncated expectation with respect to the 
gaussian fermionic integration $P_{\mathfrak g^0}(d\psi)$ with propagator $\mathfrak g^0$; 
i.e., given $n$ functions $X_1,\ldots,X_n$ of the Grassmann variables $\psi$: 
\be \EE_{\mathfrak g^0}^T(X_1;\cdots;X_n)=\frac{\dpr^n}{\dpr\l_1\cdots\dpr\l_n}\log\int P_{\mathfrak g^0}
(d\psi)e^{\l_1X_1+\cdots+\l_nX_n}\big|_{\l_i=0}\label{2.trun}\; ,\ee
and a similar definition is valid for a more general gaussian fermionic integration.
In the last line of Eq.(\ref{2.101}), one may think of $(-i/\p)\psi_{\xx,+}\psi_{\xx,-}$ as the scaling 
limit of the Grassmann operator coupled to $A_{\xx,j}$ in the source term $\lis\BB$, that is $E_{\xx,j}$
written in terms of $\psi,\c$. From its definition,
see Eq.(\ref{2.en}), it is apparent that, for every finite $a$, such Grassmann operator also includes terms 
of the form $\c_{\xx,+}\c_{\xx,-}$ or $\psi_{\xx,\o}a\dpr_\xx\psi_{\xx,\o}$: however, the
correlations among such bilinears at distinct points vanish in the scaling limit. 

The truncated expectation in the last line of Eq.(\ref{2.101}) can be graphically represented (and
explicitly computed) in terms of loop diagrams; correspondingly Eq.(\ref{2.101}) can be rewritten as
Eq.(\ref{1.maina}).
\end{enumerate}

\subsection{The perturbed scaling limit and the temperature counterterm}

The last step that is convenient to perform before setting up the multiscale analysis 
that we will use to compute Eq.(\ref{2.91}), is to properly fix the location of the singularity of 
$\lis P(d\psi)$. Note, in fact, that the propagator associated with $\lis P(d\psi)$ 
is singular at $t=t_c(0)=\sqrt2-1$ (and $\kk=\V0$), while we know that the location of the singularity 
changes in the presence of the interaction, moving to $t=t_c(\l)=\tanh(\b_c(\l)J)$, where $\b_c(\l)$ is the 
interacting inverse critical temperature computed in \cite{PS}, which will also be derived below. 
Therefore, it is convenient to rewrite the mass $\s_\psi({\bf 0})$ appearing in the propagator of the 
$\psi$ field as (recalling that $t_c(0)=\sqrt2-1$)
\be \s_\psi({\bf 0})=\frac2{a}\frac{t-t_c(0)}{t}=\s+4\p a^{-1}\n\;,\label{2.109}\ee
where the mass
\be \s=\s(a)=\frac{2}{a}\frac{t_c(0)}{t_c(\l)}\frac{t-t_c(\l)}{t}\label{2.mass}\ee
vanishes at the interacting critical point and is of order 
$1$ with respect to $a$ as $a\to 0$; we assume that $\s(a)\neq 0$ for every finite $a>0$ and 
in the scaling limit $\lim_{a\to 0}\s(a)=m^*$, with $m^*\in\mathbb R$. On the other hand,
the constant 
\be \n=\n(\l)=\frac{t_c(\l)-t_c(0)}{2\p t_c(\l)}\label{2.counter}\ee
should be thought of as a {\it counterterm}
that will be used below to fix the interacting critical temperature. The rewriting Eq.(\ref{2.109})
induces an analogous rewriting at $\kk\neq\V0$, i.e., $\s_\psi(\kk)=\s(\kk)+4\p a^{-1}\n$, with 
$\s(\kk)=a^{-1}(\cos ak_1+\cos a k_2-2)
+\s$. Correspondingly we decompose the inverse propagator of the $\psi$ field as
$\lis C_\psi(\kk)=C_\s(\kk)-4\p a^{-1}\n\s_2$, where $\s_2=\begin{pmatrix} 0&-i\\i&0\end{pmatrix}$ 
is the second Pauli matrix and 
\be C_\s(\kk)= \begin{pmatrix} -D^-_a(\kk)& i\s(\kk)\\-i\s(\kk)& -D^+_a(\kk)
\end{pmatrix}-Q(\kk)C_\c^{-1}(\kk)Q(\kk)\;,\label{2.109b}\ee
which induces the following representation for the generating function:
\be \Xi_{-,-}(\bsA)= \NN_\s \int P_{\s}(d\psi)P(d\c)
e^{\lis \BB(\psi,\c,\bsA)+\lis \VV(\psi,\c,\bsA)+a^{-1}\n\int d\xx\, \psi_{\xx}\s_2\psi_{\xx}}\;,\label{2.110}\ee
where $P_\s(d\psi)$ is the (normalized) gaussian integration associated to the matrix $C_\s$ 
and $\NN_\s$ is a suitable normalization constant. Eq.(\ref{2.110}) will be the starting point for 
the multiscale analysis discussed below. But before that, let us discuss a few relevant 
symmetry properties of the fermionic action.

\subsection{Symmetries in the Grassman representation}\label{sec2.sym}

The integration $\prod_{\kk,\o} d\hat\psi_{\kk,\o}d\hat\c_{\kk,\o}$, the quadratic contributions to 
the Grassmann action 
$ -(4\p L^2)^{-1}\sum_{\kk \in \DD_M^{\boldsymbol\a}} \hat\psi^T_{-\kk} C_\s(\kk)\hat\psi_{\kk}$ and
$ -(4\p L^2)^{-1}\sum_{\kk \in \DD_M^{\boldsymbol\a}} \hat\psi^T_{-\kk} C_\c(\kk)\hat\psi_{\kk}$, 
as well as the source and the interaction terms $\lis\BB(\psi,\c,\bsA)$, 
$\lis\VV(\psi,\c,\bsA)$ and $a^{-1}\n\int d\xx\, \psi_{\xx}\s_2\psi_{\xx}$ are each separately unchanged under any 
of the following substitutions (here we indicate $A_{\xx,j}$ by $A_{b}$, where $b$ is identified with the unordered pair $b=\{\xx,\xx+a\hat{\bf e}_j)\}$)
\begin{enumerate}
\item $ \psi_{\xx, \o}  \to i\o  \psi_{-\xx,\o}$,  $ \c_{\xx, \o}  \to i\o \c_{-\xx,\o}$, $A_{b}\to A_{-b}$ where, if $b=\{\xx,\xx'\}$, then $-b=\{-\xx,-\xx'\}$ (parity)
\item $\psi_{\xx,\o} \to \o e^{ i\o\frac{\pi}4}  \psi_{R\xx,-\o}$, $ \c_{\xx,\o} \to -\o e^{i\o\frac{\pi}{4}} \c_{R\xx,-\o}$, $A_{b} \to A_{Rb}$ with $R(x_1,x_2) = (-x_2,-x_1)$ 
and $Rb=R\{\xx,\xx'\}=\{R\xx,R\xx'\}$ (diagonal reflection)
\item $\psi_{\xx, \o}  \to i\psi_{F\xx,-\o}$,  $\c_{\xx, \o}  \to  i \c_{F\xx,-\o}$, $A_b \to A_{Fb}$ with $F(x_1,x_2)=(-x_1,x_2)$ and $Fb=F\{\xx,\xx'\}=\{F\xx,F\xx'\}$ (orthogonal reflection)
\item $\psi_{\xx, \o}  \to   \psi_{\xx,-\o}$,  $\c_{\xx, \o}  \to \c_{\xx,-\o}$, $c \to c^*$, $A_{b}\to A_b$ for all complex coefficients (complex conjugation)
\end{enumerate}
The invariance of the integration is only a matter of checking that the correct sign is produced, as indeed it is.  

To check the invariance of the quadratic terms, we consider the most general possible form which will be 
invariant under these transformations, as this will be helpful in the further analysis.  A general quadratic form in the $\psi$ fields is
\be
\sum_{\kk \in \DD_M} \hat\psi^T_{-\kk} C(\kk)\hat\psi_{\kk} =  
\sum_{\kk \in \DD_M} \hat\psi^T_{-\kk} 
\begin{pmatrix}
a (\kk) & b(\kk) \\
c(\kk) & d(\kk)
\end{pmatrix}
\hat\psi_{\kk}
\ee
with $C(\kk)=-C^T(-\kk)$, that is $a(\kk)=-a(-\kk)$, $d(\kk)=-d(-\kk)$ and 
$b(\kk)=-c(-\kk)$. This matrix transforms respectively as
\begin{enumerate}
\item 
\be
\begin{pmatrix}
a (\kk) & b(\kk) \\
c(\kk) & d(\kk)
\end{pmatrix}
\to 
\begin{pmatrix}
a (\kk) & -c(\kk) \\
-b(\kk) & d(\kk)
\end{pmatrix}
\ee

\item
\be
\begin{pmatrix}
a (\kk) & b(\kk) \\
c(\kk) & d(\kk)
\end{pmatrix}
\to 
\begin{pmatrix}
-id(R\kk) & -c(R\kk)\\
-b(R\kk) & ia(R\kk)
\end{pmatrix}
\ee

\item
\be
\begin{pmatrix}
a (\kk) & b(\kk) \\
c(\kk) & d(\kk)
\end{pmatrix}
\to 
\begin{pmatrix}
-d(F\kk) & -c(F\kk)\\
-b(F\kk) & -a(F\kk)
\end{pmatrix}
\ee

\item
\be
\begin{pmatrix}
a (\kk) & b(\kk) \\
c(\kk) & d(\kk)
\end{pmatrix}
\to 
\begin{pmatrix}
d^{*}(-\kk) & c^{*}(-\kk)\\
b^{*}(-\kk) & a^{*}(-\kk)
\end{pmatrix}
\ee
\end{enumerate}
Therefore, the quadratic form is invariant under these symmetries iff $a(\kk)=-a(-\kk)=-id(R\kk)=-d(F\kk)=-d^*(\kk)$ and $b(\kk)=b(-\kk)=-c(\kk)=b(R\kk)=b(F\kk)=-b^*(\kk)$. 
Exactly the same is true of a quadratic term in $\c$.  We note that $C_\s(\kk)$ and $C_\c(\kk)$ as 
expressed in Eqs.(\ref{2.109b}) and (\ref{2.90a}), as well as the counterterm 
$a^{-1}\n\int d\xx\,\psi_{\xx}\s_2\psi_{\xx}$, are indeed invariant under these transformation.  
In passing, let us note that any quadratic form compatible with these symmetries is given to first order in $\kk$ by 
\be
Z
\begin{pmatrix}
D_a^-(\kk) & -i\s \\
i\s & D_a^+(\kk)
\end{pmatrix}\label{2.form}
\ee
with $Z$ and $\s$ two \emph{real} parameters.  

To examine the effect on $\lis\VV$, we note that the energy density $E_b$ (which is the Grassmann
bilinear coupled to $A_b$ in the source term $\lis\BB(\psi,\c,\bsA)$) is covariant under the four 
symmetries above: 
$E_b\overset{(1)}{\to}E_{-b}$,
$E_b\overset{(2)}{\to}E_{Rb}$,
$E_b\overset{(3)}{\to}E_{Fb}$,
$E_b\overset{(4)}{\to}E_{b}$. This suffices to show that the source term $\lis\BB$ is invariant under 
these transformations; furthermore we note that $\lis\VV$ is given as a polynomial in the $E_b$ and $A_b$ with real 
coefficients, symmetric under parity, diagonal reflections and orthogonal reflections, as follows from the 
construction in Sec.\ref{sec2.1}, and therefore is also invariant. 

\section{Multiscale integration}\label{sec3}

We now want to compute $\Xi_{-,-}(\bsA)$ as expressed by Eq.(\ref{2.110}). The strategy
is to integrate the Grassmann functional integral step by step, in an inductive fashion. The outcome
is a multiscale expansion that has been described in great detail in several papers, see e.g.\ 
\cite{GM01, Gi, Ma} for some recent reviews. A self-contained presentation is also 
described below. From now on $C,C',c,c',\ldots$, indicate universal positive constants,
whose specific values may change from line to line.

\subsection{The integration of the $\c$ field}\label{3.chi}

The first step simply consists in integrating out the $\c$ field, after which we rewrite: 
\be \Xi_{-,-}(\bsA)= e^{|\L|\wt F_N+\SS^{(N)}(\bsA)
}\int P_\s(d\psi)e^{\wt \VV^{(N)}(\psi,\bsA)+
\BB^{(N)}(\psi,\bsA)+a^{-1}\n\int d\xx\,\psi_\xx\s_2\psi_\xx}\;,\label{3.1}\ee
where $\L$ is a shorthand for $\L^a_L$ and
\be e^{|\L|\wt F_N+\SS^{(N)}(\bsA)
+\wt \VV^{(N)}(\psi,\bsA)+\BB^{(N)}(\psi,\bsA)}
=\NN_\s\int P(d\c)e^{\lis\VV(\psi,\c,\bsA)+\lis\BB(\psi,\c,\bsA)}\;.\ee
Recall that the label $N$ indicates the scale of the lattice spacing, which is
equal to $a=\ell_02^{-N}$, with $\ell_0$ the unit macroscopic length (which will be set equal to 1 in the following).
The {\it effective potential} $\wt\VV^{(N)}$ 
on scale $N$ can be computed in terms of truncated expectations, namely
\bea&& |\L|\wt F_N-\log\NN_\s+\SS^{(N)}(\bsA)
+\wt\VV^{(N)}(\psi,\bsA)+\BB^{(N)}(\psi,\bsA)=\label{3.3}\\
&&+\sum_{s\ge 1}\frac1{s!}\EE^T_\c(
 \underbrace{\lis\VV(\psi,\c,\bsA)+\lis\BB(\psi,\c,\bsA);
\cdots;\lis\VV(\psi,\c,\bsA)+\lis\BB(\psi,\c,\bsA)}_{s\ {\rm times}})\nonumber\eea
where $\EE^T_\c$ is defined in way analogous to Eq.(\ref{2.trun}), with $P_{\mathfrak g^0}(d\psi)$ 
replaced by $P(d\c)$. The normalization constant $\wt F_N$ in the left hand side is defined 
in such a way that
$\SS^{(N)}(0)=
\wt\VV^{(N)}(0,\bsA)=0$, while $\BB^{(N)}(\psi,\bsA)$ collects the contributions 
that are quadratic in $\psi$ and linear in $\bsA$. Note that $\EE^T_\c$ is a multilinear operator of its 
arguments, so that each term 
in the r.h.s.\ can be computed by expanding $\lis\VV(\psi,\c,\bsA)+\lis\BB(\psi,\c,\bsA)$ into monomials in 
$\c$ and then by acting with $\EE^T_\c$ on each monomial separately. The second line 
of Eq.(\ref{3.3}) can be conveniently represented graphically as in Fig.\ref{fig6.1}.
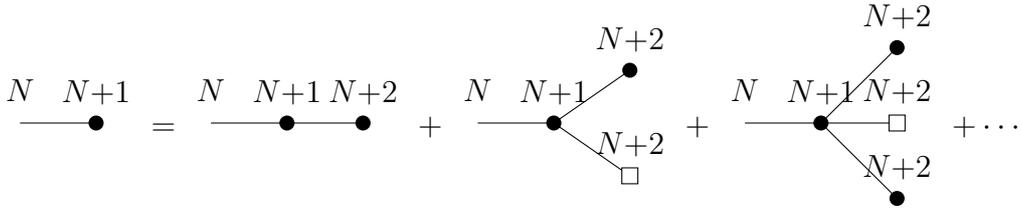
\begin{figure}[ht]
\centering
\begin{tikzpicture}[baseline=-0.4em]
\draw (0,0) node[label=above:$N$] {} -- (1,0) node[vertex,label=above:$N{+}1$] {};
\end{tikzpicture}
=
\begin{tikzpicture}[baseline=-0.4em]
\draw (0,0) node[label=above:$N$] {} -- (1,0) node[vertex,label=above:$N{+}1$] {} -- (2,0) node[vertex,label=above:$N{+}2$] {};
\end{tikzpicture}
+
\begin{tikzpicture}[baseline=-0.4em]
\draw (0,0) node[label=above:$N$] {} -- (1,0) node[vertex,label=above:$N{+}1$] {} -- (2,0.7) node[vertex,label=above:$N{+}2$] {};
\draw (1,0) -- (2,-0.7) node [specialEP,label=above:$N{+}2$] {};
\end{tikzpicture}
+
\begin{tikzpicture}[baseline=-0.4em]
\draw (0,0) node[label=above:$N$] {} -- (1,0) node[vertex,label=above:$N{+}1$] {} -- (2,1) node[vertex,label=above:$N{+}2$] {};
\draw (1,0) -- (2,0) node [specialEP,label=above:$N{+}2$] {};
\draw (1,0) -- (2,-1) node [vertex,label=above:$N{+}2$] {};
\end{tikzpicture}
$+\cdots$
\caption{The graphical representation of the second line of Eq.(\ref{3.3}).\label{fig6.1}}
\end{figure}
The tree in the l.h.s., consisting of a single horizontal branch, connecting the 
left node 
(called the {\it root} and associated to the {\it scale label} $N$)
with a black dot on scale $N+1$, represents the l.h.s.\ of Eq.(\ref{3.3}). In the r.h.s., 
the sum of all the terms with $s$ final points represents the term of order $s$ in 
the r.h.s.\ of Eq.(\ref{3.3}): a scale label 
$N$ is attached to the leftmost node (the root); a scale label 
$N+1$ is attached to the central node (corresponding to the action of $\EE^T_\c$, where $\c$ is thought 
of as the ``field on scale $N {+} 1$");
a scale label $N{+}2$ is attached to the $s$ rightmost nodes (``endpoints"), which can be 
either ``normal", in which case we draw them as black dots and associate them with $\lis\VV$, 
or ``special", in which case we draw them as open squares and associate them with $\lis\BB$.
We denote by $\TT^{(N)}_{N;n,m}$ the set of trees with $n$ normal and $m$ special endpoints.
Moreover, we denote the central node by $v_0$ and the endpoints 
by $v_1,\ldots,v_s$; there is a natural partial ordering on the trees from left to right, and 
we shall denote this ordering by $v_0<v_i$. The $\lis\VV$ and $\lis\BB$ associated with the endpoints 
consist of a sum of several different field 
monomials, which can be distinguished by assigning sets $P_{v_i}$
of field labels $f$. Each field label consist of a specification of the field type (either $A$, or $\psi$, or $\c$), 
and of the $\o$ indices of the Grassmann fields and the $j$ indices of the $A$ fields.
We shall also 
denote by $P_{v_0}$ the set of field labels associated to the $A$ and $\psi$ fields, which can 
be thought of as the fields that survive the integration $P(d\c)$. Finally, given a tree $\t\in\TT^{(N)}_{N;n,m}$ and 
the sets $P_{v_0},\ldots, P_{v_s}$, let ${\bf P}=\{P_{v_0},\ldots,P_{v_s}\}$ and let ${\cal P}_{\t}$ be the set spanned by ${\bf P}$, so that 
the r.h.s.\ of Eq.(\ref{3.3}) can be rewritten as
\bea &&\sum_{\substack{n,m\ge 0\\n+m\ge 1}}
\sum_{\t\in\TT^{(N)}_{N;n,m}}\VV^{(N)}(\t,\psi,\bsA)\;,\qquad {\cal V}^{(N)}(\t,\psi,\bsA)=\sum_{{\bf P}\in{\cal P}_{\t}}
{\cal V}^{(N)}(\t,{\bf P})\;,\nonumber\\
&&{\cal V}^{(N)}(\t,{\bf P})=\sum_{\o_{v_0}}
\int d\xx_{v_0}\psi_{P^\psi_{v_0}}A_{P^A_{v_0}}
K_{\t,{\bf P}}^{(N)}(\xx_{v_0})\;,\label{3.43a}\eea
where $\xx_{v}=\cup_{w\ge v}\cup_{f\in P_{w}}\{\xx(f)\}$, 
$\o_{v}=\cup_{w\ge v}\cup_{f\in P_{w}}\{\o(f)\}$, and, if $P_{v_0}^\psi$ and 
$P_{v_0}^A$ are the subsets of $P_{v_0}$ collecting the fields of type $\psi$ and $A$, respectively,
\be\psi_{P_{v_0}^\psi}=\prod_{f\in P_{v_0}^\psi}\psi_{\xx(f),\o(f)}\;,
\qquad A_{P_{v_0}^A}=\prod_{f\in P_{v_0}^\psi}A_{\xx(f),j(f)} \;.
\label{2.44s}\ee
Moreover, if $Q_{v_i}=P_{v_0}\cap P_{v_i}$, noting that $P_{v_i}\setminus Q_{v_i}=P_{v_i}^\c$,
\be K_{\t,{\bf P}}^{(N)}(\xx_{v_0})=\frac1{s!}
\prod_{i=1}^{s} K^{(N+1)}_{v_i}(\xx_{v_i})\; \;\EE^T_{\c}\big(\c_{P_{v_1}^\c};\cdots;\c_{P_{v_{s}}^\c}\big)\;,\label{2.45s}\ee
where $\c_{P_{v_i}^\c}$ has a definition
similar to Eq.(\ref{2.44s}) and 
$K^{(N+1)}_{v_i}(\xx_{v_i})$ is the kernel of the monomial labeled by $P_{v_i}$, whose decay 
properties follow from Proposition~\ref{prop3}. 

The action of the truncated expectation in the r.h.s.\ of Eq.(\ref{2.45s}) can be computed in terms of a tree interpolation formula, 
originally due to Battle, Brydges and Federbush \cite{BF,B,BrF} and re-derived in several 
review papers, see e.g.\ \cite{GM01,Gi}:
\be\EE^T_{\c}\big(\c_{P_{v_1}^\c};\cdots;\c_{P_{v_{s}}^\c}\big)=
\!\sum_{T\in{\bf T}}\!\a_T\prod_{\ell\in T}g^{(N+1)}_{\ell}
\!\int\! P_T({\bf t})\,\Pf G^{N+1,T}({\bf t})\;,\label{3.pfe}\ee
where:\begin{itemize}
\item the first sum runs over set of lines forming a {\it spanning tree} between
the ``boxes" or ``clusters" $v_1,\ldots,v_s$, i.e., $T$ is a set
of lines that becomes a tree if one identifies all the points
in the same clusters;
\item $\a_T$ is a sign (irrelevant for the subsequent bounds);
\item $g^{(N+1)}_\ell$ is a shorthand for $g^\c_{\o(\ell),\o'(\ell)}(\xx(\ell)-\xx'(\ell))$, where
$\o(\ell), \o'(\ell)$ and $\xx(\ell),\xx'(\ell)$ are the $\o$ and $\xx$ indices associated to 
the two ends of the line $\ell$, which should be thought  as being obtained from the pairing 
(contraction) of two fields $\c_{\xx(\ell),\o(\ell)}$ and $\c_{\xx'(\ell),\o'(\ell)}$;
\item if ${\bf t}=\{t_{i,i'}\in [0,1], 1\le i,i' \le s\}$, then $dP_{T}({\bf t})$
is a probability measure with support on a set of $\tt$ such that
$t_{i,i'}=\uu_i\cdot\uu_{i'}$ for some family of vectors $\uu_i\in \mathbb R^s$ of
unit norm;
\item if $2p=\sum_{i=1}^s|P_{v_i}^\c|$, then 
$G^{N+1,T}({\bf t})$ is an antisymmetric $(2p-2s+2)\times (2p-2s+2)$ matrix, whose
elements are given by $G^{N+1,T}_{f,f'}=t_{i(f),i(f')}g^{(N+1)}_{\ell(f,f')}$, where:
$f, f'\not\in\cup_{\ell\in T}\{f^1_\ell,f^2_\ell\}$ and $f^1_\ell,f^2_\ell$ are 
the two field labels associated to the two (entering and exiting) half-lines contracted into $\ell$;
$i(f)\in\{1,\ldots,s\}$ is s.t. $f\in P_{v_{i(f)}}$; $g^{(N+1)}_{\ell(f,f')}$ is the propagator associated to the line 
obtained by contracting the two half-lines with indices $f$ and $f'$;
\item  $\Pf G^{N+1,T}$ is the {\it Pfaffian} of $G^{N+1,T}$; given an antisymmetrix matrix
$G_{ij}=-G_{ji}$, $i,j=1,\ldots,2k$, its Pfaffian is defined as
\bea\Pf G&=&\frac{1}{2^k k!}\sum_{\p}(-1)^\p G_{\p(1)\p(2)}\cdots 
G_{\p(2k-1)\p(2k)}\nonumber\\
&=&\int d\c_1\cdots d\c_{2k}\,e^{-\frac12\sum_{i,j}\c_i 
G_{ij}\c_j }\;,\label{3.pf}\eea
where in the first line $\p$ is a permutation of $\{1,\ldots,2k\}$ and 
$(-1)^\p$ is its parity
while, in the second line, $\c_1,\ldots,\c_{2k}$ are Grassmanian variables.
A well known property is that $(\Pf G)^2=\det G$.
\end{itemize}

If $s=1$ the sum over $T$ is empty, but we can still
use the Eq.(\ref{3.pfe}) by interpreting the r.h.s.
as equal to $0$ if $P_{1}$ is empty and equal to $\Pf G^{N+1,T}({\bf 1})$ otherwise.
Note that if the Pfaffian is expanded by using Eq.(\ref{3.pf}), then Eq.(\ref{3.pfe}) reduces to the usual 
representation of the truncated expectation in terms of connected Feynman diagrams. 
The spanning trees in Eq.(\ref{3.pfe}) guarantee the minimal connection among the 
clusters of fields $v_1,\ldots,v_s$ and the Pfaffian can be thought of as a resummation of all the 
Feynman diagrams obtained by pairing (contracting) in all possible ways the fields $\c$ outside the 
spanning tree, with the rule that 
each contracted pair $(\c_{\xx,\o},\c_{\yy,\o'})$ is replaced by $g^\c_{\o,\o'}(\xx-\yy)$;
the interpolation in ${\bf t}$ is necessary in order to avoid an over-counting of the diagrams.

The reason why we prefer to use the Pfaffian expansion rather than the more usual expansion in 
connected Feynman diagrams is that the former is better behaved from a combinatorial point of 
view: using the fact that the number of spanning trees in the sum $\sum_T$ in the r.h.s.\ of 
Eq.(\ref{3.pfe}) is bounded by $s! C^{2p}$, where $2p=\sum_{i=1}^s|P_{v_i}^\c|$ (see, e.g., 
\cite[Appendix A3.3]{GM01} for a proof of this fact), we find that for fixed 
$P_{v_0}$ the contribution to the kernel of $\psi_{P_{v_0}^\psi}A_{P_{v_0}^A}$ 
coming from a fixed tree $\t\in\TT_{N;n,m}^{(N)}$ can be bounded as:
\bea && \sum_{P_{v_1},\ldots,P_{v_s}}\!\!\!\frac{1}{s!}\int d\xx_{v_0}\prod_{i=1}^s
|K^{(N+1)}_{v_i}(\xx_{v_i})|\sum_{T\in{\bf T}}\prod_{\ell\in T}|g^{(N+1)}_\ell|\,||\det G^{N+1,T}||^{1/2}_\io\quad \label{3.8s}\\
&&\le \,||g^{(N+1)}_\ell||_1^{s-1}
\sum_{P_{v_1},\ldots,P_{v_s}} \int d\xx_{v_0}\prod_{i=1}^s C^{|P_{v_i}|}
|K^{(N+1)}_{v_i}(\xx_{v_i})|\cdot||\det G^{N+1,T}||^{1/2}_\io\;,\nonumber\eea
where $||\det G^{N+1,T}||^{1/2}_\io$
indicates the square root of the $L^\io$ norm of the determinant w.r.t.\ both the position variables $\xx_{v_0}$
and the interpolation parameter ${\bf t}$ (the square root is due to the fact that 
$|\Pf G^{N+1,T}|=|\det G^{N+1,T}|^{1/2}$); moreover, in the second line, recalling that $g^{(N+1)}_\ell$
is a shorthand for $g^\c_{\o(\ell),\o'(\ell)}(\xx(\ell)-\xx'(\ell))$, 
\be ||g^{(N+1)}_\ell||_1=
\int d\xx ||g^\c(\xx)||\le C 2^{-N}\;,\ee
simply because  the 
propagator of the $\c$ field decays exponentially on scale $a=2^{-N}$, as it follows from 
its explicit expression Eq.(\ref{2.gchi}), see also Eq.(\ref{2.chisca}). More precisely, if $||\cdot||$ is the Hilbert-Schmidt norm, 
\be ||g^\c(\xx-\yy)||\le C2^{N}e^{-c2^{N}|\xx-\yy|}\;.\label{3.3a}\ee
In order to bound $\det G^T$, we use the {\it Gram-Hadamard inequality}, stating
that, if $M$ is a square matrix with elements $M_{ij}$ of the form
$M_{ij}=\media{A_i,B_j}$, where $A_i$, $B_j$ are vectors in a Hilbert space with
scalar product $\media{\cdot,\cdot}$, then
\be |\det M|\le \prod_i ||B_i||\cdot ||C_i||\;.\label{s5.6}\ee
where $||\cdot||$ is the norm induced by the scalar product. See \cite[Theorem A.1]{GM01}
for a proof of Eq.(\ref{s5.6}).
 
Let $\HHH=\RRR^s\otimes \HHH_0$, where $\HHH_0$ is the Hilbert space of the
functions ${\bf F} : \L\to \CCC^2$,
with scalar product $\media{{\bf F},{\bf G}}=\sum_{\o=\pm}\int d\zz\, F^*_\o(\zz)G_\o(\zz) $,
where $F_\o=[{\bf F}]_\o$, $G_\o=[{\bf G}]_\o$, $\o=\pm$, are the components of the vectors 
${\bf F}$ and ${\bf G}$. It is easy to verify that
\bea G^T_{f,f'}&=&t_{i(f),i(f')}\, g_{\o(f),\o(f')}(\xx(f)-\xx(f'))\nonumber\\
&=&\media{\uu_{i(f)}\otimes
{\bf B}_{\xx(f),\o(f)},
\uu_{i(f')}\otimes{\bf C}_{\xx(f'),\o(f')}}\;,\label{s5.8}\eea
where: $\uu_i\in \RRR^s$, $i=1,\ldots,s$, are vectors such that
$t_{i,i'}=\uu_i\cdot\uu_{i'}$; ${\bf B}_{\xx,\r}$ and ${\bf C}_{\xx,\r}$ have components:
\bea &&[{\bf B}_{\xx,\o}(\zz)]_\s=\frac{(2\p)^{1/2}}{L^2}\sum_{\kk} 
\frac{e^{-i\kk(\zz-\xx)}}{\big[|D_a(\kk)|^2+|\s_\c(\kk)|^2\big]^{1/4}}\d_{\o,s}\;,
\label{B.4}\\
&&[{\bf C}_{\xx,\o}(\zz)]_\s=\frac{(2\p)^{1/2}}{L^2}\sum_{\kk}
\,\frac{e^{-i\kk(\zz-\xx)}}
{\big[|D_a(\kk)|^2+|\s_\c(\kk)|^2\big]^{3/4}}\begin{pmatrix}D^+_a(\kk)
& i\s_\c(\kk)\\
-i\s_\c(\kk)& D^-_a(\kk)\end{pmatrix}_{\s,\o}\;,\nonumber\eea
so that
\be ||{\bf B}_{\xx,\o}||^2=||{\bf C}_{\xx,\r}||^2=\frac{2\p}{L^2}\sum_{\kk}
\frac{1}{\big[|D_a(\kk)|^2+|\s_\c(\kk)|^2\big]^{1/2}}
\le C 2^N\;,\label{B.5}\ee
for a suitable constant $C$. 
Using the Gram-Hadamard inequality, we find 
$||\det G^{N+1,T}||^{1/2}_\io\le ({\rm const.})^{\sum_{i=1}^s
|P_{v_i}^\c|}2^{N(\frac12\sum_{i=1}^s|P_{v_i}^\c|-s+1)}$; moreover,  using Proposition \ref{prop3},
we find that for each normal endpoint $v_i$:
%
%
\be \frac1{|\L|}\int d\xx_{v_i}
|K_{v_i}^{(N+1)}(\xx_{v_i})|\le(C|\l|)^{\max\{1,c|P_{v_i}|}\}2^{N(2-\frac12|P_{v_i}^\psi|-\frac12|P_{v_i}^\c|-|P_{v_i}^A|)}\;.\label{3.16s}\ee

Substituting these estimates into Eq.(\ref{3.8s})
we find that for fixed $P_{v_0}$ the contribution to 
the kernel of $\psi_{P_{v_0}}A_{P_{v_0}}$ 
from a given tree $\t\in\TT_{N;n,m}^{(N)}$ 
can be estimated by
\be C^{n+m}|\l|^{\max\{1,cn\}}2^{N(2-\frac12|P_{v_0}^\psi|-|P_{v_0}^A|)}\;,\label{3.19s}\ee
where we used the small factors $|\l|^{\max\{1,c|P_{v_i}^\c|\}}$ to sum over the field labels, see 
\cite[Appendix A6]{GM01}. Eq.(\ref{3.19s}) implies the analyticity in $\l$ of the kernels of $\wt \VV^{(N)}$ and of $\BB^{(N)}$,
as well as their exponential decay on scale $a=2^{-N}$. In particular, writing
\be \wt\VV^{(N)}(\psi,\bsA)=\sum_{n\ge 1,\ m\ge 0}^*\sum_{{\underline\o},
{\underline j}}\int d{\underline \xx}\,d{\underline\yy}\, \wt W^{(N)}_{2n,m;{\underline\o},{\underline j}}({\underline\xx};{\underline\yy})\Big[\prod_{i=1}^{2n}\psi_{\xx_i,\o_i}\Big]\,\Big[\prod_{i=1}^m A_{\yy_i,j_i}\Big]\;,
\label{3.5q0}\ee
where the $*$ on the sum indicates the constraint that $(2n,m)\neq (2,1)$, we have
\be  \frac1{|\L|}\int d\xx_1\cdots d\yy_m
|\wt W^{(h)}_{2n,m;{\underline\o},{\underline j}}(\xx_1,\ldots,\xx_{2n};\yy_1,\ldots,\yy_m)|\le C^{n+m}2^{N (2-2n-m)} |\l|^{c\,n}\label{3.5-0}\ee
A similar expansion and similar bounds are valid for $\SS^{(N)}(\bsA)$ as well.

\subsection{The iterative integration}

We are now left with the integration of the $\psi$ field. Recall that the propagator of the $\psi$ field is 
given by the inverse of $C_\s(\kk)$, see Eq.(\ref{2.109b}), which has a mass (i.e. an inverse decay rate)
proportional to $\s=\s(a)$, which can be arbitrarily small. In fact, recall that $\s(a)\neq 0$ 
$\forall a>0$, and 
$\lim_{a\to 0}\s(a)=m^*\in\mathbb R$. In particular, $m^*=0$ is an allowed value of the rescaled mass 
and it is our purpose to derive bounds that are uniform in $m^*$ for $m^*\to 0$ (massless limit).
In this respect, $\psi$ is essentially a massless field or, more precisely, it is non-uniformly massive.
Therefore, we cannot trivially integrate $\psi$ ``in one step", as we did for $\c$. A convenient procedure
is the following. We define a sequence of geometrically
decreasing momentum scales $2^h$, with $h=N,N-1,\ldots$ Correspondingly
we define a sequence of analytic functions $f_h(\kk)$ supported mostly around $|\kk|\sim 2^h$: 
for instance, we can choose $f_N(\kk)=1-\exp\{-2^{-2(N-1)}\kk^2\}$ and $f_h(\kk)=\exp\{-2^{-2h}\kk^2\}-\exp\{-2^{-2(h-1)}\kk^2\}$, $\forall h<N$, so that 
\be 1=\sum_{h\le N} f_h(\kk)\;.\label{3.7}\ee
The resolution of the identity Eq.(\ref{3.7}) induces a rewriting of the propagator of $\psi=:
\psi^{(\le N)}$ 
as a sum of propagators concentrated on smaller and smaller momentum scales and an iterative 
procedure to compute $\Xi_{-,-}(\bsA)$. At each step we decompose the
propagator into a sum of two propagators, the first approximately supported on
momenta $\sim 2^{h}$ (i.e.\ with a Fourier transform proportional to $f_h(\kk)$), $h\le 0$, the second approximately supported
on momenta smaller than $2^h$, $h\le N$. Correspondingly we rewrite the
Grassmann field as a sum of two independent fields: $\psi^{(\le
h)}=\psi^{(h)}+ \psi^{(\le h-1)}$ and we integrate out the field
$\psi^{(h)}$ in the same way as we did for $\c$. The result is that, 
for any $h\le N$, we can rewrite  Eq.(\ref{3.1})  as
\be \Xi_{-,-}(\bsA)=e^{|\L|F_h+\SS^{(h)}(\bsA)}\int P_{\c_h,Z_h,\s_h}(d\psi^{(\le h)})
e^{\VV^{(h)}(\psi^{(\le h)},\bsA)+
\BB^{(h)}(\psi^{(\le h)},\bsA)}\;,\label{3.8}\ee
where $F_h,Z_h,\s_h,\VV^{(h)},\BB^{(h)}$ will be defined recursively, 
$\c_h(\kk)=\sum_{k\le h}f_k(\kk)$ and $P_{\c_h,Z_h,\s_h}(d\psi^{(\le h)})$ is the 
gaussian integration with propagator (recall the definition $D_a^\pm=a^{-1}(i\sin ak_1\pm\sin ak_2)$)
\be g^{(\le h)}(\xx)=\frac{2\p}{Z_h}\frac1{L^2}\sum_{\kk\in\DD_M}\frac{e^{-i\kk\xx}\c_h(\kk)}{|D_a(\kk)|^2+\s_h^2}
\begin{pmatrix}D_a^+(\kk)&i\s_h\\ -i\s_h& D^-_a(\kk)\end{pmatrix}\;.\ee
If $h=N$, then: $Z_N=1$, $\s_N=\s(a)=\s$ and 
\be \VV^{(N)}(\psi,\bsA)=\wt\VV^{(N)}(\psi,\bsA)
-\frac1{4\p L^2}\!\!\sum_{\kk\in\DD_M}\!\!\psi^T_{-\kk}\Big[C_\s(\kk)+\begin{pmatrix}
D_a^-(\kk)& -i\s\\i\s&D_a^+(\kk)\end{pmatrix}\Big]\psi_\kk\ee
In the following steps, the {\it effective potential} $\VV^{(h)}$ and the effective source 
term $\BB^{(h)}$ will be shown to have the following structure: 
$\BB^{(h)}(\psi,\bsA)$ is quadratic in $\psi$ and linear in $\bsA$, while $\VV^{(h)}$ admits an 
expansion analogous to Eq.(\ref{3.5q0})
\be \VV^{(h)}(\psi,\bsA)=\sum_{\substack{n,m\ge 0\\n+m\ge 1}}^*\sum_{{\underline\o},
{\underline j}}\int d{\underline \xx}\,d{\underline\yy}\,  W^{(h)}_{2n,m;{\underline\o},{\underline j}}
({\underline\xx};{\underline\yy})\Big[\prod_{i=1}^{2n}\psi_{\xx_i,\o_i}\Big]\,\Big[\prod_{i=1}^m A_{\yy_i,j_i}\Big]\;,
\label{3.5q}\ee
where, as we will see below,
\be  \frac1{|\L|}\int d\xx_1\cdots d\yy_m
|W^{(h)}_{2n,m;{\underline\o},{\underline j}}(\xx_1,\ldots,\xx_{2n};\yy_1,\ldots,\yy_m)|\le C^{n+m}2^{h (2-2n-m)} |\l|^{c\,n}\label{3.5}\ee
The iteration continues until the scale $h=h_{\s}:=\lfloor \log_2 \s(a)\rfloor$ is reached. At that point,
the left-over propagator, $g^{(\le h_\s)}$ is massive on the ``right scale" (i.e. on the very same scale
$2^{h_\s}$), so that the associated degrees of freedom can be integrated in one step. The result is the 
desired generating function, from which we can finally compute the multi-point energy correlation 
functions.

\subsection{Localization and renormalization}\label{sec3.loc}

In order to inductively prove Eq.\pref{3.8}
we write
\be {\cal V}^{(h)}(\psi,\bsA) =\LL{\cal V}^{(h)}(\psi)+\RR{\cal V}^{(h)}(\psi,\bsA)\;,\label{3.loc}
\ee
where, if we think of the kernel $\hat W^{(h)}_{2,0;(\o_1,\o_2)}$ as a $2\times2$ matrix with 
matrix indices $\o_1,\o_2$, 
\bea &&
\LL{\cal V}^{(h)}(\psi)=\label{2.localize}\\
&&=\frac1{L^2}\sum_{\kk\in\DD_M}\hat \psi_{-\kk}^T
\big[(\PP_0+\PP_1)\hat W_{2,0}^{(h)}(\V0)+\dd_a(\kk)\cdot\dpr_\kk\PP_0\hat W_{2,0}^{(h)}(\V0)\big]
\hat\psi_{\kk}+\nonumber\\ && +\sum_{\o_1,\o_2,\o_3,\o_4}\int d\xx_1\cdots d\xx_4 
W_{4,0;\underline\o}(\xx_1,\xx_2,\xx_3,\xx_4)
 \psi^{(\le h)}_{\xx_1,\o_1}\psi^{(\le h)}_{\xx_1,\o_2}\psi^{(\le h)}_{\xx_1,\o_3}\psi^{(\le h)}_{\xx_1,\o_4}
\;,\nonumber\eea
Here $\PP_0$ (resp. $\PP_1$) is an operator that extracts the order $0$ (resp. order 1) in 
$\{\s_k\}_{k>h}$ from the kernel which it acts on, and $\dd_a(\kk):=a^{-1}(\sin ak_1,\sin ak_2)$.
\\

{\bf Remarks.}\begin{enumerate}
\item $\LL$ will be called the {\it localization operator}, which should be thought of as the linear operator 
extracting from the effective potential $\VV^{(h)}$ its local (singular) part, while 
$\RR$ will be called the {\it 
renormalization operator}, which is the linear  operator extracting from $\VV^{(h)}$ its regular part. 
Note that the action of $\RR$ on the kernels quadratic in $\psi$ is equivalent (via
the use of the remainder's formula in the Taylor's expansion) to the action of a second order 
differential operator of the type $\kk^2\dpr_\kk^2$ or $\sum_{k>h}\s_k^2\dpr_{\s_k}^2$ (or a suitable
combination of the two). Dimensionally, the differential operator $\kk^2\dpr_\kk^2$ (or, equivalently,
its $\xx$-space counterpart, which is of the form $\xx^2\dpr^2_\xx$)
behaves as 
$2^{2(h_1-h_2)}$, where $h_1$ is the scale of $\kk$ and $h_2$ the scale of the propagator that 
$\dpr_\kk$ acts on; moreover, the iterative integration is set up in such a way that by construction 
$h_1<h_2$, so that $2^{2(h_1-h_2)}$ is a dimensional gain, sufficient to regularize the quadratic 
kernels 
of the effective potential. Similarly, the action of the operator $\sum_{k>h}\s_k^2\dpr_{\s_k}^2$
on the kernels quadratic in $\psi$ is dimensionally equivalent to a multiplication by $\s_h^22^{-2h}
\simeq \big(\frac{\s_h}{\s_{h_\s}}\big)^2 2^{2(h_\s-h)}$; as we will see, $\s_h\simeq \s=\s(a)$ at all 
scales,
so that  the action of $\sum_{k>h}\s_k^2\dpr_{\s_k}^2$ is dimensionally equivalent to a multiplication
by the gain factor $2^{2(h_\s-h)}$ which is enough to regularize the quadratic kernels of the effective 
potential. A similar discussion is valid for the action of $\RR$ on the quartic kernels.
\item A key fact which makes the theory at hand treatable (and asymptotically free)
is that the quartic term in the second line is zero ``by the Pauli principle", i.e., simply by the Grassmann 
rule $\psi_{\xx,\o}^2=0$. In fact, note that at least two of the four $\o$ indices must be equal among each other. 
Therefore the integrand in the second line is identically zero. This property can be diagramatically interpreted by 
saying that the fermionic nature of the theory automatically renormalizes the 
four-field interaction, which is dimensionally marginal (see below) 
but effectively irrelevant thanks to the cancellation that we just mentioned.
\end{enumerate}

Similarly, we decompose the source term as 
\be \BB^{(h)}(\psi,\bsA) =\LL\BB^{(h)}(\psi,\bsA)+\RR\BB^{(h)}(\psi,\bsA)\;,\label{3.locbb}
\ee
with 
\be \LL\BB^{(h)}(\psi,\bsA)=\frac1{L^4}\sum_{\kk,\pp\in\DD_M}(\hat A_{\pp,1}+\hat A_{\pp,2})
\psi_{-\kk-\pp}^T
\PP_0\hat W_{2,1}^{(h)}(\V0)\hat\psi_{\kk}
\;,\label{3.10}\ee
where $\hat A_{\pp,j}=\int d\xx e^{i\pp\xx}A_{\xx,j}$. As we will see below, $\LL\BB^{(h)}$ is 
a marginal operator in the Renormalization Group sense. 

The symmetries of the theory, which are described in Sec.\ref{sec2.sym} and
are preserved by the iterative integration procedure, imply that the kernel in the second
line of Eq.(\ref{2.localize}) has a structure analogous to Eq.(\ref{2.form});
we also choose to separate the mass term (i.e. the constant in the off-diagonal elements of the 
matrix) in two parts, one proportional to the bare mass and one independent of it, namely:
\bea \LL\VV^{(h)}(\psi)&=&\frac1{L^2}\sum_\kk
\hat \psi_{-\kk}^T\Big[\frac1{4\p}
\begin{pmatrix} \z_hD_a^-(\kk) & -is_h \\
is_h & \z_hD_a^+(\kk)\end{pmatrix}+2^h\n_h\s_2\Big]\hat \psi_\kk\nonumber\\
&=:&\LL_0\VV^{(h)}(\psi)+2^h\n_hF_\n(\psi)\;, \label{3.11}\eea
where $F_\n(\psi):=\int d\xx\,\psi_\xx\s_2\psi_\xx$ and $\z_h,s_h,\n_h$ are suitable real constants, such that $s_h$ is linear in $\{\s_k\}_{k>h}$ and 
$\n_h$ is independent of $\{\s_k\}_{k>h}$. Note that the counterterm on scale $h=N$ is defined so that 
\be \n_N:=\n+2^{-N-1}{\rm Tr}\big[\s_2\PP_0\hat W^{(N)}_{2,0}(\V0)\big]\;.\ee
%
We stress once again that 
the local part of $\VV^{(h)}$ is purely quadratic in $\psi$: the quartic term, a priori present in 
$\LL\VV^{(h)}$ is zero thanks to the ``Pauli principle", in the sense of the  Remark 2 above.
Similarly, the local part of the source term can be written as
\be \LL\BB^{(h)}=\frac{Z^{(1)}_h}{2\p}\int d\xx (A_{\xx,1}+A_{\xx,2})\psi_{\xx}\s_2\psi_{\xx}\;.\ee
for a suitable real constant $Z_{h}^{(1)}$, with $Z_N^{(1)}=1$.

Once that the above definitions are given, we can describe our iterative
integration procedure for $h\le N$. We start from Eq.(\ref{3.8}), which is inductively assumed to be valid
at the $h$-th step, and we prove the validity of the representation for $h-1$. We rewrite 
Eq.(\ref{3.8}) as
\bea && e^{|\L|F_h+\SS^{(h)}(\bsA)}\!\int\! P_{\c_h,Z_h,\s_h}(d\psi^{(\le h)}) \, e^{\LL_0{\cal V}^{(h)}
(\psi^{(\le h)})+2^h\n_hF_\n(\psi^{(\le h)})}\cdot\nonumber\\
&&\hskip2.3truecm\cdot e^{\RR{\cal V}^{(h)}
(\psi^{(\le h)},\bsA)+\BB^{(h)}(\psi^{(\le h)},\bsA)}\label{3.12}\eea
Next we include $\LL_0\VV^{(h)}$ in the fermionic integration, so obtaining
\bea && e^{|\L|(F_h+e_h)+\SS^{(h)}(\bsA)}\!\int\! P_{\c_h,\lis Z_{h-1},\lis \s_{h-1}}(d\psi^{(\le h)}) \, 
e^{2^h\n_hF_\n(\psi^{(\le h)})}\cdot\nonumber\\
&&\hskip2.9truecm\cdot e^{\RR{\cal V}^{(h)}
(\psi^{(\le h)},\bsA)+\BB^{(h)}(\psi^{(\le h)},\bsA)}\label{3.13}\eea
where 
\be \lis Z_{h-1}(\kk)= Z_h +\z_h \c_h(\kk)\;,\qquad
\lis Z_{h-1}(\kk)\lis \s_{h-1}(\kk)= Z_h \s_h+s_h \c_h(\kk)\;,\label{3.14}\ee
and $e_h$ is a constant fixed in such a way that $\int P_{\c_h,\lis Z_{h-1},\lis \s_{h-1}}(d\psi^{(\le h)})=1$.
Now we can perform the integration of the $\psi^{(h)}$ field.
We rewrite the Grassmann field $\psi^{(\le h)}$ as a sum of two independent
Grassmann fields $\psi^{(\le h-1)}+\psi^{(h)}$ and correspondingly, if we let $Z_{h-1}:=\lis Z_{h-1}(\V0)$
and  $\s_{h-1}:=\lis \s_{h-1}(\V0)$, we rewrite Eq.(\ref{3.13}) as
\bea && e^{|\L|(F_h+e_h)+\SS^{(h)}(\bsA)}\!\int\! P_{\c_{h-1},Z_{h-1},\s_{h-1}}(d\psi^{(\le h-1)}) \int P_{\widetilde f_h,
Z_{h-1},\wt \s_{h-1}}(d\psi^{(h)}) \cdot\qquad\nonumber \\
&&\qquad \cdot e^{2^h\n_hF_\n(\psi^{(\le h-1)}+\psi^{(h)})+\RR{\cal V}^{(h)}
(\psi^{(\le h-1)}+\psi^{(h)},\bsA)+\BB^{(h)}(\psi^{(\le h-1)}+\psi^{(h)},\bsA)}\nonumber\eea
%
where
\bea && \widetilde f_h(\kk)=Z_{h-1}\Big[\frac{\c_h(\kk)}{\lis Z_{h-1}(\kk)}\frac{|D_a(\kk)|^2+\s^2_{h-1}}
{|D_a(\kk)|^2+\lis \s^2_{h-1}(\kk)}-\frac{\c_{h-1}(\kk)}{Z_{h-1}}
\Big]\;,\nonumber\\
&&\wt f_h(\kk) \wt \s_{h-1}(\kk)=\big[\wt f_h(\kk)+\c_{h-1}(\kk)\big]\lis \s_{h-1}(\kk)-
\c_{h-1}(\kk)\s_{h-1}\;.\nonumber\eea
The single scale propagator is
\be g^{(h)}(\xx-\yy)=\frac{2\p}{Z_{h-1}}\frac1{L^2}\sum_{\kk\in\DD_M}\frac{e^{-i\kk\xx}\,\wt f_h(\kk)}{|D_a(\kk)|^2+\wt \s_{h-1}^2(\kk)}
\begin{pmatrix}D_a^+(\kk)&i\wt \s_{h-1}(\kk)\\ -i\wt \s_{h-1}(\kk)& D^-_a(\kk)\end{pmatrix}\;.
\label{3.36yt}\ee
A key remark is that, if $|Z_h-1|\le ({\rm const}.)|\l|$ and $|\s_h-\s|\le ({\rm const}.)|\s\l|$ (as we shall 
inductively prove below) and for all scales larger than $h_\s$, the propagator $g^{(h)}$ satisfies a 
bound analogous to Eq.(\ref{3.3a}):
\be ||g^{(h)}(\xx-\yy)||\le C2^{h}e^{-c2^{h}|\xx-\yy|}\;.\label{3.3a2}
\ee
Therefore, we can integrate the field on scale $h$ in the same way as we did for $\c$, after which 
we are left with an integral involving the fields
$\psi^{(\le h-1)}$ and the new effective interaction, defined as
\bea && e^{|\L|\lis e_h+\lis\SS_{h-1}(\bsA)+{\cal V}^{(h-1)}
(\psi^{(\le h-1)},\bsA)+\BB^{(h-1)}(\psi^{(\le h-1)},\bsA)}=\\
&&\qquad = \int P_{\widetilde f_h,
Z_{h-1},\wt \s_{h-1}}(d\psi^{(h)})e^{\widehat \VV^{(h)}(\psi^{(\le h-1)}+\psi^{(h)},\bsA)}\;,\nonumber\eea
where 
\be \widehat \VV^{(h)}(\psi)=2^h\n_hF_\n(\psi)+\LL\BB^{(h)}(\psi,\bsA)+\RR{\cal V}^{(h)}
(\psi,\bsA)+\RR\BB^{(h)}(\psi,\bsA)\;,\nonumber\ee
and ${\cal V}^{(h-1)}(0,\bsA)=0$, so that $F_{h-1}=F_h+e_h+\lis e_h$ and
$\SS^{(h-1)}(\bsA)=\SS^{(h)}(\bsA)+\lis\SS_{h-1}(\bsA)$, 
while $\BB^{(h-1)}(\psi,\bsA)$
collects all the terms quadratic in $\psi$
and linear in $\bsA$. The potential ${\cal V}^{(h-1)}$ is of the form Eq.(\ref{3.5q}) as one sees 
by using the identity 
\bea && |\L|\lis e_h+\lis\SS_{h-1}(\bsA)+{\cal V}^{(h-1)}(\psi^{(\le h-1)},\bsA)+\BB^{(h-1)}(\psi^{(\le h-1)},\bsA)=\label{3.16}\\
&&=\sum_{n\ge 1}\frac{1}{n!}\EE^T_h( \underbrace{\widehat \VV^{(h)}(\psi^{(\le h-1)}+\psi^{(h)},\bsA);
\cdots;\widehat \VV^{(h)}(\psi^{(\le h-1)}+\psi^{(h)},\bsA)}_{n\ {\rm times}})
\;,\nonumber\eea
where $\EE^T_h$ is the truncated expectation w.r.t.\ the propagator $g^{(h)}$.

The iteration continues up to scale $h_\s$. At that point, we integrate in one step all the remaining degrees of 
freedom, by taking advantage of the fact that 
\be ||g^{(\le h_\s)}(\xx-\yy)||\le C2^{h_\s}e^{-c2^{h_\s}|\xx-\yy|}\;,\label{3.3abis}\ee
where by definition  $2^{h_\s}$ is of the same order as $\s=\s(a)\neq 0$. After the integration of 
$\psi^{(\le h_\s)}$ we are left with the desired generating function, $\log \Xi_{-,-}(\bsA)$.

Note that the above procedure allows us to write the
{\it effective constants} $(Z_h,\s_h,\n_h,Z^{(1)}_h)$ with $h\le N$,
in terms of $(Z_k,\s_k,\n_k,Z^{(1)}_k)$ with $h<k\le N$:
\bea && Z_{h-1}=Z_h+\b^{Z}_h\;,\qquad 
\frac{\s_{h-1}}{\s_h}=1+\b^{\s}_h\;,\label{pp0}\\
&&\n_{h-1}=2\n_h+\b^\n_h\;,\qquad  \frac{Z^{(1)}_{h-1}}{Z^{(1)}_h}=1+\b^{Z,1}_h\;,\label{pp} \eea
where $\b_h^\#=\b_h^\#\big((Z_h,\s_h,\n_h,Z^{(1)}_h),\ldots,
(Z_N,\s_N,\n_N,Z^{(1)}_N)\big)$ is the so--called {\it Beta function}.

The effective constants, sometimes called the {\it running coupling constants} measure the strength 
of the local part of the effective potential. The key point, to be proved below, is that, 
if we assume that the running coupling constants stay close to their bare values at scale $N$, then the expansion for the 
effective potential induced by the iterative procedure above, is well defined and {\it analytic} in the sequence of running 
coupling constants. Moreover, under the same assumptions, 
the beta function itself is well defined and analytic in the running coupling constants: this allows us to 
study the evolution of these effective constants under the dynamical system defined by $\b_h^\#$. 
We will see that, {\it uniformly} in $h$,
\be Z_h=1+O(\l)\;,\qquad Z_h^{(1)}=1+O(\l)\;,\qquad \n_h=O(\l)\;,\qquad 
\s_h=\s(1+O(\l))\;,\label{3.flow}\ee
where $\s=\s(a)$. The boundedness of the flow of the running coupling constants, Eq.(\ref{3.flow}),
is one of the key ingredients that allow us to prove the analyticity of the theory in the bare coupling 
constants, see below.

\subsection{Tree expansion for the effective potential} \label{sec.trees}

In order to prove the analyticity properties of the effective potential announced in the previous section, 
we first need to prove that the kernels of the effective potential, 
as obtained from the iterative procedure described above, are expressed by absolutely convergent 
series in the sequence of the running coupling constants, provided that these are assumed to be close 
to their bare values at scale $h=N$. Then we will show that this assumption is justified, by proving that 
their flow under the beta function remains bounded. In any case, as a first step towards the 
full control of the theory, we need a more explicit representation of the kernels of the effective potential 
and a systematic and efficient way to bound them. The idea is to systematically represent the action 
of $\EE^T_h$ in Eq.(\ref{3.16}) in a way analogous to the one described in Section \ref{3.chi} for the
integration of the $\c$ field. By iterating the graphical equation Fig.\ref{fig6.1}
we obtain a tree expansion for the effective potential, which has been first introduced by G. Gallavotti 
and F. Nicol\`o in \cite{GN} and since then it has been 
described in detail in several papers that make use of constructive renormalization group 
methods, see e.g.\ \cite{Ga} and the more recent reviews \cite{GM01,Gi,Ma}. 
The main features of this expansion (to be referred to as the expansion in Gallavotti-Nicol\`o trees, or 
in GN trees for short) are described below. 

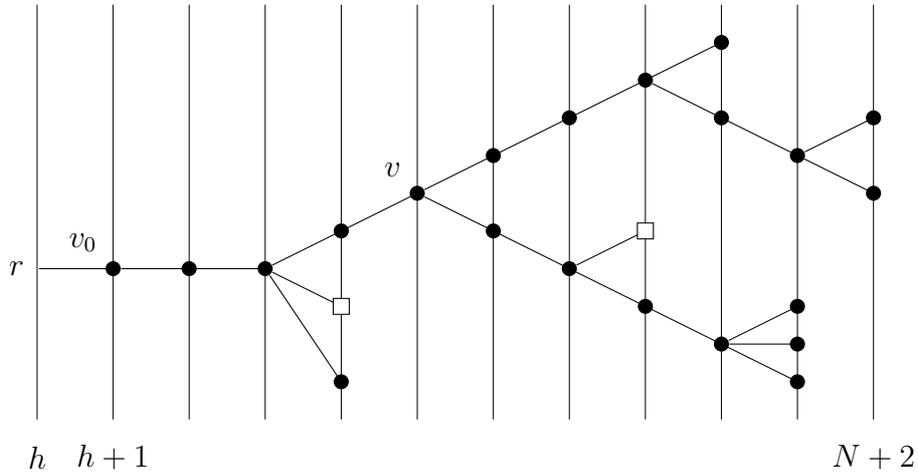
\begin{figure}[ht]
\centering
\begin{tikzpicture}
    \foreach \x in {0,...,11}
    {
        \draw[very thin] (\x ,2) -- (\x , 7.5);
    }
    	\draw (0,1.5) node {$h$};
    	\draw (1,1.5) node {$h+1$};
	\draw (11,1.5) node {$N+2$};
	\draw (-0,4) node[inner sep=0,label=180:$r$] (r) {};
	\draw (r) -- ++(1,0) node [vertex,label=130:$v_0$] (v0) {};
	\draw (v0) -- ++ (1,0) node[vertex] {} -- ++ (1,0) node[vertex] (v1) {};
	\draw (v1) -- ++ (1,-1.5) node[vertex] {};
	\draw (v1) -- ++ (1,-0.5) node [specialEP] {};
	\draw (v1) -- ++(1,0.5) node [vertex] {} -- ++ (1,0.5) node [vertex,label=130:$v$] (v) {};
	\draw (v) -- ++ (1,0.5) node [vertex] {} -- ++ (1,0.5) node [vertex] {} -- ++ (1,0.5) node [vertex] (v2) {} -- ++ (1,0.5) node [vertex] {};
		\draw (v2) -- ++ (1,-0.5) node [vertex] {} -- ++ (1,-0.5) node [vertex] (v3) {} -- ++ (1,-0.5) node [vertex] {};
		\draw (v3) -- ++ (1,0.5) node [vertex] {};
	\draw (v) -- ++ (1,-0.5) node [vertex] {} -- ++ (1,-0.5) node [vertex] (v4) {}-- ++ (1,-0.5) node [vertex] {}-- ++ (1,-0.5) node [vertex] (v5) {}-- ++ (1,-0.5) node [vertex] {};
		\draw (v4) -- +(1,0.5) node[specialEP] {};
		\draw (v5) -- +(1,0.5) node [vertex] {};
		\draw (v5) -- +(1,0) node [vertex] {};
\end{tikzpicture}
\caption{A tree $\t\in\TT^{(h)}_{N;n,m}$ with $n=7$ and $m=2$: the root is on scale $h$ and the  
endpoints are on scales $\le N+2$.\label{fig6.3}}
\end{figure}

\begin{enumerate}
\item Let us consider the family of all trees which can be constructed
by joining a point $r$, the {\it root}, with an ordered set of $n+m\ge 1$
points, the {\it endpoints} of the {\it unlabeled tree},
so that $r$ is not a branching point. The endpoints can be of two types, either normal or special,
the first drawn as black dots, the second as open squares, see Fig.\ref{fig6.3}; $n$ 
and $m$ indicate the number of normal and special endpoints, respectively.
The branching points will be called
the {\it non trivial vertices}.
The unlabeled trees are partially ordered from the root
to the endpoints in the natural way; we shall use the symbol $<$
to denote the partial order.
Two unlabeled trees are identified if they can be superposed by a suitable
continuous deformation, so that the endpoints with the same index coincide.
We shall also consider the {\it labelled trees} (to be called
simply trees in the following); they are defined by associating
some labels with the unlabelled trees, as explained in the
following items.
\item We associate a label $0\le h\le N$ with the root and we denote by
$\TT^{(h)}_{N;n,m}$ the corresponding set of labeled trees with $n$
normal and $m$ special endpoints. Moreover, we introduce a family of vertical lines,
labeled by an integer taking values in $[h,N+2]$, and we represent
any tree $\t\in\TT^{(h)}_{N;n,m}$ so that, if $v$ is an endpoint, it is contained 
in a vertical line with index $h+1< h_v\le N+2$, while
if it is a non trivial vertex, it is contained in a vertical line with index
$h<h_v\le N+1$, to be called the {\it scale} of $v$; the root $r$ is on
the line with index $h$.
In general, the tree will intersect the vertical lines in set of
points different from the root, the endpoints and the branching
points; these points will be called {\it trivial vertices}.
The set of the {\it
vertices} will be the union of the endpoints, of the trivial
vertices and of the non trivial vertices; note that the root is not a vertex.
Every vertex $v$ of a
tree will be associated to its scale label $h_v$, defined, as
above, as the label of the vertical line whom $v$ belongs to. Note
that, if $v_1$ and $v_2$ are two vertices and $v_1<v_2$, then
$h_{v_1}<h_{v_2}$.
\item There is only one vertex immediately following
the root, called $v_0$ and with scale label equal to $h+1$.
\item Given a vertex $v$ of $\t\in\TT^{(h)}_{N;n,m}$ that is not an endpoint,
we can consider the subtrees of $\t$ with root $v$, which correspond to the
connected components of the restriction of
$\t$ to the vertices $w\ge v$. If a subtree with root $v$ contains only
$v$ and one endpoint on scale $h_v+1$,
it will be called a {\it trivial subtree}.
\item With each normal (resp. special) endpoint $v$ on scale $N+2$ 
we associate a factor $\lis \VV(\psi^{(\le N)},\c,\bsA)$
(resp. $\lis\BB(\psi^{(\le N)},\c,\bsA)$); here $\c=:\psi^{(N+1)}$ 
should be thought of as the field on scale $N+1$. 
With the endpoints on scale $h_v\le N+1$ we associate
a factor  $2^{h_v-1}\n_{h_v-1}F_\n(\psi^{(\le h_v-1)})$ if the endpoint is normal or a factor 
$\LL\BB^{(h_v-1)}$ if the endpoint is special. The vertex $v'$ immediately preceding an endpoint 
$v$ on scale $h_v\le N+1$ is necessarily non trivial. Note that none of the endpoint on scale $\le N+1$ 
is associated with a quartic operator in the $\psi$ fields: this is due to the cancellation mentioned 
in Remark 2 at the beginning of Sec.\ref{sec3.loc}.
\end{enumerate}

In terms of these trees, the effective potential 
${\cal V}^{(h)}$ can be written as
\bea &&  |\L| \lis e_{h+1}+\lis\SS_{h}(\bsA)+{\cal V}^{(h)}(\psi^{(\le h)},\bsA) +\BB^{(h)}(\psi^{(\le h)},\bsA) =
\nonumber\\
&&\qquad =
\sum_{\substack{n,m\ge 0\\ n+m\ge 1}}\sum_{\t\in\TT^{(h)}_{N;n,m}}
\VV^{(h)}(\t,\psi^{(\le h)},\bsA)\;,\label{3.17}\eea
where, if $v_0$ is the first vertex of $\t$, if $\t_1,\ldots,\t_s$ ($s=s_{v_0}$)
are the subtrees of $\t$ with root $v_0$, and if $\EE^T_{h+1}$ is the truncated expectation 
associated to the propagator $g^{(h)}$,
\be {\cal V}^{(h)}(\t,\psi^{(\le h)},\bsA)=\frac{1}{s!} \EE^T_{h+1}
\big(\lis V^{(h+1)}(\t_1,\psi^{(\le h+1)},\bsA);\ldots; \lis V^{(h+1)}
(\t_{s},\psi^{(\le h+1)},\bsA)\big)\;,\label{3.18}\ee
and $\lis V^{(h+1)}(\t_i,\psi^{(\le h+1)},\bsA)$:
\begin{itemize}
\item is equal to $\RR\VV^{(h+1)}(\t_i,\psi^{(\le h+1)},\bsA))$ (where $\RR$ is the linear 
operator induced by the definitions Eqs.(\ref{3.loc}),(\ref{2.localize}),(\ref{3.locbb}),(\ref{3.10})) if $\t_i$ is non trivial;
\item is equal to $2^{h+1}\n_{h+1}F_\n(\psi^{(\le h+1)})$ if $\t_i$ is trivial, $h< N$ and the endpoint
of $\t_i$ is normal;
\item is equal to $\LL\BB^{(h+1)}(\psi^{(\le h+1)},\bsA)$ if $\t_i$ is trivial, $h< N$ and the endpoint
of $\t_i$ is normal;
\item is equal to $\lis \VV(\psi^{(\le N+1)},\bsA)$ (resp. $\lis \BB(\psi^{(\le N+1)},\bsA)$) 
if $\t_i$ is trivial, $h=N$ and the endpoint of $\t_i$ is normal (resp.\ special).
\end{itemize}
Note that, by the definition of $\lis V^{(h)}$, it is apparent that all the nodes that are not endpoints are 
associated with the action of the renormalization operator $\RR$. The local (singular) parts only appear 
in the contributions associated with the endpoints: in this sense, the singular parts ``do not 
accumulate". 
Using its inductive definition Eqs.(\ref{3.17})-(\ref{3.18}) and the Pfaffian representation 
for the truncated expectations, the right hand side of Eq.(\ref{3.17}) can be
put in a form similar to the one derived in Section \ref{3.chi} for the integration of the $\c$ field. 
To do that, we further expand $\VV^{(h)}(\t,\psi^{(\le h)},\bsA)$ by distinguishing the different 
contributions arising from the choices of the monomials in the factors $\lis\VV$ and $\lis\BB$
associated with the endpoints on scale $N+2$, as well as the scale at which each field in these 
monomials is contracted. To this purpose, we introduce a few more definitions, which generalize the 
definition of $P_v$ in Sec.\ref{3.chi}.

We introduce a {\it field label} $f$ to distinguish the field variables
appearing in the monomials associated with the endpoints;
the set of field labels associated with the endpoint $v$ will be called $I_v$; 
if $v$ is not an endpoint, we shall
call $I_v$ the set of field labels associated with the endpoints following
the vertex $v$. Note  that every field can be either of type $A$ or $\psi$: correspondingly, we denote by 
$I_v^A$ and $I_v^\psi$ the set of field labels of type $A$ and $\psi$, respectively, associated with $v$. 
Furthermore, we denote by $\xx(f)$ the
space-time point  of the field variable with label $f$; if $f\in I_v^A$, we denote by $j(f)$ the $j$ index of 
the external field with label $f$;  if $f\in I_v^\psi$, we denote by $\o(f)$ the $\o$ index of the 
Grassmann field with label $f$. 

We associate with any vertex $v$ of the tree a subset $P_v$ of $I_v$,
the {\it external fields} of $v$; we further denote by $P_v^A$ and $P_v^\psi$ the subsets 
of $P_v$ of fields of type $A$ and $\psi$, respectively (of course, $P_v^A\cap P_v^\psi=\emptyset$ and 
$P_v^A\cup P_v^\psi=P_v$). These subsets must satisfy various
constraints. First of all, if $v$ is not an endpoint and $v_1,\ldots,v_{s_v}$
are the $s_v\ge 1$ vertices immediately following it, then
$P_v \subseteq \cup_i
P_{v_i}$; if $v$ is an endpoint, $P_v=I_v$.
If $v$ is not an endpoint, we shall denote by $Q_{v_i}$ the
intersection of $P_v$ and $P_{v_i}$; this definition implies that $P_v=\cup_i
Q_{v_i}$. The union of the subsets $P_{v_i}\setminus Q_{v_i}$
is, by definition, the set of the {\it internal fields} of $v$,
and is non empty if $s_v>1$. Similar definitions are valid for $Q_v^\psi, Q_v^A$, etc. 
Note that $P_v^A=Q_v^A$ for all $v$, simply because $A$ is, by definition, an external field. 
Given $\t\in\TT^{(h)}_{N;n,m}$, there are many possible choices of the
subsets $P_v$, $v\in\t$, compatible with all the constraints. We
shall denote by ${\cal P}_\t$ the family of all these choices and by ${\bf P}$
the elements of ${\cal P}_\t$. 

Let us note that the resulting expansion for the effective potential, as compared to the expansion 
for $\wt\VV^{(N)}$ described in Sec. \ref{3.chi}, has an important difference related to the iterative action of the $\RR$ operator on the 
nodes of $\t$ that are not endpoints; as observed in the first Remark at the beginning of Sec.\ref{sec3.loc},
this is equivalent to the action of a suitable differential operator, whose precise form 
has been discussed in several papers on the subject, see e.g.\ \cite{BG,BGPS,BM01}; fortunately, in our 
context we do not need to describe the form of the interpolation operator exactly, but only some of its 
general properties, discussed in the following.

\subsubsection{The non-renormalized expansion}

Let us start with describing the basic dimensional bounds of the effective potential, {\it temporarily 
neglecting the action of the renormalization operator}. More precisely,  let us temporarily pretend that 
the action of $\RR$ on the nodes of $\t$ that are not endpoints, induced by definition 
Eq.(\ref{3.18}) and by the first item immediately following it, is replaced by the identity. 
Then the result of the iteration would lead to the 
following relation:
\be\VV_*^{(h)}(\t,\psi^{(\le h)},\bsA)=\sum_{
{\bf P}\in\PP_\t}\sum_{T\in{\bf T}}\int d\xx_{v_0}W^*_{\t,{\bf P},{\bf T}}
(\xx_{v_0})\psi^{(\le h)}_{P_{v_0}^\psi}A_{P_{v_0}^A}
\;,\label{5.49a}\ee
where $\xx_{v_0}$ is the set of integration variables
associated with $\t$ and $T=\bigcup_{v\ {\rm not}\ {\rm e.p.}} T_v$ is the union of the spanning trees 
associated with all the nodes that are not endpoints in $\t$. Moreover,  
$W^*_{\t,{\bf P},{\bf T}}$ is given by
\bea && W^*_{\t,{\bf P},{\bf T}}(\xx_{v_0})=\label{5.50s}\\
&&\ =\Big[\prod_{v\ {\rm e.p.}} K^{(h_v)}_{v}(\xx_{v})\Big]
\Big\{\prod_{v\ {\rm not}\ {\rm e.p.}}\frac1{s_v!} \int
dP_{T_v}({\bf t}_v) \Pf G^{h_v,T_v}({\bf t}_v)
\Big[\prod_{\ell\in T_v} g^{(h_v)}_\ell\Big]\Big\}\;,\nonumber\eea
which can be thought of as the multiscale version of the expansion derived in Section \ref{3.chi}.
Here $s_v$ indicates the number of nodes immediately following $v$ on $\t$.

Eq.(\ref{5.50s}) can be bounded in a way analogous to Eq.(\ref{3.8s}); roughly speaking 
the rationale is that each $g^{(h_v)}_\ell$ is replaced by its $L_1$ norm, which dimensionally 
is proportional to $2^{-h_v}$, and each $\Pf G^{h_v,T_v}({\bf t}_v)$ by the square root of its $L_\io$ norm, 
which dimensionally is proportional to $2^{h_v(\frac12\sum_{i=1}^{s_v}|P_{v_i}^\psi|-\frac12|P_v^\psi|-s_v+1)}$. The result is 
\bea && \frac1{|\L|}\int d\xx_{v_0} |W^*_{\t,{\bf P},T}(\xx_{v_0})|\le
C^{n+m}\sum_{{\bf P}}\Big[\prod_{v\,{\rm normal}\ {\rm e.p.}}(C|\l|)^{\max\{1,c|P_{v}^\psi|\}}\Big]\cdot\label{5.51s}\\
&&\cdot
\Big[\prod_{v\,{\rm not}\ {\rm e.p.}} \frac1{s_v!}2^{-h_v(s_v-1)}2^{h_v(\frac12\sum_{i=1}^{s_v}|P_{v_i}^\psi|-\frac12|P_v^\psi|-s_v+1)}\Big]
\,\Big[\prod_{v\ {\rm e.p.}}2^{h_{v'}(2-\frac12|P_v^\psi|-|P_{v}^A|)}\Big]\;.\nonumber\eea
where in the last factor $v'$ indicates the node immediately preceding $v$ on $\t$.
Note that the bound is valid provided the running coupling constants are close to 
their bare values.
The big product in the second line can be conveniently reorganized by using the identities
\bea && \hskip-.7truecm
\sum_{v\ {\rm not}\ {\rm e.p.}}h_v(s_v-1)=h(n+m-1)+\sum_{v{\rm not}\ {\rm e.p.}}(h_v-h_{v'})
(n(v)+m(v)-1)\;,\nonumber\\
&&\hskip-.7truecm  \sum_{v\ {\rm not}\ {\rm e.p.}}\hskip-.2truecm h_v\Big[\big(\sum_{i=1}^{s_v}
|P_{v_i}^\psi|\big)-|P_v^\psi|\Big]=h(|I_{v_0}^\psi|-|P^\psi_{v_0}|)+\hskip-.2truecm
\sum_{v\ {\rm not}\ {\rm e.p.} }\hskip-.2truecm
(h_v-h_{v'})(|I_v|-|P_v|)\;,\nonumber\eea
where $n(v)$ (resp. $m(v)$) is the number of normal (resp. special) endpoints following $v$ on $\t$.
In the second line, $|I_{v_0}^\psi|=\sum_{v\ {\rm e.p.}}|P_v^\psi|$, so that, plugging these identities into 
Eq.(\ref{5.51s}) and using also the fact that $\sum_{v\ {\rm e.p.}}|P_v^A|=|P_{v_0}^A|$, gives:
\bea && \frac1{|\L|}\int d\xx_{v_0} |W^*_{\t,{\bf P},T}(\xx_{v_0})|\le
C^{n+m}\Big[\prod_{v\,{\rm normal}\ {\rm e.p.}}(C|\l|)^{\max\{1,c|P_{v}^\psi|\}}\Big]\cdot\nonumber\\
&&\cdot 2^{h(2-\frac12|P_{v_0}^\psi|-
|P_{v_0}^A|)}\Big[\prod_{v\,{\rm not}\ {\rm e.p.}} 2^{(h_v-h_{v'})(2-\frac12|P_v^\psi|-|P_v^A|)}
\Big]\;.\label{5.51t}\eea
We call 
$d_v=-2+\frac{|P_v^\psi|}{2}-|P_v^A|$ 
the {\it scaling dimension}
of $v$, depending on the number of the external
fields of $v$. If $d_v<0$ for any $v$ one can
sum over $\t,{\bf P},T$ obtaining convergence
for $\l$ small enough; however 
$d_v\le 0$ when $(|P_{v}^\psi|,|P_v^A|)$ is equal either to $(2,0)$, in which case the scaling dimension 
is 1 (relevant operator), or to $(4,0)$ or to $(2,1)$, in which cases  the scaling dimension 
is 0 (marginal operator). The apparent divergences associated with the presence of 
nodes with these special fields configurations is cured by the action of the regularization operator $\RR$, as discussed in the following. 

\subsubsection{The renormalized expansion}

Of course, the expansion for $\VV^{(h)}$ described above does not lead to Eq.(\ref{5.49a}), but rather to a different expansion,
which is similar in many respects to Eq.(\ref{5.49a}), modulo the presence of a certain number of 
interpolation operators arising from the action of $\RR$ on the nodes of $\t$ that are not endpoints. 
The precise definition of these interpolation operators is rather technical, see e.g.\ Section 3 of 
\cite{BM01}. The outcome can be expressed as 
\bea && \RR\VV^{(h)}(\t,\psi^{(\le h)},\bsA)=\label{5.49as}\\
&&=\sum_{
{\bf P}\in\PP_\t}\sum_{T\in{\bf T}}\sum_{\b\in B_T}\int d\xx_{v_0}W_{\t,{\bf P},{\bf T},\b}
(\xx_{v_0})\Big[\prod_{f\in P_{v_0}^\psi}\hat \dpr^{q_\b(f)}_{j_\b(f)}\psi^{(\le h)}_{\xx_\b(f),\o(f)}\Big]A_{P_{v_0}^A}
\;,\nonumber\eea
where: $B_T$ is a set of indices which allows to distinguish the
different terms produced by the non trivial $\RR$ operations;
$\xx_\b(f)$ is a coordinate obtained by interpolating two points in $
\xx_{v_0}$, in a suitable way depending on $\b$; $q_\b(f)$ is a nonnegative 
integer $\le 2$; $j_\b(f)=1,2$ and $\hat\dpr^q_j$ is a suitable differential 
operator, dimensionally equivalent to $\dpr^q_j$ (see \cite{BM01} for a precise 
definition); $W_{\t,{\bf P},{\bf T},\b}$ is given by (defining $\SS_1=1-\PP_0$ and $\SS_2=1-\PP_0-\PP_1$, as in \cite{GM05}; 
recall that $\PP_0$ and $\PP_1$ were defined at the beginning of 
Sec.\ref{sec3.loc}):
\bea && W_{\t,{\bf P},{\bf T},\b}(\xx_{v_0})=
\Big[\prod_{v\ {\rm e.p.}}  d^{b_\b(v)}_{j_\b(v)}(\xx_{v,\b}))\PP_{I_\b(v)}^{C_\b(v)}
\SS_{i_\b(v)}^{c_\b(v)}K^{(h_v)}_{v}(\xx_{v})\Big]\cdot\nonumber\\
&&\cdot
\Big\{\prod_{v\ {\rm not}\ {\rm e.p.}}\frac1{s_v!} \int
dP_{T_v}({\bf t}_v) \PP_{I_\b(v)}^{C_\b(v)}
\SS_{i_\b(v)}^{c_\b(v)} \Pf G^{h_v,T_v}({\bf t}_v)\cdot\nonumber\\
&&\cdot
\Big[\prod_{\ell\in T_v} \hat\dpr^{q_\b(f^1_\ell)}_{j_\b(f^1_\ell)}
\hat\dpr^{q_\b(f^2_\ell)}_{j_\b(f^2_\ell)}[d^{b_\b(\ell)}_{j_\b(\ell)}(\xx_\ell-\xx'_\ell) 
\PP_{I_\b(\ell)}^{C_\b(\ell)}
\SS_{i_\b(\ell)}^{c_\b(\ell)}g^{(h_v)}_\ell\Big]\Big\}\label{5.55s}\eea
where: $d_{j_\b(\ell)}(\xx_\ell-\xx'_\ell)$
is a suitable interpolation operators, dimensionally equivalent to $|\xx_\ell-\xx'_\ell|$ (and similarly 
$d_{j_\b(v)}(\xx_{v,\b}))$, is dimensionally equivalent to the distance 
operator among two distinct points in $\xx_v$); the indices $b_\b(v)$, $b_\b(l)$, $q_\b(f^1_\ell)$ and $q_\b(f^2_\ell)$ are
nonnegative integers $\le 2$; $j_\b(v)$, $j_\b(f^1_\ell)$, $j_\b(f^2_\ell)$
and $j_\b(\ell)$ can be either 
$1$ or $2$; $i_\b(v)$ and $i_\b(\ell)$ can be either $1$ or $2$;
$I_\b(v)$ and $I_\b(l)$ can be either $0$ or $1$; $C_\b(v)$, $c_\b(v)$, $C_\b(l)$
and $c_\b(l)$ can be either $0,1$ and $\max\{C_\b(v)+c_\b(v),C_\b(l)+ c_\b(l\})\le 1$;
$G^{h_v,T_v}_\b({\bf t}_v)$ is obtained from $G^{h_v,T_v}({\bf t}_v)$
by substituting the element $t_{i(f),i(f')}g^{(h_v)}(f,f')$ with 
$t_{i(f),i(f')}\hat\dpr^{q_\b(f)}_{j_\b(f)}
\hat\dpr^{q_\b(f')}_{j_\b(f')}g^{(h_v)}(f,f')$.

It would be very difficult to give a precise description
of the various contributions of the sum over $B_T$: however, we only need to know that the 
number of ``zeros" $d_{j_\b(\ell)}(\xx_\ell-\xx'_\ell)$, the number of ``derivatives" $\hat\dpr_{j_\b(f)}$
and the scale of the propagators they act on are, by construction, 
such that all the potentially dangerous nodes $v$ that are not e.p.\ 
(i.e. the nodes $v$ that are not e.p.\ and with $(P_{v}^\psi,P_v^A)$  equal either to $(2,0)$, or $(4,0)$, or to $(2,1)$) gain a ``scale jump" $2^{-(h_v-h_{v'})z(P_v)}$, where 
\be z(P_v)=\begin{cases}
1 & {\rm if}\ (|P_v^\psi|,|P_v^A|)=(4,0)\;,\\
2 & {\rm if}\ (|P_v^\psi|,|P_v^A|)=(2,0)\;,\\
1 & {\rm if}\ (|P_v^\psi|,|P_v^A|)=(2,1)\;,\\
0 & {\rm otherwise}\end{cases}\label{5.57s}\ee
which is enough to cure the apparent divergences present in Eq.(\ref{5.51t}). For more details, 
see e.g.\ \cite{BM01,GM01,GM05}, see also \cite[Sec.2.2]{GMP10b} for a recent pedagogical description of this point. 

In conclusion, the kernels of the effective potential are bounded in 
a way completely analogous to Eq.(\ref{5.51t}), modulo the improved dimensional factors induced by the 
action of the $\RR$ operators:
\bea && \frac1{|\L|}\sum_{\b\in B_T}\int d\xx_{v_0} |W_{\t,{\bf P},T,\b}(\xx_{v_0})|\le
C^{n+m}\Big[\prod_{v\,{\rm normal}\ {\rm e.p.}}(C|\l|)^{\max\{1,c|P_{v}^\psi|\}}\Big]\cdot\nonumber\\
&&\cdot 2^{h(2-\frac12|P_{v_0}^\psi|-
|P_{v_0}^A|)}\Big[\prod_{v\,{\rm not}\ {\rm e.p.}} 2^{(h_v-h_{v'})(2-\frac12|P_v^\psi|-|P_v^A|-z(P_v))}
\Big]\;,\label{5.51tv}\eea
which is valid {\it provided that the running coupling constants stay close to their values on scale 
$h=N$}, for all scales between $h$ and $N$. Under this assumption, Eq.(\ref{5.51tv})  implies the analyticity of the 
kernels of $\VV^{(h)}$ and $\BB^{(h)}$ and the decay bounds Eq.(\ref{3.5}).

An immediate corollary of the bound Eq.(\ref{5.51tv}) is
that contributions from trees $\t\in\TT_{N;n,m}^{(h)}$ with a vertex $v$ on
scale $h_v=k>h$ admit an improved bound with respect to
Eq.(\ref{3.5}), with an extra dimensional factor $2^{\th(h-k)}$, $0<\th<1$,
which can be thought of as a dimensional gain with
respect to the ``basic'' dimensional bound in Eq.(\ref{3.5}). This
improved bound is usually referred to as the {\it short memory}
property (i.e., long trees are exponentially suppressed); it is due 
to the fact that the renormalized scaling dimensions $d_v=2-\frac12|P_v^\psi|-|P_v^A|-z(P_v)$
in Eq.(\ref{5.51tv})  are all negative, and can be obtained by taking a fraction of 
the factors $2^{(h_v-h_{v'})d_v}$ associated to the branches of the tree
$\t$ on the path connecting the vertex on scale $k$ to the one on scale $h$.

Under the same assumptions, the beta function itself, $\b_h^\#$, is analytic and dimensionally bounded 
by a constant independent of $h$. Moreover, the contributions to it from trees that have at least one node 
on scale $k>h$ is dimensional bounded proportionally to $2^{\th(h-k)}$, with $0<\th<1$.
It is remarkable that thanks to these bounds, the dynamical system 
induced by the beta function can be fully studied and shown to lead to a bounded 
flow of the running coupling constants, as discussed in the next section.

\subsection{The flow of the running coupling constants}

As announced above, the flow of the running coupling constants is controlled by the equations:
\bea && Z_{h-1}=Z_h+\b^{Z}_h\;,\qquad 
\frac{\s_{h-1}}{\s_h}=1+\b^{\s}_h\;,\nonumber\\
&&\n_{h-1}=2\n_h+\b^\n_h\;,\qquad  \frac{Z^{(1)}_{h-1}}{Z^{(1)}_h}=1+\b^{Z,1}_h\;,\label{ppbis} \eea
where $\b_h^\#=\b_h^\#\big((Z_h,\s_h,\n_h,Z^{(1)}_h),\ldots,
(Z_N,\s_N,\n_N,Z^{(1)}_N)\big)$ is an analytic function of its argument, with an analyticity 
domain bounded by: $|Z_k-1|+|Z_k^{(1)}-1|+|\n_k|+\big|\frac{\s_k}{\s}-1\big|\le \e_0$, for all 
$h\le k\le N$, where $\e_0$ is a suitable (small) positive constant. 

Let us start by studying the flow of $Z_h$ and $\n_h$. By construction,
their beta functions, $\b_h^Z$ and $\b_h^\n$, are independent of $Z^{(1)}_h$ and $\s_h$ and, in 
particular, they 
are independent of the infrared cutoff scale $h_\s$, for all $h>h_m$. Therefore, the two coupled 
equations for $Z_h$ and $\n_h$ can be naturally studied for all $h\le N$, without any infrared cutoff.
Note also that both $\b_h^Z$ and $\b_h^\n$ can be expressed as sums over GN trees with at least two endpoints, and with at least 
one endpoint on scale $N+2$, the reason being that the local part of the trees with only normal 
endpoints of scale $\le N+1$ is zero by the support properties of the single-scale propagators
that enter the definition of $\b_h^\#$.
Therefore, by the short memory property, $|\b_h^Z|,|\b_h^\n|\le C_\th|\l|2^{\th(h-N)}$.

In order to solve the two coupled equations for $Z_h$ and $\n_h$ we use a fixed point argument,
following the same strategy as 
\cite[Appendix A5]{GM05}. Let  $\mathfrak{M}_{N;K,\th}$ be the space 
of sequences $\underline\n=\{\n_h\}_{h\le N}$ such that $|\n_h|\le K|\l|2^{-\th(h-N)}$, 
$\forall h\le N$; we shall think $\mathfrak{M}_{N;K,\th}$
as a Banach space with norm $||\cdot||_\th$, where $||\underline\n||_{\th}=\sup_{k\le N}2^{\th(N-k)}$. 
Given $\underline\n\in\mathfrak{M}_{N;K,\th}$, we first solve the equation for $Z_h$, 
$Z_{h-1}=Z_h+\b^{Z}_h$, which leads to a bounded and exponentially 
convergent flow. In fact, using the fact that $|\b^{Z}_h|\le C|\l|2^{\th(h-N)}$, 
it is straightforward to check inductively that, if $Z_N=1$, 
\be |Z_h(\underline\n)-1|\le C|\l|\;,\qquad |Z_h(\underline\n)-Z_{-\io}(\underline\n)|\le C|\l|2^{\th h}\label{3.53a}\ee
uniformly in $\underline\n$ 
for $\underline\n\in\mathfrak{M}_{N;K,\th}$. Once that $Z_h=Z_h(\underline\n)$ is fixed as a function of 
$\underline\n$, the beta function for $\n_h$ can be thought of as function of $\underline\n$ only,
and we shall write $\b_h^\n(\underline\n)$. 
The goal is to fix $\n_N$ in a smart way, so that the flow generated by $\b_h^\n(\underline\n)$ with 
initial datum $\n_N$ generates a sequence in $\mathfrak{M}_{N;K,\th}$. We note that 
the solution to the flow equation induced by the beta function with initial condition 
$\lim_{h\to-\io}\n_h=0$ is a fixed point of the map  $\mathbf{T}:
\mathfrak{M}_{N;K,\th} \to \mathfrak{M}_{N;K,\th} $ defined by
\be(\mathbf{T} \underline \n)_h = -\sum_{j\le h} 2^{j-h-1} \b_j^\n(\underline\n).\label{3.map}\ee
The fact that, for $K$ sufficiently large, $\mathbf{T} $ is a map from $\mathfrak{M}_{N;K,\th}$ to itself 
is a simple consequence 
of the bound $|\b_h^\n|\le C|\l|2^{\th(h-N)}$. More interestingly, we can prove that $\mathbf{T} $
is a contraction map on $\mathfrak{M}_{N;K,\th}$. In fact, using the expansion in GN tress, the short 
memory property and the fact that the trees contributing to $\b_h^\n$ have at least one endpoint on 
scale $N+2$, we find that if $\underline\n,\underline\n'\in\mathfrak{M}_{N;K,\th}$, then
\be |\b_h^\n(\underline\n)-\b_h^\n(\underline\n')|\le C|\l|\sum_{k\ge h}|\n_k-\n_k'|2^{\th(h-N)}\le
C'|\l|2^{\th(h-N)}||\underline\n-\underline\n'||_\th\;,\ee
which implies, if plugged back into Eq.(\ref{3.map}), that  
\be ||\mathbf{T} \underline \n -  \mathbf{T} \underline \n ' ||_{\th} 
 \le c'' |\l| \| \underline \n -  \underline \n ' \|_{\th}\;,\ee
i.e.\ $\mathbf{T}$ is a contraction for $|\l|$ sufficiently 
 small.  Then the Banach fixed point theorem shows that $\mathbf{T}$ has a unique fixed point 
 $\underline{\n}^*$ in  $\mathfrak{M}_{N;K,\th}$, such that 
 \be \n_N=\n^*_N=-\sum_{j\le N} 2^{j-N-1} \b_j^\n(\underline\n^*)\;,\label{3.58s}\ee
 which is of order $\l$, as it should be. In fact, plugging Eq.(\ref{3.58s}) into Eq.(\ref{3.11})
 and using the relation of $\n$ with the perturbed critical temperature, Eq.(\ref{2.counter}),
 shows that the distance between the interacting and the free critical temperatures is of order $\l$,
 as expected. A close examination of the proof also shows that the resulting $\n$ is exactly independent 
 of $a$ (rather than being ``just" bounded by $C|\l|$, independently of $a$). Similarly, 
 the limiting value of $Z_h=Z_h(\underline\n^*)$ as $h\to-\io$, to be called simply $Z_{-\io}(\l)$, 
 is exactly independent of $a$. 
 
 Once that the flows of $Z_h$ and $\n_h$ have been controlled, the flows of $\s_h$ and $Z^{(1)}_h$
 can be studied easily, in the same fashion as the flow of $Z_h$ at $\underline\n\in\mathfrak{M}_{N;K,\th}$ 
fixed. First of all, note that by construction both $\b_h^\s$ and $\b_h^{Z,1}$ are independent of 
 $\s_h$ and, in particular, they
are independent of the infrared cutoff scale $h_\s$, for all $h>h_\s$. Therefore, also the two 
equations for $\frac{\s_{h-1}}{\s_h}$ and $\frac{Z_{h-1}^{(1)}}{Z_h^{(1)}}$ can be naturally studied 
for all $h\le N$, without any infrared cutoff. Moreover, $\b_h^\s$, in complete analogy with the beta 
functions for $Z_h$ and $\n_h$, can be expressed as sums over GN trees with at least two endpoints, 
and with at least 
one endpoint on scale $N+2$; similarly, $\b_h^{Z,1}$ is a sum over GN trees with exactly one special 
endpoint (on an arbitrary scale)
and at least one normal endpoint on scale $N+2$. 
Therefore, by the short memory property, $|\b_h^\s|,|\b_h^{Z,1}|\le C_\th|\l|2^{\th(h-N)}$, from which
we get the analog of Eq.(\ref{3.53a}) (recall that $\s_N=\s$ and $Z^{(1)}_N=1$):
\bea 
 &&|Z_h^{(1)}-1|\le C|\l|\;,\qquad |Z_h^{(1)}-Z_{-\io}^{(1)}|\le C|\l|2^{\th h}\;,\\
 && \big|\frac{\s_h}{\s}-1\big|\le C|\l|\;,\label{3.60uy}\qquad 
 \big|\frac{\s_h}{\s}-Z^*(\l)\big|\le C|\l|2^{\th h}\;,\label{3.61uy}\eea
where $Z^{(1)}_{-\io}(\l)$ and $Z^*(\l)$ are two analytic functions of $\l$, 
analytically close to 1. As already observed for $Z_{-\io}(\l)$, these two functions are independent 
of $a$. 

\section{The renormalized correlation functions}\label{sec:correlations}

We are left with computing the renormalized $m$-points energy correlation functions and proving the 
main results stated in Section \ref{sec1} and summarized in Theorems \ref{mainprop} and \ref{thm2}.
We closely follow \cite[Section 2.3]{BFM07}. 
Recall that the $m$-point correlation in the thermodynamic limit and at distinct points $\xx_1,\ldots,\xx_m$, 
can be obtained by performing the functional derivative of the generating function $\Xi_{-,-}(\bsA)$ 
constructed in the previous sections (here $\b=\b(a)$, fixed so that $2(t-t_c(0))/t=a \s+4\p\n$, 
where $t=\tanh(\b(a)J)$, $\n$ is the temperature counterterm fixed as explained in the previous 
section and $\s=\s(a)$ is an a priori prescribed mass, such that $\lim_{a\to 0}\s(a)=m^*\in\mathbb R$):
\be \media{\e_{\xx_1,j_1};\cdots;\e_{\xx_m,j_m}}^T_{\b}=\lim_{L\to\infty}
a^{-2m}\frac{\dpr^m}{\dpr A_{\xx_1,j_1}
\cdots\dpr A_{\xx_m,j_m}}\log \Xi_{-,-}({\boldsymbol A})\Big|_{{\bsA}={\bf 0}}\;.\label{4.1}\ee
Performing this functional derivative is just a way of extracting from the generating function the 
kernel associated with the monomial $A_{\xx_1,j_1}
\cdots A_{\xx_m,j_m}$ in $\log\Xi_{-,-}(\bsA)=\sum_{h=h_\s-1}^N\lis\SS_h(\bsA)$,
where $\lis \SS_h(\bsA)$ with $h<N$ is defined by Eq.(\ref{3.16}) and 
$\lis\SS_N(\bsA):=\SS^{(N)}(\bsA)$, see Eq.(\ref{3.3}). The function $\lis\SS_h(\bsA)$ can be written
in a way similar to Eq.(\ref{3.5q}):
\be \lis\SS_h(\bsA)=\sum_{m\ge 1}\sum_{j_1,\ldots,j_m}\int d\xx_1\cdots d\xx_m S^{(h)}_{m;{\underline j}}(\xx_1,\ldots,\xx_m)\prod_{i=1}^m A_{\xx_i,j_i}\;,
\label{3.5qk}\ee
so that
\be \media{\e_{\xx_1,j_1};\cdots;\e_{\xx_m,j_m}}^T_{\b}=m!\sum_{h=h_\s-1}^N S^{(h)}_{m;{\underline j}}(\xx_1,\ldots,\xx_m)\ee
and the thermodynamic limit $L\to\infty$  is understood. Note the combinatorial factor $m!$ arising from the action 
of the derivatives w.r.t.\ $A_{\xx_i,j_i}$ on $ \lis\SS_h(\bsA)$. 
As explained in the previous sections, 
$S^{(h)}_{m;{\underline j}}(\xx_1,\ldots,\xx_m)$ can be expressed by a sum over trees $\t\in\TT^{(h)}_{N;n,m_0}$
with $m_0\le m$ {\it special endpoints} and $n\ge 0$ normal end-points, and over 
a suitable set of field labels ${\bf P}\in\PP_\t$; $m_0<m$ 
corresponds to the case where there is at least one normal endpoint $v$ on scale $N+2$ 
such that $P_v^A\neq\emptyset$.  In formulae, we can write:
\bea &&  \media{\e_{\xx_1,j_1};\cdots;\e_{\xx_m,j_m}}^T_{\b}=\label{4.1a}\\
&&=m!\sum_{h=h_\s-1}^N \sum_{n\ge 0}\sum_{m_0=0}^{m}
\,\sum_{\t\in\TT^{(h)}_{N;n,m_0}}\,\sum_{\substack{{\bf P}\in\PP_\t:\\
|P_{v_0}|=|P_{v_0}^A|=m}}S_{\t,{\bf P}}(\xx_1,j_1;\cdots;\xx_m,j_m)\;,\nonumber\eea
where $S_{\t,{\bf P}}(\xx_1,j_1;\cdots;\xx_m,j_m)$ can be represented as in Eqs.(\ref{5.49as})-(\ref{5.55s}), and bounded in a way similar to Eq.(\ref{5.51tv}). 
Given $\t\in\TT^{(h)}_{N;n,m_0}$ and 
${\bf P}\in\PP_\t$, let $E^*_{\t,{\bf P}}$ be the set of endpoints $v$ in $\t$ 
such that $P_v^A\neq\emptyset$ (i.e. $E^*_{\t,{\bf P}}$ contains all the special endpoints of $\t$ plus
the normal endpoints $v$ on scale $N+2$ such that $|P_v^A|\ge 1$); let us denote by $\tau^*$
the minimal subtree of $\t$ connecting all the endpoints in $E^*_{\t,{\bf P}}$.
For each $v \in \tau$, if $m_v:=|P_v^A|$, let $s_v^*$ the number of vertices
immediately following $v$ with $m_v \ge 1$. Moreover, let
$V_{nt}(\tau^*)$ be the set of vertices in $\t^*$ with $s_v^*>1$, which are the branching points 
of $\t^*$. For future reference, we also define $v_0^*$ to be the leftmost vertex on $\t^*$ 
and $h_0^*$ its scale. See Fig.\ref{specialTreeFig}.

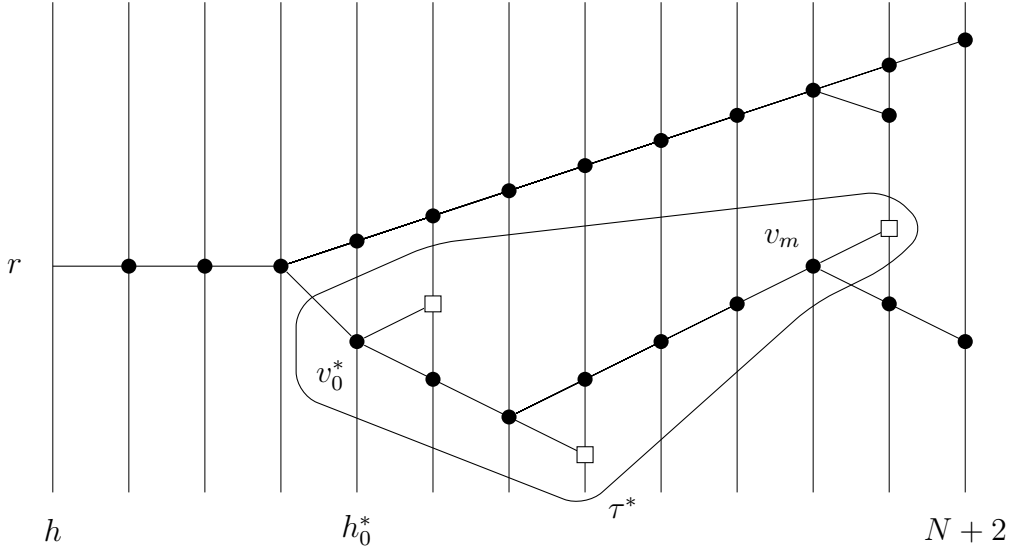
\begin{figure}
\centering
\begin{tikzpicture}
    \foreach \x in {0,...,12}
    {
        \draw[very thin] (\x ,1) -- (\x , 7.5);
    }
    \draw (0,0.5) node {$h$};
      \draw (-0.5,4) node {$r$};
    \draw (12,0.5) node {$N+2$};
    \draw (0,4) node (v0) {} -- (1,4) node  [vertex] {} -- (2,4) node [vertex] {} -- (3,4) node [vertex] (v1) {};
    \node [vertex] (v0star) at (4,3) [label=-95:$v_0^*$] {};
    \draw (v1) -- (v0star);
    \draw (v1) \foreach \x in {4,...,12} {-- (\x,3+\x/3) node [vertex] {}};
    \draw (10,19/3) -- (11,6) node [vertex] {};
    \draw (v0star) --  (5,3.5) node [specialEP] (x) {};
    \node at (4,0.5) {$h_0^*$};
    \draw (v0star) -- (5,2.5) node [vertex] {} -- (6,2) node [vertex] (v2) {} -- (7,1.5) node [specialEP] {};
    \draw (v2) \foreach \x in {7,8,9} {-- (\x,-1+\x/2) node [vertex] {}} -- (10,4) node (v3) [vertex,label=100:$v_m$] {};
    \draw (v3) -- (11,4.5) node [specialEP] {};
    \draw (v3) -- (11,3.5) node [vertex] {} -- (12,3) node [vertex] {};
    \draw[rounded corners=0.3cm] (3.2,2.3) -- (7,0.8) -- (10,3.5) -- (11,4) -- (11.5,4.5) -- (11,5) -- (5,4.3) --(3.2,3.5) -- cycle;
    \node at (7.5,0.8) {$\tau^*$};
\end{tikzpicture}
\caption{Example of a tree $\tau\in \TT^{(h)}_{N;n,m_0}$ appearing in the expansion 
for the $m$ points correlation function, with $n=3$ and $m_0=m=3$. The subtree $\tau^*$ 
associated with $\t$ and with a choice of ${\bf P}$ such that $|P_{\bar v_0}|=|P_{\bar v_0}^A|=3$ 
is highlighted.} \label{specialTreeFig}
\end{figure}

Given these definitions and recalling that 
$d_v(P_v)=2-\frac12|P_v^\psi|-|P_v^A|-z(P_v)$, we can write the bound for 
$S_{\t,{\bf P}}(\xx_1,j_1;\cdots;\xx_m,j_m)$ as
\bea && |S_{\t,{\bf P}}(\xx_1,j_1;\cdots;\xx_m,j_m)|\le
C^{n+m}2^{h(2-m)}\Big[\!\!\!\prod_{v \in V_{nt}(\tau^*)}\!\!\! \!\! 2^{2(s_v^*-1)h_v} e^{-c 2^{h_v}\delta_v} \Big]
\cdot\label{4.basic}\\
&&\cdot\Big[\prod_{v\,{\rm not}\,{\rm e.p.}} 2^{(h_v-h_{v'})d_v(P_v)}\Big]
  \prod_{v \textup{ e.p.} } \left\{ \begin{array}{ll} |\l|^{c|P_{v}^\psi|}, & 
 {\rm if}\ v\ {\rm is}\ {\rm normal}, \ h_v=N+2 
\\ |\l|2^{\th (h_v-N)}, & 
{\rm if}\ v\ {\rm is}\ {\rm normal}, \ h_v<N+2\\
  1,& {\rm if}\ v\ {\rm is}\ {\rm special}\end{array} \right.
\nonumber\eea
This estimate is very similar to the bound Eq.(\ref{5.51tv}) for the renormalized kernels of the effective potential,
the most important difference consisting in the product 
$\prod_{v \in V_{nt}(\tau^*)}  2^{2(s_v^*-1)h_v} e^{-c 2^{h_v}\delta_v}$, where 
$\d_v=\d_v(\xx_1,\ldots,\xx_m)$ is the tree distance of the set $\xx^*_v:=\cup_{f\in P_v^A}\{\xx(f)\}$,
i.e. the length of the shortest tree graph on $\mathbb R^2$ connecting the points of $\xx^*_v$.
In this product, the factors  $2^{2(s_v^*-1)h_v}$ take into account the dimensional gain coming from the 
fact that we are {\it not} integrating over the space labels $\xx_i$ of the external fields (the gain is meant in 
comparison with 
Eq.(\ref{5.51tv}) where, on the contrary, we integrated over all the field variables); moreover, the factors 
$e^{-c 2^{h_v}\delta_v}$ come from the decaying factors $e^{-c2^{h_w}|\xx(\ell)-\xx'(\ell)|}$ associated with 
the propagators $g_\ell^{h_w}$ on the portion of spanning trees inside the cluster $v$, that is 
from the lines $\ell\in T_w$, $w\ge v$, see \cite[Section 2.3]{BFM07} for a few more details about how to 
extract these factors starting from the bounds on the kernels of the effective potential. Moreover, 
note that in the second line we estimated  the contribution from 
each endpoint of type $\n$ by $({\rm const.})|\l|2^{\th (h-N)}$, consistently with the bounds on $\n_h$ derived in the previous section.

The bound Eq.(\ref{4.basic}) is one of the basic ingredients from which the main results of this paper follow. 
It makes apparent that all the terms with at least one endpoint on scale $N$ or with one endpoint of type 
$\n$ are strongly suppressed, in the sense that they are proportional to a short memory factor 
$2^{\th (h-N)}$, which go to zero in the limit $N\to\infty$ (which corresponds to the scaling limit 
$a\to 0$): these are the 
terms leading to the correction term $R^{(a)}(\xx_1,j_1;\cdots;\xx_m,j_m)$ in  Theorem \ref{thm2}, as 
explained in more detail below. The dominant terms are those coming from trees $\t\in \TT^{(h)}_{N;0,m}$, 
which are of order 1 with respect 
to $\l$ (because the size of the special endpoints is measured by $Z^{(1)}_h$, which is analytically close to 1),
and whose dimensional bound does not vanish in the $N\to\infty$ limit. In the following, we show that: (1) 
the contributions to the $m$ point function from trees $\t\in \TT^{(h)}_{N;0,m}$ can be written 
as the expression in the second line of Eq.(\ref{1.43w}) plus a remainder, bounded by the r.h.s.\ of Eq.(\ref{10b}); 
(2) the sum of the contributions
from all the trees with $n\ge 1$ or $m_0<m$ can be bounded by the r.h.s.\ of Eq.(\ref{10b}). This will conclude the 
proof of the main results stated in the introduction.

\subsection{The dominant contributions}

As discussed above, the dominant contribution to the $m$ point energy correlation function come from trees 
$\t\in\TT^{(h)}_{N;0,m}$. Among these, we further distinguish the contributions from trees in $\TT^{(h)}_{N;0,m}$
with at least one endpoint on scale $N+2$, which will be discussed at the end of this subsection, and those
with all the endpoints on scales $<N+2$, which can be written as 
\be K(\xx_1,\ldots,\xx_m):=m!\!\!\sum_{h=h_\s-1}^{N}\,\sum^*_{\t\in\TT^{(h)}_{N;0,m}}\,
\sum_{\substack{{\bf P}\in\PP_\t:\\
|P_{v_0}|=|P_{v_0}^A|=m}}\!\!\!S_{\t,{\bf P}}(\xx_1,j_1;\cdots;\xx_m,j_m)\;,\label{4.3}\ee
where the $*$ on the sum indicates the constraint that all the endpoints are on scale $<N+2$. By summing over 
trees, and by explicitly computing the tree values, $K(\xx_1,\ldots,\xx_m)$  can be written in a way very similar to Eq.(\ref{1.maina}): 
\bea&&
K(\xx_1,\ldots,\xx_m)= -\frac{1}{2m}\Big(\frac{i}\p\Big)^m\cdot\label{eq:qFree_corr}
 \\
&&\qquad \cdot\sum_{\p\in\Pi_{\{1,\ldots,m\}}}\sum_{\substack{\o_1,\ldots,\o_m\\ h_1,\ldots, h_m}}
\Big[\prod_{j=1}^{m}\o_jZ^{(1)}_{\min\{h_j,h_{j+1}\}}
{g}^{(h_j)}_{\o_j,-\o_{j+1}}(\xx_{\p(j)}-\xx_{\p(j+1)})\Big]\;,\nonumber\eea
where the sum over the $h_j$'s runs over $h_\s\le h_j\le N$, with the convention that $g^{(h_\s)}$ is the 
propagator of $\psi^{(\le h_\s)}$.
Equation~\eqref{eq:qFree_corr} differs from the corresponding expression for the free system in the renormalization 
of the special endpoints, $Z^{(1)}_{\min\{h_j,h_{j+1}\}}$, and in the propagators. 

\begin{figure}
\centering
\subfloat[]{
\begin{tikzpicture}[decoration=snake] 
\node[specialEP,label=below right:$\xx_1$] (x1) at (0:2) {};
\node[specialEP,label=below:$\xx_2$] (x2) at (60:2) {};
\node[specialEP,label=below:$\xx_3$] (x3) at (120:2) {};
\node[specialEP,label=below left:$\xx_4$] (x4) at (180:2) {};
\node[specialEP,label=above:$\xx_5$] (x5) at (240:2) {};
\node[specialEP,label=above:$\xx_6$] (x6) at (300:2) {};
\draw (x1) .. controls +(90:0.75) and+(330:0.75) .. (x2) node[sloped,midway, above] {$h_1$};
\draw (x1) .. controls +(270:0.75) and+(30:0.75).. (x6) node[sloped,midway, above] {$h_0^*$} ; 
\draw (x5) .. controls +(330:0.75)  and +(210:0.75) .. (x6) node[sloped,midway, above] {$h_3$};
\draw (x2) .. controls +(150:0.75) and +(30:0.75) .. (x3) node[sloped,midway, above] {$h+1$} ;
\draw (x3) .. controls +(210:0.75) and +(90:0.75) .. (x4) node[sloped,midway, above] {$h_2$};
\draw (x4) .. controls +(270:0.75) and +(150:0.75) .. (x5) node[sloped,midway, above] {$h_4$};
\draw[decorate] (x1) -- + (0:1);
\draw[decorate] (x2) -- +(60:1);
\draw[decorate] (x3) -- +(120:1);
\draw[decorate] (x4) -- + (180:1.5);
\draw[decorate] (x5) -- + (240:1.5);
\draw[decorate] (x6) -- + (300:1.5);
\draw[rounded corners=0.2cm] (0.5,1.8) --(1,2.2) -- (2.5,0) -- (2,-0.5) -- cycle;
\draw[rounded corners=0.2cm] (-1,-2) -- (-2,-1) -- (-2.3,0) -- (-2,0.4) -- (-0.5,-1.8) -- cycle;
\draw[rounded corners=0.2cm] (1,-2.2) -- (0,-2.4) -- (-2,-2) -- (-2.8,0) -- (-2,0.6) -- (1.2,-1) -- (1.5,-1.5)-- cycle;
\draw[rounded corners=0.2cm] (1.5,-0.5) -- (-0.5,2) -- (-1.5,2.3) -- (-2.5,1.5)-- (-3,0) -- (-2.4,-2.4) -- (1,-2.7) -- (2,-1) -- cycle;
\end{tikzpicture}
}
\\
\subfloat[]{
\begin{tikzpicture}
    \foreach \x in {0,...,12}
    {
        \draw[very thin] (\x , 1) -- (\x , 7);
    }
    \draw (0,0.5) node (h) {$h$};
	\draw (h) +(5,0) node {$h_1$};
	\draw (h)+(6,0) node {$h_2$};
	\draw (h)+(8,0) node {$h_3$};
	\draw (h)+(10,0) node {$h_4$};
 \draw (h) + (4,0) node {$h_0^*$};
     \draw (-0.5,4) node {$r$};
    \draw (h) + (12,0) node {$N+2$};
    \draw (0,4) node (v0) {} -- (1,4) node  [vertex] {} -- (2,4) node [vertex] {} -- (3,4) node [vertex] (v1) {};
    \node [vertex] (v0star) at (4,4) [label=-95:$v_0^*$] {};
    \draw (v1) -- (v0star);
    \draw (v0star) --  (5,5) node [vertex] (x) {};
    \draw (5,5) -- ++ (1,1) node [specialEP,label=0:$\xx_1$] {};
    \draw (5,5) -- ++ (1,-1) node [specialEP,label=0:$\xx_2$] {};
    \draw (v0star) -- (5,3) node [vertex] {} -- (6,2) node [vertex] (v2) {} -- (7,1.5) node [specialEP,label=0:$\xx_3$] {};
    \draw (v2) \foreach \x in {7,8,9} {-- (\x,-1+\x/2) node [vertex] {}} -- (10,4) node (v3) [vertex] {};
   \draw (8,3) -- ++(1,-0.5) node [specialEP,label=0:$\xx_6$] {};
    \draw (v3) -- (11,4.5) node [specialEP,label=0:$\xx_4$] {};
    \draw (v3) -- (11,3.5) node [specialEP,label=0:$\xx_5$] {};
\end{tikzpicture}
}
\caption{(a) A diagram representing one of the terms in the expression in Equation~\eqref{4.3} for $K(\xx_1,\ldots,\xx_6)$, 
with the scale labels of the propagators shown.  The vertices are grouped into clusters according to scales of the connecting 
propagators.  (b) A GN tree $\t$ associated with this term; note that the branching points of the trees correspond to the 
clusters of vertices shown above, and that this tree is uniquely specified apart from the choice of which endpoint is 
identified with which coordinate.}\label{qfreeTreeFig}
\end{figure}
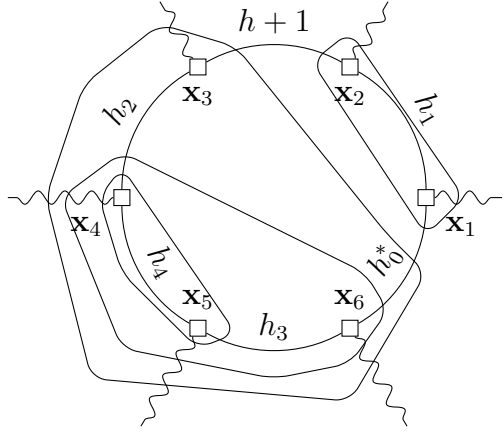
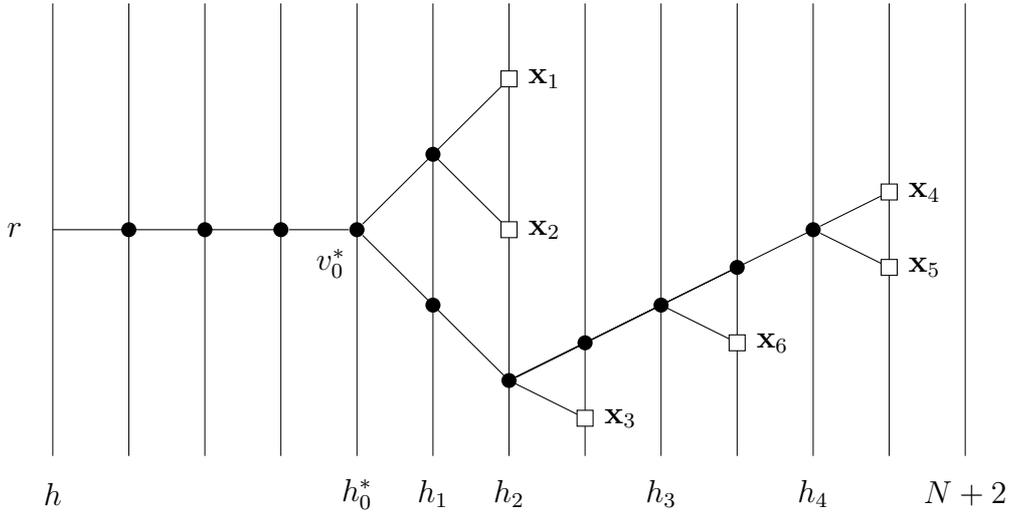

First of all, using Eq.(\ref{3.60uy}), 
we can rewrite $Z^{(1)}_{\min\{h_j,h_{j+1}\}}=Z^{(1)}_{-\io}(\l)+O(|\l|2^{\th h})$. Similarly, using 
the explicit expression for $g^{(h)}$, Eq.(\ref{3.36yt}), and the bounds 
Eqs.(\ref{3.53a}),(\ref{3.61uy}), we rewrite (after having taken the thermodynamic limit $L\to\infty$)
$g^{(h)}(\xx-\yy)={\mathfrak g}^{(h)}(\xx-\yy)+r_h(\xx-\yy)$ where, for all $h_\s<h\le N$,
\be \mathfrak g^{(h)}(\xx)=\frac{1}{Z_{-\io}(\l)}\int\limits_{[-\frac{\pi}{a},\frac{\pi}{a}]^2} 
\frac{d\kk}{2\p}\frac{e^{-i\kk\xx} f_h(\kk)}{|D_a(\kk)|^2+\big[Z^*(\l)\s\big]^2}
\begin{pmatrix}D_a^+(\kk)&iZ^*(\l) \s\\ -iZ^*(\l)\s& D_a^-(\kk)\end{pmatrix}\;,
\label{3.36ym}\ee
where $r_h$ is a correction bounded as
\be ||r_h(\xx-\yy)||\le C|\l|2^{\th(h-N)}2^{h}e^{-c 2^h|\xx-\yy|}\;.\label{4.7}\ee
If $h=h_\s$ we still write $g^{(h_\s)}(\xx)={\mathfrak g}^{(h_\s)}(\xx)+r_{h_\s}(\xx)$, with 
$r_{h_\s}$ bounded as in Eq.(\ref{4.7}) and ${\mathfrak g}^{(h_\s)}(\xx)$ defined by an expression 
analogous to Eq.(\ref{3.36ym}), with the only difference that $f_{h_\s}(\kk)$ is replaced by 
$\sum_{k\le h_\s}f_k(\kk)=e^{-2^{-2h_\s}\kk^2}$.

Correspondingly, we rewrite Eq.(\ref{eq:qFree_corr}) as the term obtained by replacing 
all the factors $Z^{(1)}_{\min\{h_i,h_{i+1}\}}$ by $Z^{(1)}_{-\io}(\l)$ and all the propagators $g^{(h)}$ by $\mathfrak g^{(h)}$ plus a remainder, 
which is the sum of all the terms involving at least one factor $Z^{(1)}_h-Z^{(1)}_{-\io}$, or one 
factor $r_h(\xx)$:
\be
K(\xx_1,\ldots,\xx_m)= K^{(0)}(\xx_1,\ldots,\xx_m)+K'(\xx_1,\ldots,\xx_m)\ee
where 
\bea && K^{(0)}(\xx_1,\ldots,\xx_m)=
-\frac{1}{2m}\Big(\frac{i}\p\Big)^m\cdot\label{eq:qFree_corr2}
 \\
&&\qquad \cdot\sum_{\p\in\Pi_{\{1,\ldots,m\}}}\sum_{\substack{\o_1,\ldots,\o_m\\ h_1,\ldots, h_m}}
\Big[\prod_{j=1}^{m}\o_jZ_{-\io}^{(1)}(\l)
{\mathfrak g}^{(h_j)}_{\o_j,-\o_{j+1}}(\xx_{\p(j)}-\xx_{\p(j+1)})\Big]\;.\nonumber\eea
Using the fact that $\sum_{h\le N}f_h(\kk)=1$ we see that Eq.(\ref{eq:qFree_corr2})
is equal to the expression in the second line of Eq.(\ref{1.43w}),
provided that $\bar Z(\l)$ in Eq.(\ref{1.42q})
is fixed to be $\bar Z(\l)=Z_{-\io}^{(1)}(\l)/Z_{-\io}(\l)$.
Moreover, using the bounds Eqs.(\ref{3.3a2}), (\ref{3.60uy}), and (\ref{4.7}) on $g^{(h)}, Z^{(1)}_h-Z^{(1)}_{-\io}$,
and $r_h$, respectively, 
the correction term $K'(\xx_1,\ldots,\xx_m)$ can be bounded as:
\be
|K'(\xx_1,\ldots,\xx_m)|\le C^m |\l|m! \sum_{h\le N} 2^h\!\!\sum_{\t\in \TT^{(h)}_{N;0,m}} \!\!
2^{\th (h_{v_m}-N)}
\Big[\!\!\prod_{v\in V_{nt}(\t)}\!\! 2^{(s_v-1) h_v} e^{-c 2^{h_v} \d_v}\Big] \;, \label{alpha_bound}\ee
where the sum over trees has been re-introduced in order to properly keep
track of the hierarchical scale structure induced by the choice of the scale labels of the propagators, as explained in 
the caption to Fig.\ref{qfreeTreeFig}.
Moreover $v_m$ is the rightmost non trivial vertex of $\t$ (e.g., $v_m={v_3}$ in the example of 
Fig.\ref{qfreeTreeFig}) and the set $V_{nt}(\t)$ is the set of non trivial vertices of $\t$. 
Note that $\sum_{v\in V_{nt}(\t)}2^{h_v}\d_v\ge 2^{h_0^*}D$, where $h_0^*:=h_{v_0^*}$ and
$D=D_{\underline\xx}$ is the diameter of $\underline\xx=\{\xx_1,\ldots,\xx_m\}$. Therefore, 
\bea
|K'(\xx_1,\ldots,\xx_m)|&\le& C^m |\l|m! \sum_{h\le N} 2^h\!\!\sum_{\t\in \TT^{(h)}_{N;0,m}} e^{-c'2^{h_0^*}D}
2^{\th (h_{v_m}-N)}\cdot\nonumber\\
&&\cdot
\Big[\!\!\prod_{v\in V_{nt}(\t)}\!\! 2^{(s_v-1) h_v} e^{-c' 2^{h_v} \d_v}\Big] \;, \label{alpha_bound_bis}\eea
where $c'$ can be chosen to be half of the constant $c$ in Eq.(\ref{alpha_bound}).
We then separate the sum over $\t$ into a sum over the various scale labels, and a sum over the remaining
structure of the tree, which is indexed by an unlabeled tree $t$ having the same endpoints as $\t$.
We denote by $\mathfrak{T}_{0,m}$ the set of unlabeled trees with $n=0$ normal endpoints and $m$ special 
endpoints. 
Then introducing the variable $ j = h-h_0^*$, ignoring all the constraints in the sum over scales other than 
$h< h_0^*$,
\bea && |K'(\xx_1,\ldots,\xx_m)|\le C^m |\l|m!  2^{-\th N}\sum_{t\in{\mathfrak T}_{0,m}} 
\sum_{j<0} 2^{j} \sum_{h_0^*=-\io}^{\io} 2^{s_{v_0^*} h_0^*} e^{-c'2^{h_0^*} (D+\d_{v_0^*})} \cdot\nonumber\\
&&\cdot\sum_{h_{v_m}=-\io}^{\io}
2^{h_{v_m}(s_{v_m}-1+\th)} e^{-c' 2^{h_{v_m}} \d_{v_m}}
\Big[\prod_{\substack{v  \in V_{nt}(\t):\\ v\neq v_0^*,v_m}}\Big(
\sum_{h_v=-\io}^\io 2^{(s_v-1) h_v} e^{-c'2^{h_v}\d_v} \Big)\Big]\;.
\nonumber\eea
Now, the sum over $j$ gives a constant and we perform the other sums as follows.
Letting $-h_{\d_v}$ be the integer part of $\log_2 {\d_v}$ so that $ 2^{-h_{\d_v}}\le \d_v \le 2^{(1-h_{\d_v})}$, 
then for any $\a_v > 0$,
\be\sum_{h_v=-\io}^\io 2^{\alpha_v h_v} e^{-c' 2^{h_v}\d_v} \le 2^{\a_v}\d_v^{-\a_v} 
\sum_{h_v=-\io}^\io 2^{\a_v (h_v-h_{\d_v})} e^{-c'2^{h_v-h_{\d_v}}}=:2^{\a_v}\d_v^{-\a_v} c_{a_v}\;,
\label{alpha_sum}\ee
where 
\be
c_\a= \sum_{h=-\io}^\io 2^{\a h} e^{-c' 2^{h} }\le \int_{-\io}^\io\!\!\! dx\, 2^{\a x} e^{-c' 2^{x} } + \max_y y^{\a} e^{-c'y} = 
(c')^{-\a}\Big[\frac{\G(\a)}{\log 2}  + \left( \frac{ \a}{e} \right)^{ \a}\Big]\ee
In particular, if $\a\ge 1$, then $c_\a\le C^\a \a^\a$ for a suitable $C>0$. Therefore,
letting $\d=\d_{\underline\xx}=\min_{1 \le i < j \le m} |\xx_i - \xx_j|$ and 
using also the fact that $\d_v\ge (s_v-1)\d$, 
\bea && |K'(\xx_1,\ldots,\xx_m)|\le (C')^m |\l|m!  2^{-\th N}\sum_{t\in{\mathfrak T}_{0,m}} 
\Big(\frac{s_{v_0^*}}{D+(s_{v_0^*}-1)\d}\Big)^{s_{v_0^*}} \cdot\nonumber\\
&&\cdot \Big(\frac{s_{v_m}-1+\th}{(s_{v_m}-1)\d}\Big)^{s_{v_m}-1+\th} 
\Big[\prod_{\substack{v  \in V_{nt}(t):\\ v\neq v_0^*,v_m}}\Big(
\frac{s_{v}-1}{(s_{v}-1)\d}\Big)^{s_{v}-1}\Big]\;.
\nonumber\eea
Recall that $s_{v_0^*}\ge 2$: under this condition it is easy to check that 
\be \Big(\frac{s_{v_0^*}}{D+(s_{v_0^*}-1)\d}\Big)^{s_{v_0^*}}\le\frac{4\d^{2-s_{v_0^*}}}{D^2}\;. \ee
Finally, recalling that $\sum_{v\in V_{nt}(t)}(s_v-1)=m-1$ and noting 
that the sum over $t$ simply gives the number of unlabeled trees with $m$ 
endpoints (which is no more than $4^m$), we get
\be|K'(\xx_1,\ldots,\xx_m)|\le (C'')^m |\l|m! \frac1{\d^m} \Big(\frac{a}{\d}\Big)^{\th}\Big(\frac{\d}{D}\Big)^2\;,
\label{4.17rg}
\ee
which is the desired estimate on the contributions to the $m$ points function
coming from the trees in $\TT^{(h)}_{N;0,m}$. 

The contributions to the $m$ point function coming from trees $\t\in\TT^{(h)}_{N;0,m}$ with at least 
one special endpoint on scale $N+2$ can be bounded by a similar expression, 
with the difference that the factor $|\l|$ in the r.h.s.\ of Eq.(\ref{4.17rg}) should be replaced by 1. We do not 
belabor the details of this computation, which is completely analogous to the one leading to Eq.(\ref{4.17rg}).
Let us just remark that these are the terms responsible of the fact that the r.h.s.\ of Eq.(\ref{10b}) is of order 1 w.r.t.\ 
$\l$. Of course, if desired, these contributions can be written explicitly (possibly modulo further corrections bounded 
as in Eq.(\ref{4.17rg})), by making use of the explicit expression of the source term $\lis\BB(\psi,\c,\bsA)$, which is 
obtained from Eq.(\ref{2.en}) by performing the linear change of variables Eq.(\ref{2.lin}), and of  
the explicit expression of the propagator of the $\c$ field, Eq.(\ref{2.gchi}). If we decided to isolate these terms from 
the correction $R^{(a)}$, then the remaining contributions would be bounded as in Eq.(\ref{4.17rg}), 
as it follows from Eq.(\ref{4.17rg}) itself and from the discussion in the following subsection.

\subsection{The subdominant contributions}

It remains to bound the trees with $n \ge 1$,
\bea &&K''(\xx_1,j_1;\cdots;\xx_m,j_m):=\\
&&\quad :=m!\sum_{h=h_\s-1}^N \sum_{\substack{n\ge 1\\ m_0\le m}}
\,\sum_{\t\in\TT^{(h)}_{N;n,m_0}}\,\sum_{\substack{{\bf P}\in\PP_\t:\\
|P_{v_0}|=|P_{v_0}^A|=m}}S_{\t,{\bf P}}(\xx_1,j_1;\cdots;\xx_m,j_m)\;,\nonumber\eea
where $S_{\t,{\bf P}}(\xx_1,j_1;\cdots;\xx_m,j_m)$ can be estimated as in Eq.(\ref{4.basic}). For simplicity,
we start by looking at the contributions to $K''$ coming from trees with $m=m_0$, to be called $K'''$:
\bea &&K'''(\xx_1,j_1;\cdots;\xx_m,j_m):=\\
&&\quad :=m!\sum_{h=h_\s-1}^N \sum_{n\ge 1}
\,\sum_{\t\in\TT^{(h)}_{N;n,m}}\,\sum_{\substack{{\bf P}\in\PP_\t:\\
|P_{v_0}|=|P_{v_0}^A|=m}}S_{\t,{\bf P}}(\xx_1,j_1;\cdots;\xx_m,j_m)\;,\nonumber\eea 
which is bounded as:
\bea &&|K'''(\xx_1,j_1;\cdots;\xx_m,j_m)|\le m!\sum_{h=h_\s-1}^N \sum_{n\ge 1}
\,\sum_{\t\in\TT^{(h)}_{N;n,m}}\sum_{\substack{{\bf P}\in\PP_\t:\\
|P_{v_0}|=|P_{v_0}^A|=m}} C^{n+m}\cdot\nonumber\\
&&\hskip3.truecm\cdot 2^{h(2-m)}2^{\th (h_+-N)}
\Big[\prod_{v\in E_n(\t) }(C|\l|)^{\max\{1,c|P_v^\psi|\}}\Big]\cdot\label{4.20g}\\
&&\hskip3.truecm\cdot
\Big[\!\!\!\prod_{v \in V_{nt}(\tau^*)}\!\!\! \!\! 2^{2(s_v^*-1)h_v} e^{-c 2^{h_v}\delta_v} \Big]
\Big[\prod_{v\,{\rm not}\,{\rm e.p.}} 2^{(h_v-h_{v'})d_v(P_v)}\Big]\;,\nonumber\eea
where $h_+$ is the scale of the rightmost normal endpoint and $E_n(\t)$ is the set of normal 
endpoints of $\t$. We now want to reduce this expression to a simplified form as close as possible to 
Eq.(\ref{alpha_bound}) and then bound the simplified expression by a strategy similar to the one used in the 
previous subsection. The first step consists in ``pruning'' the tree $\tau$ of the branches that are not in $\t^*$.
To this purpose, we make the following rearrangement. Let $\mathcal{T}^{*}_{N,h_0^*;m}$ be the set of labelled
trees with $m$ special endpoints, with no normal endpoints, and whose leftmost vertex $v_0^*$ is nontrivial 
and on scale $h_0^*$. If $\t^*\in\mathcal{T}^*_{N,h_0^*;m}$ and $h<h_0^*$, 
let $\mathcal{T}_N^{(h)}(\tau^*;n)\subseteq\TT^{(h)}_{N;n,m}$ be the set of trees $\tau$ with root scale $h$ and $n$
normal endpoints such that $\tau^*$ is the subtree of $\tau$ connecting its special endpoints. Fix any $\th \in (0,1)$,
$\e \in (0,1/2)$ and $\phi \in (\max\{0,1-\th-3\e\} , 1 - \th)$ and define, recalling that $m_v=|P_v^A|$,
\be \tilde d_v=\tilde {d}_v(P_v^A) := \left\{ 
\begin{array}{ll}
1+\e -m_v, &  m_v > 1 \\
-\th - \phi, & m_v = 1 \\
0, & m_v=0
\end{array}\right.\;. \ee
It is easy to check that, if $|P_v^\psi|\ge 2$, and recalling that $d_v(P_v)=\min\{1-\frac12|P_v^\psi|-m_v,-1\}$, then
\be d_v-\tilde d_v=d_v(P_v)-\tilde d_v(P_v^A)\le -\frac16(1-\th-\phi)|P_v^\psi|-\th \d_{m_v,0}\;.\label{4.22as}\ee
It is convenient to rewrite Eq.(\ref{4.20g}) as
\bea
&&|K'''(\xx_1,j_1;\cdots;\xx_m,j_m)|\le C^m m!\sum_{h=h_\s-1}^N 2^{h(2-m)}\sum_{h<h_0^*\le N+1}\,
\sum_{\tau^* \in \mathcal{T}^*_{N,h_0^*;m}}\cdot\nonumber\\
&&\cdot2^{(1+\e-m)(h_0^*-h)}\Big[
\prod_{v \in V_{nt}(\tau^*)}
2^{2(s_v^*-1)h_v} e^{-c 2^{h_v}\delta_v} \Big]
\Big[\prod_{v \in V(\tau^*)} 2^{\tilde{d}_v(h_v-h_{v'})} \Big]
\sum_{n\ge 1}C^{n}\cdot\nonumber\\
&&\cdot\hskip-.5truecm
\sum_{\t\in\TT^{(h)}_N(\t^*;n)}\hskip-.5truecm 2^{\th (h_+-N)}\hskip-.5truecm
\sum_{\substack{{\bf P}\in\PP_\t:\\
|P_{v_0}|=|P_{v_0}^A|=m}}\hskip-.5truecm
\Big[\prod_{v\in V(\t)} 2^{(d_v-\tilde d_v)(h_v-h_{v'})}\Big]
\Big[\prod_{v\in E_n(\t) }(C|\l|)^{\max\{1,c|P_v^\psi|\}}\Big]\;,\nonumber
\eea
where $V(\t)$ is the set of vertices of $\t$ that are not endpoints, and similarly for $V(\t^*)$. Using 
Eq.(\ref{4.22as}) we can perform the sum over $\t\in \TT^{(h)}_N(\t^*;n)$ and over ${\bf P}\in\PP_\t$, which gives
\bea && \sum_{\t\in\TT^{(h)}_N(\t^*;n)}\hskip-.4truecm 
2^{\th (h_+-N)}\hskip-.6truecm\sum_{\substack{{\bf P}\in\PP_\t:\\
|P_{v_0}|=|P_{v_0}^A|=m}}\hskip-.4truecm
\Big[\prod_{v\in V(\t)} 2^{(d_v-\tilde d_v)(h_v-h_{v'})}\Big]
\Big[\prod_{v\in E_n(\t) }(C|\l|)^{\max\{1,c|P_v^\psi|\}}\Big]\nonumber\\
&&\qquad \le C^{n+m}|\l|^n2^{\th (h_m - N)}\;,\label{4.23}\eea
where $h_m$ is the highest scale in $\tau^*$. Plugging Eq.(\ref{4.23}) back into the previous bound on $K'''$,
and summing over $n\ge 1$ and over $h<h_0^*$ gives:
\bea
&&|K'''(\xx_1,j_1;\cdots;\xx_m,j_m)|\le C^m|\l| m!\sum_{h_0^*\le N+1} 
2^{h_0^*(2-m)}
\sum_{\tau^* \in \mathcal{T}^*_{N,h_0^*;m}}2^{\th (h_m - N)}\cdot\nonumber\\
&&\cdot\Big[
\prod_{v \in V_{nt}(\tau^*)}
2^{2(s_v^*-1)h_v} e^{-c 2^{h_v}\delta_v} \Big]
\Big[\prod_{v \in V(\tau^*)} 2^{\tilde{d}_v(h_v-h_{v'})} \Big]\;.\label{4.24}
\eea
We now go from the sum over $\tau^*$ to an estimate in terms of a ``contracted tree", called $\t^{**}$,
which is obtained from $\t^*$ by removing the set $E$ of vertices which precede exactly one special endpoint, see 
Figure~\ref{EandTauFig}. 
\begin{figure}
\centering
\subfloat[ ]{
    \begin{tikzpicture}
        \foreach \x in {4,...,11}
        {
            \draw[very thin] (\x , 0) -- (\x , 6);
        }
        \node [vertex] (v0star) at (4,3) [label=-95:$v_0^*$] {};
        \draw (3,3) node {} -- (4,3) node {};
        \draw (v0star) --  (5,3.5) node [specialEP] (x) {};
        \draw (v0star) -- (5,2.5) node [vertex] {} -- (6,2) node [vertex] (v2) {} -- (7,1.5) node [specialEP] {};
        \draw (v2) \foreach \x in {7,8,9} {-- (\x,-1+\x/2) node [vertex] {}} -- (10,4) node (v3) [vertex] {};
        \draw (v3) -- (11,4.5) node [specialEP] {};
        \draw[rounded corners=0.3cm] (7,2.2) -- (10,3.7) -- (10.4,4) -- (10,4.3) -- (7,2.8) -- (6.6,2.5) -- cycle;
        \node at (10,-0.5) {$h_m$};
    \end{tikzpicture}
}
\hfill
\subfloat[ ]{
    \begin{tikzpicture}
        \foreach \x in {4,...,7}
        {
            \draw[very thin] (\x , 0) -- (\x , 6);
        }
        \node [vertex] (v0star) at (4,3) [label=-95:$v_0^*$] {};
        \draw (v0star) --  (5,3.5) node [specialEP] (x) {};
        \draw (3,3) node {} -- (4,3) node {};
        \draw (v0star) -- (5,2.5) node [vertex] {} -- (6,2) node [vertex] (v2) {} -- (7,1.5) node [specialEP] {};
        \draw (v2) -- (7,2.5) node [specialEP] {};
        \node at (6,-0.5) {$h_m^{**}$};
    \end{tikzpicture}
}
\caption{(a) The tree $\tau^*$ of Figure~\ref{specialTreeFig}, with the set $E$ highlighted.  (b) The corresponding 
``contracted tree" $\tau^{**}$.} \label{EandTauFig}
\end{figure}
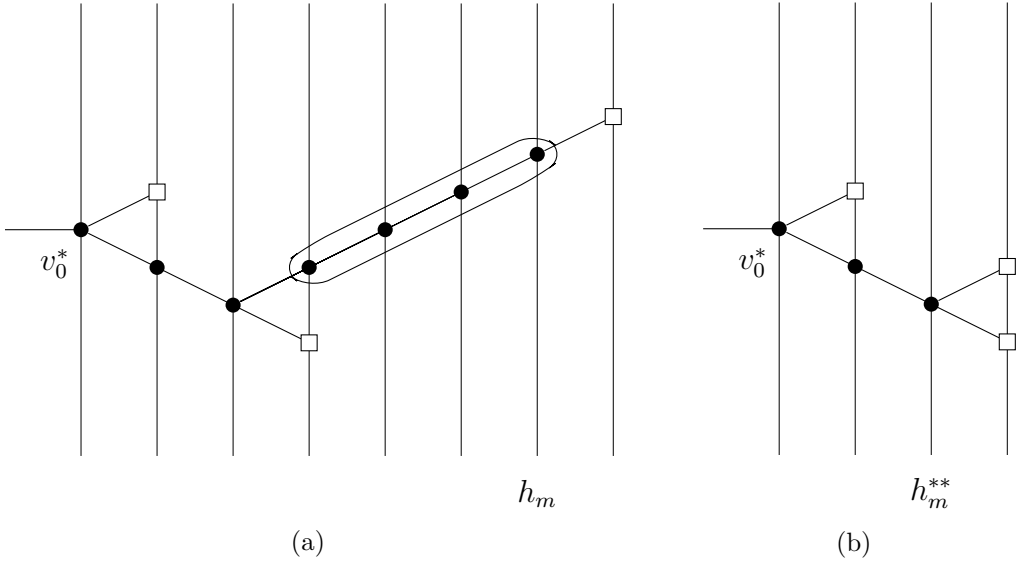
The sum over $\tau^*$ can then be expressed as the sum over $\tau^{**}$ (trees with only special endpoints, where 
each endpoint is attached to a branching point) and the sum over compatible $E$ (i.e. over the number of vertices 
inserted before each special endpoint); then Eq.(\ref{4.24}) becomes (calling $\TT^{**}_{N,h_0^*;m}$ the set of
contracted trees $\t^{**}$ with root on scale $h_0^*-1$ and $m$ endpoints)
\bea &&|K'''(\xx_1,j_1;\cdots;\xx_m,j_m)|\le C^m|\l| m!\sum_{h_0^*\le N+1} 
2^{h_0^*(2-m)}
\sum_{\tau^{**} \in \mathcal{T}^{**}_{N,h_0^*;m}}\cdot\nonumber\\
&&\qquad \cdot\Big[
\prod_{v \in V_{nt}(\tau^{**})}
2^{2(s_v^*-1)h_v} e^{-c 2^{h_v}\delta_v} \Big]
\Big[\prod_{v \in V(\tau^{**})} 2^{\tilde{d}_v(h_v-h_{v'})} \Big]\cdot\nonumber\\
&&\qquad \cdot\sum_{E\sim \t^{**}}
\Big[\prod_{v\in E}2^{-(\th +\phi)(h_v-h_{v'})}\Big]
2^{\th (h_m - N)}\;.\eea
We use the factors $2^{-\phi(h_v-h_{v'})}$ to perform the summation over $E$, and use the remaining factor to 
replace $h_m$ with the scale of the last branching point of $\tau^{**}$, which we denote by $h_m^{**}$, see 
Fig.\ref{EandTauFig}:
\be \sum_{E\sim \t^{**}}
\Big[\prod_{v\in E}2^{-(\th +\phi)(h_v-h_{v'})}\Big]
2^{\th (h_m - N)}\le C^m  2^{\th (h^{**}_m - N)}\;.\ee
This leaves us with the estimate:
\bea &&|K'''(\xx_1,j_1;\cdots;\xx_m,j_m)|\le C^m|\l| m!\sum_{h_0^*\le N+1} 
2^{h_0^*(2-m)}
\sum_{\tau^{**} \in \mathcal{T}^{**}_{N,h_0^*;m}}\cdot\nonumber\\
&&\cdot\Big[
\prod_{v \in V_{nt}(\tau^{**})}
2^{2(s_v^*-1)h_v} e^{-c 2^{h_v}\delta_v} \Big]
\Big[\prod_{v \in V(\tau^{**})} 2^{\tilde{d}_v(h_v-h_{v'})} \Big] 2^{\th (h^{**}_m - N)}\;,\label{4.27}
\eea
which is similar to Eq.(\ref{alpha_bound}). We would be tempted to proceed as we did after Eq.(\ref{alpha_bound}),
that is by ignoring all the constraints in the sum over the scale labels and reduce the r.h.s.\ of Eq.(\ref{4.27}) to a 
product of terms of the form $\sum_{h=-\io}^{\io}2^{\a_v h_v}e^{-c 2^h\d_v}$, which is fine provided that $\a_v>0$. 
However, here the exponent $\a_v$ associated with $v_0^*$ is equal to $2s_{v_0^*}^*-m$, which in general 
is negative. 

To rectify this we engage in a further rearrangement of the estimate, for which we introduce the following notation.  
Let $S_v^{*,2}$ denote, for $v \in V_{nt}(\tau^{**})$, the set of branching points following $v$ but not following 
any vertex $w\in V_{nt}(\t^{**})$ such that $w>v$
(see Figure~\ref{Svstar_fig}), 
$S_v^{*,1}$ the set of endpoints immediately following $v$, and $S_v^* = S_v^{*,1} \cup S_v^{*,2}$; also let $s_v^{*}=|S_v^{*}|$ 
and similarly for $s_v^{*,1}$ and $s_v^{*,2}$.  
\begin{figure}
\centering
\begin{tikzpicture}
    \foreach \x in {3,...,8}
    {
        \draw[very thin] (\x , 0) -- (\x , 7);
    }
    \draw (2,4) node 
     {} -- (3,4) node [vertex,label=-95:$v$] (v) {};
    \draw (v) -- (4,5) node [vertex] (v2) {} -- (5,6) node [specialEP] {};
    \draw (v2) -- (5,5) node [specialEP] {};
    \draw (v) -- (4,3) node [vertex] {} -- (5,2) node [vertex] (v1) {};
    \draw (v) -- (4,1) node [specialEP] {};
    \draw [rounded corners=0.3cm] (3.5,1) -- (4,1.5) -- (4.5,1) -- (4,0.5) -- cycle;
    \node at (3.5,0.5) {$S_v^{*,1}$};
    \draw (v1) -- (6,1) node [specialEP] {};
    \draw (v1) -- (6,3) node [vertex] {} -- (7,4) node [vertex] (v2) {} -- (8,5) node [specialEP] {};
    \draw (v2) -- (8,3) node [specialEP] {};
    \draw[rounded corners=0.3cm] (3.5,5) -- (4,5.5) -- (4.5,5) -- (5.5,2) -- (5,1.5) -- (4.5,2) -- (4.5,3.5) -- cycle;
    \node at (3.5,5.5) {$S^{*,2}_v$};
\end{tikzpicture}
\caption{Example of the sets $S_v^{*,1}$ and  $S_v^{*,2}$.} \label{Svstar_fig}
\end{figure}
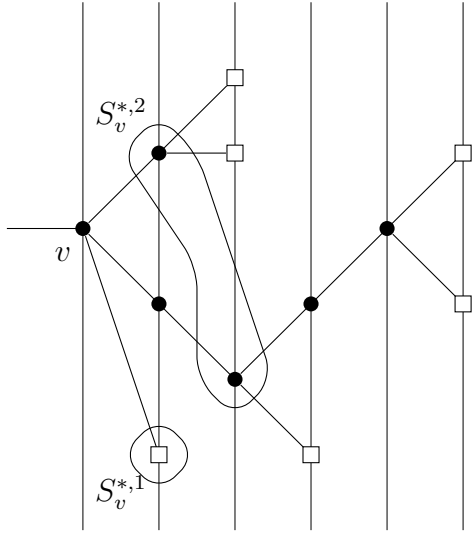
Since we have $\tilde d_v = 1+\e-m_v$ for all non-endpoint vertices, we have the following estimate  for each $v \in S_{v_0^*}^{*,2}$,
\be
\prod_{v_0^* < w \le v} 2^{\tilde d_w (h_w-h_{w'})} 
= 2^{(1+\e-m_v)(h_v-h_0^*)}
\ee
which since $\sum_{v \in S^{*,2}_{v_0^*}} (m_v -1-\e) = m-s_{v_0^*}^{*,1}-(1+\e) s_{v_0^*}^{*,2}=
m-(1+\e)s^*_{v_0^*}+\e s_{v_0^*}^{*,1}$ implies
\bea && 2^{(2s_{v_0^*}^*-m)h_0^*}\prod_{v \in S_{v_0^*}^{*,2}} \Big(2^{2(s_v^* -1)h_{v}} \prod_{v_0^* < w \le v} 
2^{\tilde d_w(h_w-h_{w'})} \Big) =\nonumber\\ &&= 2^{[(1-\e)s_{v_0^*}^*+\e s^{*,1}_{v_0^*}] h_0^*} 
\prod_{v \in S_{v_0^*}^{*,2}} 2^{(2s_v^*-m_{v}-1+\e)h_{v}}\eea
We now have a factor of $2^{(2s_v^*-m_v-1+\e)h_v}$ for each $v\in S_{v_0^*}^{*,2}$, and this can still 
produce a divergence.  For these vertices we have
\bea && 2^{(2s_{v}^*-m_v -1+\e)h_v}\prod_{u \in S_{v}^{*,2}} \Big(2^{2(s_u^* -1)h_{u}} 
\prod_{v < w \le u} 2^{\tilde{d}_w(h_w-h_{w'})} \Big) =\nonumber\\
&& = 2^{[(1-\e)(s_{v}^* -1) + \e s_v^{*,1}] h_v} \prod_{u \in S_{v}^{*,2}} 2^{(2s_u^*-m_u-1+\e)h_u}\eea
which we apply inductively to all branching vertices to obtain
\be|K'''(\xx_1,j_1;\cdots;\xx_m,j_m)|\le C^m|\l| m!\,2^{-\th N}\hskip-.4truecm\sum_{\substack{h_0^*\le N+1\\
 \tau^{**} \in \mathcal{T}^{**}_{N,h_0^*;m}}}\hskip-.2truecm\Big[\hskip-.2truecm\prod_{v\in V_{nt}(\t^{**})}
2^{\a_v h_v}e^{-c 2^{h_v}\d_v}\Big]\;,\label{4.31}\ee
with 
\be\a_v = \left\{
\begin{array}{ll}
(1-\e)s_{v}^*+\e s_v^{*,1}, & v=v_0^* \\
(1-\e)(s_{v}^* -1) + \e s_v^{*,1}+\th, & v=v_m^{**} \\
(1-\e)(s_{v}^* -1) + \e s_v^{*,1}, & \textup{otherwise}
\end{array}
\right.\ee
Note that we then have $\alpha_v \ge 1/2$ and $\sum \alpha_v = m+\theta$.
We can now sum this in the same way as explained after Eq.\eqref{alpha_bound}; noting that  
$\a_{v_0^*} \ge 2-2\e$, this gives
\be|K'''(\xx_1,j_1;\cdots;\xx_m,j_m)|\le C^m|\l| m!\,\frac1{\d^m}\Big(\frac{a}{\d}\Big)^\th \Big(\frac{\d}{D}\Big)^{2-2\e}\;,\label{4.34}\ee
which is valid for any $\th\in(0,1)$ and any $\e \in (0,1/2)$, but that the ``constants'' depend on $\th,\e$, so we cannot take $\th\to 1$ or $\e \to 0$.

The contributions to the $m$ points energy correlations from trees with $n\ge 1$ normal endpoints and 
$m_0<m$ special endpoints can be estimated via an analogous procedure and we are lead again to a bound 
like Eq.(\ref{4.34}). We do not belabor the details of this computation here. Therefore, the proofs of Theorems \ref{mainprop} and \ref{thm2} are complete. \\

{\bf Acknowledgements.} We acknowledge financial support from the
ERC Starting Grant CoMBoS-239694.


\end{document}